\documentclass[11pt]{article}

\usepackage{xcolor}
\usepackage{amsfonts,amsmath,amssymb,mathtools,amsthm}
\usepackage{graphicx}

\IfFileExists{stmaryrd.sty}{\usepackage{stmaryrd}}{%
  \newcommand{\llbracket}{\mathopen{[\mkern-3mu[}}%
  \newcommand{\rrbracket}{\mathclose{]\mkern-3mu]}}%
}

\usepackage{thmtools}
\DeclareGraphicsExtensions{.eps}

\usepackage{float}
\usepackage{geometry}
\allowdisplaybreaks

\renewcommand{\theequation}{\thesection.\arabic{equation}}

\newcommand\encadremath[1]{\vbox{\hrule\hbox{\vrule\kern8pt
\vbox{\kern8pt \hbox{$\displaystyle #1$}\kern8pt}
\kern8pt\vrule}\hrule}}
\def\enca#1{\vbox{\hrule\hbox{
\vrule\kern8pt\vbox{\kern8pt \hbox{$\displaystyle #1$}
\kern8pt} \kern8pt\vrule}\hrule}}

\newcommand\framefig[1]{
\begin{figure}[bth]
\hrule\hbox{\vrule\kern8pt
\vbox{\kern8pt \vbox{
\begin{center}
{#1}
\end{center}
}\kern8pt}
\kern8pt\vrule}\hrule
\end{figure}
}

\newcommand\figureframex[3]{
\begin{figure}[bth]
\hrule\hbox{\vrule\kern8pt
\vbox{\kern8pt \vbox{
\begin{center}
{\mbox{\epsfxsize=#1.truecm\epsfbox{#2}}}
\end{center}
\caption{#3}
}\kern8pt}
\kern8pt\vrule}\hrule
\end{figure}
}
\newcommand\figureframey[3]{
\begin{figure}[bth]
\hrule\hbox{\vrule\kern8pt
\vbox{\kern8pt \vbox{
\begin{center}
{\mbox{\epsfysize=#1.truecm\epsfbox{#2}}}
\end{center}
\caption{#3}
}\kern8pt}
\kern8pt\vrule}\hrule
\end{figure}
}

\renewcommand{\thesection}{\arabic{section}}
\renewcommand{\theequation}{\arabic{section}-\arabic{equation}}
\makeatletter
\@addtoreset{equation}{section}
\makeatother
\newtheorem{theorem}{Theorem}[section]
\newtheorem{conjecture}{Conjecture}[section]

\newtheorem{proposition}{Proposition}[section]
\newtheorem{lemma}{Lemma}[section]
\newtheorem{corollary}{Corollary}[section]

\theoremstyle{definition}
\newtheorem{remark}{Remark}[section]
\newtheorem{definition}{Definition}[section]

\def\br{\begin{remark}\rm\small}
\def\er{\end{remark}}
\def\bt{\begin{theorem}}
\def\et{\end{theorem}}
\def\bd{\begin{definition}}
\def\ed{\end{definition}}
\def\bp{\begin{proposition}}
\def\ep{\end{proposition}}
\def\bl{\begin{lemma}}
\def\el{\end{lemma}}
\def\bc{\begin{corollary}}
\def\ec{\end{corollary}}
\def\beaq{\begin{eqnarray}}
\def\eeaq{\end{eqnarray}}

\theoremstyle{definition}

\newcommand{\be}{\begin{equation}}
\newcommand{\ee}{\end{equation}}
\newcommand{\beq}{\begin{equation}}
\newcommand{\eeq}{\end{equation}}
\newcommand{\bea}{\begin{eqnarray}}
\newcommand{\eea}{\end{eqnarray}}
\newcommand{\beqq}{\begin{equation*}}
\newcommand{\eeqq}{\end{equation*}}
\newcommand{\beaa}{\begin{eqnarray*}}
\newcommand{\eeaa}{\end{eqnarray*}}

\newcommand{\Tr}{{\operatorname {Tr}}}

\newcommand{\diag}{{\operatorname{diag}}}

\newcommand{\om}{\omega}

\newcommand{\x}{{\rm x}}
\newcommand{\y}{{\rm y}}

\newcommand{\td}{\tilde}

\newcommand\blfootnote[1]{%
  \begingroup
  \renewcommand\thefootnote{}\footnote{#1}%
  \addtocounter{footnote}{-1}%
  \endgroup
}

\newcommand{\Res}{\mathop{\,\rm Res\,}}
\usepackage[pdftex]{pict2e}
\usepackage[pdftex]{hyperref}
\hypersetup{colorlinks,urlcolor=magenta,citecolor=red,linkcolor=blue,filecolor=black}
\usepackage{tikz}
\usetikzlibrary{positioning}
\usetikzlibrary{shapes,arrows}
\usepackage{caption}
\usepackage{circuitikz}
\usetikzlibrary{calc,trees,positioning,arrows,chains,shapes.geometric,%
    decorations.pathreplacing,decorations.pathmorphing,shapes,%
    matrix,shapes.symbols}

\tikzset{
>=stealth',
  punktchain/.style={
    rectangle, 
    draw=black, very thick,
    text width=10em, 
    minimum height=3em, 
    text centered, 
    on chain},
  line/.style={draw, very thick, <-},
  element/.style={
    tape,
    top color=white,
    bottom color=blue!50!black!60!,
    minimum width=8em,
    draw=blue!40!black!90, very thick,
    text width=10em, 
    minimum height=3.5em, 
    text centered, 
    on chain},
  every join/.style={->,very thick,shorten >=1pt},
  decoration={brace},
  tuborg/.style={decorate},
  tubnode/.style={midway, right=2pt},
}
\usepackage{lscape}
\usepackage{pdflscape}

\title{\bf{Explicit Hamiltonian representations of meromorphic connections and duality from different perspectives: a case study}}

\date{\vspace{-5ex}}
 
\author{$_{1}$Mohamad Alameddine\footnote{Universit\'{e} Jean Monnet Saint-\'{E}tienne, CNRS, Institut Camille Jordan UMR 5208, F-42023, Saint-\'{E}tienne, France.} \footnote{Section de Math\'ematiques, Universit\'e de Gen\`eve, Geneva, Switzerland.} \,\,,
$_{2}$Olivier Marchal\footnote{Universit\'{e} Jean Monnet Saint-\'{E}tienne, CNRS, Institut Camille Jordan UMR 5208, Institut Universitaire de France, F-42023, Saint-\'{E}tienne, France}
}

\begin{document}

\maketitle

\vspace{0.5cm}

\begin{abstract}
In this article, we present an explicit study of $\hbar$-deformed meromorphic connections in $\mathfrak{gl}_3(\mathbb{C})$ with an unramified irregular pole at infinity of order $r_\infty=3$ and their spectral dual corresponding to the $\mathfrak{gl}_2(\mathbb{C})$ Painlev\'{e} IV Lax pair. Using the apparent singularities and their dual partners on the spectral curves as Darboux coordinates, we obtain the Hamiltonian evolutions, the reduction of these evolutions to a single non-trivial direction, the Jimbo-Miwa-Ueno tau-functions, the fundamental symplectic two-forms and the associated Hermitian matrix models on both sides. We then prove that the spectral duality connecting both sides extends to all these aspects, providing an explicit illustration of the generalized Harnad duality. We finally propose a conjecture identifying the Jimbo-Miwa-Ueno differential with the $\hbar=0$ evaluation of the Hamiltonian differential in these Darboux coordinates that could provide insights into the geometric interpretation of the $\hbar$ formal parameter. As a byproduct we also obtain a rank $3$ Lax pair for the Painlev\'{e} IV equation.
\blfootnote{\textit{Email Addresses:}$_{1}$\textsf{Mohamad.Alameddine@unige.ch}, 
$_{2}$\textsf{olivier.marchal@univ-st-etienne.fr}}
\end{abstract}



\tableofcontents

\newpage

\section{Introduction}
The idea of isomonodromic deformations of linear differential equations dates back to B. Riemann in 1857 and the theory for Fuchsian systems was developed by L. Schlesinger \cite{schlesinger1912klasse} and R. Fuchs \cite{fuchs1907lineare}. Nowadays, this problem is known as the Fuchsian case and the equations emerging from it are known as the Schlesinger equations. These deformations became an extensive field of study when the sixth Painlev\'{e} equation was derived from a Fuchsian system and the link with the other Painlev\'{e} equations and the Painlev\'{e} property inspired a huge community \cite{Fuchs,Gambier,Garnier,Painleve,Picard,schlesinger1912klasse}. The topic was then pursued by R. Garnier and K. Okamoto \cite{okamoto1979deformation,Okamoto1986Iso,Okamoto1986} who studied Garnier systems in scalar form or Schlesinger systems \cite{schlesinger1912klasse} in matrix form. Ultimately, it was Garnier who extended the results to all other Painlev\'{e} equations, obtaining them as completely integrable systems. The Hamiltonian formulation of the Painlev\'{e} equations was achieved by A.J. Malmquist \cite{Malmquist1922} while the relations with isomonodromic deformations of linear ordinary differential equations with irregular singularities were given in \cite{Okamoto1980}. The irregular case consisting of deformations of meromorphic connections defined on the trivial vector bundle of arbitrary rank above the Riemann sphere with an arbitrary pole structure still presents an active field of study, in particular, the interest in the general case was revived by the Japanese school of Jimbo, Miwa and Ueno (JMU) \cite{JimboMiwa,JimboMiwaUeno} who considered arbitrary pole structure on arbitrary rank bundles and built, using the apparent singularities as Darboux coordinates, a vast family of nonlinear differential equations resulting from the isomonodromic deformations known as the largest to have the  ``Painlev\'{e} property". There are currently many ways to tackle the isomonodromic deformation problem and its link with the \textit{Stokes phenomenon} and the moduli spaces of meromorphic connections appearing in various areas of mathematical physics and geometry offers a variety of perspectives. Indeed, one may study the symplectic structure \cite{AtiyahBott,FockRosly99,Goldman84} of moduli spaces of flat meromorphic connections \cite{BoalchThesis,Boalch2001,Boalch2012,Boalch2022}. Another recent perspective developed by the Montr\'{e}al school, led by J. Harnad, J. Hurtubise and M. Bertola, relies on the study of moment maps to central extensions of loop algebras \cite{HarnadHurtubise,Harnad_1994} whose goal is to extend isospectral deformations to isomonodromic deformations using an appropriate set of Darboux coordinates. In fact, both approaches have been recently linked in the $\mathfrak{gl}_2$ setting by an explicit and non-trivial time-dependent change of Darboux coordinates \cite{marchal2023isomonodromic} relating the isomonodromic Hamiltonians to the isospectral ones. At the same time in \cite{MarchalAlameddineP1Hierarchy2023,marchal2023hamiltonian}, the authors used the apparent singularities as Darboux coordinates and a gauge transformation to the oper gauge to explicitly solve the compatibility equations and obtain an explicit Hamiltonian formulation resulting from the general isomonodromic deformations of meromorphic connections defined on trivial rank $2$ vector bundles over $\mathbb{P}^1$ with arbitrary poles of arbitrary order for the generic and the twisted case. This result builds an explicit bi-rational map between the isomonodromy connections of JMU/Boalch and the symplectic Ehresmann connection built from the Hamiltonian system. One of the main goals of this article is to extend this construction to a rank $3$ case hence showing that the construction of \cite{MarchalAlameddineP1Hierarchy2023,marchal2023hamiltonian} is applicable in the higher rank setting.\\

Another recent interesting aspect of isomonodromic deformations that we shall study in this article is the so-called Harnad's duality also known as spectral duality, generalized abstract Fourier-Laplace transform or $x-y$ symmetry. Indeed, a simple family of isomonodromy equations was introduced in \cite{JMMS} by the Japanese school arising as equations for correlation functions when studying the Schr\"{o}dinger equation.  The particular feature of these equations is the presence of two sets of deformation parameters, namely, the positions of the poles and the irregular times. J. Harnad \cite{Harnad_1994} showed that one may swap the two sets while preserving the isomonodromy equations. At the level of the spectral curves, this swap is equivalent to the exchange $x \leftrightarrow y$ on the spectral curve but the extension of this symmetry to the Hamiltonian structure is very non-trivial. This unexpected observation led to the terminology Harnad's duality and paved the way to many works to understand and generalize this duality \cite{HarnadHurtubise,Balser,Bertola2001DualityBP,Bertola2003,HarnadIts,Luu2019} to other examples. Of particular interest are the generalization by N.M.J. Woodhouse \cite{woodhouse2007duality} and the understanding of P. Boalch \cite{Boalch2012} that extended this duality to a continuous group action by considering a new class of isomonodromy equations and showing that they admit Kac–Moody Weyl group symmetries. Another fruitful application of this duality lies in the construction of the additive analogue of the \textit{middle convolution} \cite{BibiloMiddle,HaraokaMiddleConv,Katz,YamakawaMiddleConvolutions}, an operator that acts on the local systems on the punctured Riemann sphere. When considering the middle convolution, the duality is manifested by the \textit{Fourier-Laplace transform} and the link was established by G. Sanguinetti and N.M.J. Woodhouse \cite{Sanguinetti_2004}. More generally, the additive middle convolution is used to construct Weyl group symmetries of the moduli spaces of meromorphic connections on trivial bundles over the Riemann sphere \cite{Hiroe_2017}. Under some assumptions, D. Yamakawa proved that this operator preserves the isomonodromy properties of admissible meromorphic connections \cite{yamakawa2014fourierlaplace}. \\

At the level of applications, isomonodromy systems have appeared now in various problems in mathematical physics. For instance, in topological string theory, certain generating functions of Gromov-Witten invariants are shown to correspond to tau-functions of integrable hierarchies, which can be understood through isomonodromy systems. Another topic is the relation of integrable systems with Hermitian matrix models \cite{MBertola_2003,Bertola:2004ws,BertolaKorotkin2021,Morozov1999,MulaseMM}. Closely related to our work is the link between Hermitian matrix models and the \textit{topological recursion} \cite{EO07}, a powerful tool that has many applications in enumerative geometry. In fact, one may think of the topological recursion as a mechanism that associates to a classical spectral curve, an infinite family of differential forms computed recursively. These differentials enjoy various interesting properties and can be regrouped to quantize the classical spectral curve hence recovering the Lax system from the simple knowledge of the classical spectral curve \cite{Quantization_2021}. Although TR is a relatively recent tool, the $x-y$ symmetry has known a very rich development due to its various applications. For example, this transformation interchanges the potentials $V_1$ and $V_2$ when applied to Hermitian two-matrix models \cite{CEO06,EO2MM}. It is also used in free probability theory \cite{borot2021topological} and helps in the computation of intersection numbers on the moduli spaces of complex curves. Finally, it also has significant implications for topological string theory and mirror symmetry. Let us also mention that considerable work has been done in the past few years to relate the correlators of both sides of the $x-y$ swap \cite{alexandrov2022universal,Hock_2023b,Hock_2023a}, along with some properties preserved by this symmetry such as the KP integrability and the spectral duality \cite{alexandrov2023kp,alexandrov2024symplectic}.\\

In this article, we consider the duality from different perspectives in a non-trivial example. Our starting point is the moduli space of generic meromorphic connections having a pole at infinity of order $r_\infty=3$ defined on the trivial rank $3$ vector bundle over the Riemann sphere and its general ($\hbar$-deformed) isomonodromic deformations (\autoref{DefMeroGl3} and \autoref{GL3DefIsoDeform}). Our first result is to extend the methods applied in \cite{MarchalAlameddineP1Hierarchy2023,marchal2023hamiltonian} to obtain the explicit expressions for the Lax matrices (\autoref{LaxMatrixgl3} and \autoref{AppendixAuxiliaryGeneral}) and the explicit Hamiltonian evolutions (\autoref{Defs}). Our choice of Darboux coordinates is the apparent singularities and their dual partner on the spectral curve as originally considered by the Japanese school. The second step is then to realize that one may split the tangent space into a trivial subspace where the Hamiltonian evolutions are trivial (linear in the Darboux coordinates so that with an appropriate shift of the Darboux coordinates they may be set to zero) and a subspace of dimension $g=1$ (where $g$ is the genus of the spectral curve) where the evolutions are non-trivial giving rise to only one non-trivial Hamiltonian (\autoref{TheoNonTrivialEvolution}). A key aspect of this symplectic reduction is that it naturally reduces the fundamental symplectic two-form $\Omega$ \cite{marchal2023hamiltonian,Yamakawa2019FundamentalTwoForms} and the Jimbo-Miwa-Ueno differentials (\autoref{SymplecticReduction} and  \autoref{PropJMUDifferential}). We observe that the latter is expressed through the Hamiltonian differential formally evaluated at $\hbar=0$ (with fixed Darboux coordinates), with the stated sign and an exact time-dependent term and we propose it as a conjecture that would provide a geometric understanding of the formal parameter $\hbar$ (\autoref{ConjectureJMU}). The next step is then to obtain similar results on the dual side, which happens to be the $\mathfrak{gl}_2(\mathbb{C})$ Painlev\'{e} IV Lax pair studied in \cite{marchal2023hamiltonian} (i.e. meromorphic connections with one pole at infinity of order  $r_\infty=3$ and one finite pole of order $r_1=1$). Note that our results (\autoref{TheoHamP4} and \autoref{ThJMUP4}) may be seen as an extension of the results of \cite{marchal2023hamiltonian} since we do not assume the monodromies at each pole to have vanishing sum and because we propose the analysis of the JMU differential that was not performed in \cite{marchal2023hamiltonian}. The final step is then to study the duality between both systems at different levels. In particular, we prove the duality at the level of spectral curves (i.e. that the spectral curves are related by the $x-y$ swap) in \autoref{TheoDualitySpecCurves}, this is then extended to the Hamiltonian evolutions (\autoref{TheoDualHamiltonians}), the fundamental symplectic two-form (\autoref{TheoCorrespondenceFundamentalTwoForms}), the JMU differentials (\autoref{TheoDualityOmegaJMU}) and the partition functions of the corresponding matrix models when the spectral curves degenerate to genus $0$ (\autoref{TheoDualitynew}). Finally, at the level of the perturbative part of the partition functions, the duality relates the free energies of topological recursion, with the explicit correction $f$ at genus zero and time-independent constants at higher orders, when the classical spectral curves degenerate to genus $0$. When the spectral curves are of genus $1$, we formulate conjectures regarding the relations between the partition functions of the matrix models and the JMU tau-functions (\autoref{Prop2MM} and \autoref{PartitionFunctionP4}) and we conjecture duality relations with the non-perturbative TR-partition function (\autoref{ConjTRHMM}). The main results obtained in this article are summarized in \autoref{Fig2Diagram} providing a global picture of our work. 

\medskip

The paper is organized as follows. In \autoref{secGl3} we study the $\mathfrak{gl}_3(\mathbb{C})$ side. In particular, we obtain the Lax matrices (\autoref{SectionNormalized}), the Hamiltonian evolutions (\autoref{SecGL3Iso}) and their reduction, the fundamental symplectic two-form and its reduction (\autoref{SecGL3Reduction}), the JMU differential and its reduction (\autoref{SecGl3JMU}) and the associated Hermitian matrix model (\autoref{SecGl32MM}). In \autoref{SecP4JM} we review the $\mathfrak{gl}_2(\mathbb{C})$ side with focus on the Hamiltonian evolutions and fundamental two-form (\autoref{SecGl2Iso}) and their reduced form (\autoref{SecGl2Reduction}), the JMU differential (\autoref{SecGl2JMU}) and the associated Hermitian matrix model (\autoref{SecGl21MM}). Finally, in \autoref{SecDuality}, we study the duality at different levels and propose some conclusions and outlooks in  \autoref{SectionOutlooks}. In the end, the construction and results of the paper are summarized into the following figure:

 \begin{center}
\newgeometry{left=0.3cm,right=0.5cm,bottom=1cm,top=1cm}
\begin{landscape}
\thispagestyle{empty}
 \begin{figure}\centering
        \resizebox{\linewidth}{!}{\footnotesize
        \begin{tikzpicture}
        [
         roundnode/.style={circle, draw=white , ultra thin, minimum size=0mm},
squarednode/.style={rectangle, draw=black, very thick,  text width=10em, text centered, minimum size=5mm},
squarednodeprime/.style={rectangle, draw=black, thick,  text width=5em, text centered, minimum size=5mm},
line/.style ={draw, very thick, -latex', shorten >=0pt}
roundnodeprime/.style={point, draw=white , ultra thin, text width=3em, minimum size=0mm},
        ]

        \node[squarednode]      (maintopic)                              {Space of meromorphic connections in $\mathfrak{gl}_3(\mathbb{C})$: $\hat{F}_{\{\infty\},\mathbf{3}}$};
        \node[squarednode]      (Birkhoff)       [below= of maintopic] {Birkhoff factorisations or formal asymptotic expansions};
        \node[squarednode]        (JMU1)       [right=4cm of Birkhoff] {The JMU one-form $\omega_{JMU}(q,p,\mathbf{t})$};
        \node[roundnode] (1) [left=0.1cm of maintopic] {} ;
        \node[roundnode] (2) [above=14cm of 1] {} ;
        \node[squarednode] (Iso1) [right=0.1cm of 2] {Isomonodromic deformations $ \mathcal{L}_{\boldsymbol{\alpha}}  [ \Tilde{ \Psi}] = \Tilde{A}_{\boldsymbol{\alpha}} \Tilde{\Psi}$};
        \node[squarednode] (Oper1) [below=2cm of Iso1] {Oper gauge $L(\lambda), A_{\boldsymbol{\alpha}}(\lambda)$}; 
        \node[squarednode] (Lax1) [below=5cm of Iso1] {Lax matrix $\hbar \partial_\lambda \Tilde{\Psi} = \Tilde{L} \Tilde{\Psi}$};
        \node[roundnode] (5) [left=0.1cm of Lax1] {};
        \node[squarednode] (Partition1)   [above=of JMU1]   {Partition function $\ln Z_N^{(\text{2MM})}(\mathbf{t;\hbar})$};
        \node[squarednode] (2MM) [above=0.5cm of Partition1] {Classical spectral curve $\mathcal{S}_0$};
        \node[squarednode] (2M) [left=0.3cm of 2MM] {Two-matrix model};
        \node[roundnode] (7) [right=0.1cm of Iso1] {};
        \node[roundnode] (8) [right=0.1cm of Oper1] {};
        \node[roundnode] (9) [right=0.1cm of Lax1] {};
        \node[roundnode] (13) [below=1cm of 9] {};
        \node[squarednode] (Ham1) [right=2cm of 13] {Hamiltonian evolutions $\text{Ham}_{(\boldsymbol{\alpha})}(q,p,\mathbf{t};\hbar) \Leftrightarrow$ Hamiltonian differential $\overline{\omega}(q,p;\hbar)$};
        \node[squarednode] (Split1) [above= of Ham1]  {Split of the tangent space $(q,p,\mathbf{t}) \longleftrightarrow (\check{q},\check{p}, \mathbf{T},\tau )$};
        \node[squarednode] (2form1) [above=4cm of Split1]  {Fundamental symplectic two-form $\Omega$};
        \node[roundnode] (15) [above=0.1em of 2form1] {};
        \node[roundnode] (17) [left=0.8cm of 2form1] {};
        \node[roundnode] (18) [left=0.8cm of Ham1] {};
        \node[squarednodeprime] (reduceddiff1) [right=0.3cm of Split1] {Reduced Hamiltonian differential};
        \node[squarednode] (reducedHam1) [above=0.6cm of reduceddiff1] {Reduced Hamiltonian $\text{Ham}_{(\tau)} (\check{q},\check{p}, \tau;\hbar)$};
        \node[roundnode] (21) [above=0.2cm of reducedHam1] {};
        \node[roundnode] (23) [above=1.1cm of Split1] {};
        \node[squarednode] (sc1) [below=2.2cm of Ham1]  {Spectral curve $\mathcal{S}$ $\det (yI_3 - \Tilde{L}(\lambda))$=0};
        \node[roundnode] (25) [below=1.25cm of sc1] {};
        \node[roundnode] (27) [right=0.5cm of sc1] {};
        \node[roundnode] (29) [below=1.6cm of 27] {};
        \node[roundnode] (31) [right=2cm of Ham1] {};
        \node[roundnode] (33) [right=0.4cm of JMU1] {};
        \node[roundnode] (35) [below=1.5em of 31] {};
        \node[roundnode] (p1) [below=0.3em of 2form1] {};
         \node[squarednodeprime] (reduced2form1) [right=1.7cm of p1] {Reduced fundamental $2-$form};
        \node[roundnode] (p3) [above=4cm of 1] {};
        \node[roundnode] (f1) [left=0.5cm of p1] {};
         \node[roundnode] (s1) [above=0.3cm of 2form1] {};
       \node[roundnode] (s2) [left=3.5cm of s1] {};

        \node[squarednode] (gl2side) [right=18cm of maintopic] {Space of meromorphic connections in $\mathfrak{gl}_2(\mathbb{C})$: $\hat{F}_{\{\infty,X_1\},\{3,1\}}$};
        \node[squarednode] (Birkhoff2) [below=of gl2side] {Birkhoff factorisations or formal asymptotic expansions};
        \node[squarednode] (JMU2)    [left=4cm of Birkhoff2]   {The JMU one-form $\omega_{JMU}^{(\text{P4})}(Q,P,\mathbf{s})$};
        \node[roundnode] (3) [right=0.1cm of gl2side] {} ;
        \node[roundnode] (4) [above=14cm of 3] {} ;
        \node[squarednode] (Iso2) [left=0.1cm of 4] {Isomonodromic def. $\mathcal{L}_{\boldsymbol{\beta}}  [ \tilde{ \Psi}_{\text{P4}}] = \tilde{A}^{(\text{P4})}_{\boldsymbol{\beta}} \tilde{\Psi}_{\text{P4}}$};
        \node[squarednode] (Oper2) [below=2cm of Iso2] {Oper gauge $L_{\text{P4}}(\xi),A^{(\text{P4})}_{\boldsymbol{\beta}}(\xi)$}; 
        \node[squarednode] (Lax2) [below=5cm of Iso2] {Lax matrix $\hbar \partial_\xi \tilde{\Psi}_{\text{P4}} = \tilde{L}_{\text{P4}} \tilde{\Psi}_{\text{P4}}$};
        \node[roundnode] (6) [right=0.1cm of Lax2] {};
        \node[squarednode] (Partition2)   [above=of JMU2]   {Partition function $\ln Z_N^{(\text{1MM})}(\mathbf{s;\hbar})$};
        \node[squarednode] (1MM) [above=0.5cm of Partition2] {Classical spectral curve $\mathcal{S}_{\text{P4},0}$};
        \node[squarednode] (1M) [right=0.3cm of 1MM] {One-matrix model};
        \node[roundnode] (10) [left=0.1cm of Iso2] {};
        \node[roundnode] (11) [left=0.1cm of Oper2] {};
        \node[roundnode] (12) [left=0.1cm of Lax2] {};
        \node[roundnode] (14) [below=1cm of 12] {};
        \node[squarednode] (Ham2) [left=2cm of 14] {Hamiltonian evolutions $\text{Ham}^{(\text{P4})}_{(\boldsymbol{\beta})}(Q,P,\mathbf{s};\hbar) \Leftrightarrow $ Hamiltonian differential $\overline{\omega}^{(\text{P4})}(Q,P;\hbar)$};
         \node[squarednode] (Split2) [above= of Ham2]  {Split of the tangent space $(Q,P,\mathbf{s}) \longleftrightarrow (\check{Q},\check{P}, \mathbf{S},\tilde{X}_1 )$};
         \node[squarednode] (2form2) [above=4cm of Split2]  {Fundamental symplectic two-form $\Omega^{(\text{P4})}$};
         \node[roundnode] (16) [above=0.1em of 2form2] {};
         \node[roundnode] (19) [right=0.8cm of 2form2] {};
        \node[roundnode] (20) [right=0.8cm of Ham2] {};
        \node[squarednodeprime] (reduceddiff2) [left=0.3cm of Split2] {Reduced Hamiltonian differential};
        \node[squarednode] (reducedHam2) [above=0.6cm of reduceddiff2] {Reduced Hamiltonian $\text{Ham}^{(\text{P4})}_{(\tilde{X}_1)} (\check{Q},\check{P}, \tilde{X}_1;\hbar)$};
         \node[roundnode] (22) [above=0.2cm of reducedHam2] {};
         \node[roundnode] (24) [above=1.1cm of Split2] {};
          \node[squarednode] (sc2) [below=2.2cm of Ham2]  {Spectral curve $\mathcal{S}_{\text{P4}}$ $\det (YI_2 - \tilde{L}_{\text{P4}}(\xi))$=0};
          \node[roundnode] (26) [below=1.15cm of sc2] {};
          \node[roundnode] (28) [left=0.5cm of sc2] {};
          \node[roundnode] (30) [below=1.5cm of 28] {};
           \node[roundnode] (32) [left=2cm of Ham2] {};
           \node[roundnode] (34) [left=0.4cm of JMU2] {};
           \node[roundnode] (p2) [below=0.2em of 2form2] {};
           \node[squarednodeprime] (reduced2form2) [left=1.7cm of p2] {Reduced fundamental $2-$form};
           \node[roundnode] (p4) [above=4cm of 3] {};
           \node[roundnode] (f2) [right=0.5cm of p2] {};
                   \node[roundnode] (s3) [above=0.3cm of 2form2] {};
       \node[roundnode] (s4) [right=2.1cm of s3] {};
           
\draw[->] (maintopic) edge node[right] {Eq. \eqref{PsiTT}} (Birkhoff);
\draw[->] (gl2side) --    (Birkhoff2);
\draw[dashed, color=red,very thick][-] let
    \p1=(JMU1.east), \p2=(JMU2.west) in
    ($(\x1,\y1-1em)$) edge node[above]  { Th. \ref{TheoDualityOmegaJMU}} ($(\x2,\y2-1em)$)  ;

\draw[color=black][-] let
    \p1=(Birkhoff.east), \p2=(JMU1.west) in
    ($(\x1,\y1)$) edge node[above] {JMU residue formula  }($(\x2,\y2)$);

\draw[color=black][-] let
    \p1=(JMU2.east), \p2=(Birkhoff2.west) in
    ($(\x1,\y1)$) edge node[above] {JMU residue formula  } ($(\x2,\y2)$)  ;

\draw[color=black][-] let
    \p1=(Birkhoff.east), \p2=(JMU1.west) in
    ($(\x1,\y1)$) edge node[below]  {Eq. \eqref{JMUGL3Def} } ($(\x2,\y2)$)  ;

\draw[color=black][-] let
    \p1=(JMU2.east), \p2=(Birkhoff2.west) in
    ($(\x1,\y1)$) edge node[below]  {Eq. \eqref{JMUGL2Def} } ($(\x2,\y2)$)  ;

 \draw[color=black][-] let
    \p1=(1.center), \p2=(p3.center) in
    ($(\x1,\y1)$) edge node[right,text width=10em] {Choice of Darboux coordinates $(q,p)$} ($(\x2,\y2)$) ; 
 
 \draw[color=black][-] let
    \p1=(3.center), \p2=(p4.center) in
    ($(\x1,\y1)$) edge node[left,text width=10em, align=flush right] {Choice of Darboux coordinates $(Q,P)$} ($(\x2,\y2)$) ;
 
\draw[-] (1M) -- (1MM);
\draw[-] (2M) -- (2MM);
\draw[-] (p3.center) -- (5.center);
\draw[-] (2.center) -- (5.center);
\draw[-] (maintopic.west) -- (1.center);
\draw[-] (p4.center) -- (6.center);
\draw[-] (4.center) -- (6.center);
\draw[-] (gl2side.east) -- (3.center);

\draw[-] (2.center) -- (Iso1.west);
\draw[-] (4.center) -- (Iso2.east);
\draw[-] (5.center) -- (Lax1.west);
\draw[-] (6.center) -- (Lax2.east);
\draw[->] (Lax1) edge node[left] {$G(\lambda)$ } (Oper1);
\draw[->] (Lax1) edge node[right] {Prop. \ref{LaxMatrixgl3}} (Oper1);
\draw[->] (Lax2) edge node[left] {\cite{marchal2023isomonodromic}} (Oper2);
\draw[-] (JMU1) edge node[left] {Conj. \ref{Prop2MM}} (Partition1);
\draw[-] (JMU2) edge node[right] {Conj. \ref{PartitionFunctionP4}} (Partition2);

\draw[dashed, color=red,very thick][-] let
    \p1=(Partition1.east), \p2=(Partition2.west) in
    ($(\x1,\y1)$) edge node[above] {Conj. \ref{ConjTRHMM}} ($(\x2,\y2)$)  ;

\draw[-] (Partition1.west) -- (2M.south);
\draw[-] (Partition2.east) -- (1M.south);

\draw[-] (7.center) -- (Iso1.east);
\draw[-] (8.center) -- (Oper1.east);
\draw[-] (9.center) -- (Lax1.east);
\draw[-] (7.center) -- (13.center);

\draw[->] (13) edge node [below] {Th. \ref{Defs}} (Ham1)  ;
\draw[-] (13.east) -- (13.center);

\draw[-] (10.center) -- (Iso2.west);
\draw[-] (11.center) -- (Oper2.west);
\draw[-] (12.center) -- (Lax2.west);
\draw[-] (10.center) -- (14.center);
\draw[<-] (Ham2) edge node[below] {Prop. \ref{TheoHamP4}} (14) ;
\draw[-] (14.center) -- (14.west);

\draw[dashed,color=red,very thick][-] let
    \p1=(Ham1.east), \p2=(Ham2.west) in
    ($(\x1,\y1)$) edge node[above] {Th. \ref{TheoDualHamiltonians}} ($(\x2,\y2+0.05cm)$)  ;

\draw[-] (Ham1) -- (Split1);
\draw[-] (Ham2) -- (Split2);

\draw[densely dotted, color=green,very thick][-] let
    \p1=(reduced2form1.east), \p2=(reduced2form2.west) in
    ($(\x1,\y1)$) edge node[below]  {Th. \ref{TheoCorrespondenceFundamentalTwoForms}} ($(\x2,\y2)$) ;
    \draw[-]  (2form1)  -- (p1.center);
    \draw[-]  (2form2)  -- (p2.center);

\draw[dashed,color=red,very thick][-] let
    \p1=(2form1.east), \p2=(2form2.west) in
    ($(\x1,\y1+0.4em)$) edge node[above]  {Th. \ref{TheoCorrespondenceFundamentalTwoForms}} ($(\x2,\y2+0.5em)$) ;

\draw[<-] (2form1) -- (17.center);
\draw[<-] (2form2) -- (19.center);
\draw[-] (17.center) edge node[right] {Def. \ref{DefinitionFundamentalTwoForm}} (18.center);
\draw[-] (19.center) edge node[left] {Prop. \ref{PropReductionOmega2}} (20.center);

\draw[<-] (reduceddiff1) -- (Split1);
\draw[<-] (reduceddiff2) -- (Split2);

\draw[densely dotted,color=green,very thick][-] let
    \p1=(reduceddiff1.east), \p2=(reduceddiff2.west) in
    ($(\x1,\y1-0.15em)$) edge node[above] {Th. \ref{TheoDualHamiltonians}} ($(\x2,\y2)$)  ;

\draw[densely dotted,color=green,very thick][-] (reducedHam1) -- (21.center);
\draw[densely dotted,color=green,very thick][-] (reducedHam2) -- (22.center);
\draw[densely dotted,color=green,very thick][-] let
    \p1=(21.center), \p2=(22.center) in
    ($(\x1,\y1+0.1em)$) edge node[above,align=flush left] {Th. \ref{TheoDualReducedHamiltonian}\,\,\,\,} ($(\x2,\y2)$)  ;

\draw[-] (Split1) -- (23.center);
\draw[-] (Split2) -- (24.center);

\draw[color=black][-] let
    \p1=(23.center), \p2=(reducedHam1.west) in
    ($(\x1,\y1)$) edge node[above] {Th. \ref{TheoNonTrivialEvolution}} ($(\x2,\y2+0.1em)$); 
    \draw[->] let
    \p1=(23.center), \p2=(reducedHam1.west) in  ($(\x1,\y1)$) --  ($(\x2,\y2+0.1em)$) ;
\draw[color=black][-] let
    \p1=(reducedHam2.east), \p2=(24.center) in
    ($(\x1,\y1-0.1em)$) edge node[above, midway,align=flush right] {\,\,\,\,Prop. \ref{PropReducedP4}}  ($(\x2,\y2)$)  ;
    \draw[color=black][<-] let
    \p1=(reducedHam2.east), \p2=(24.center) in
    ($(\x1,\y1-0.1em)$) --   ($(\x2,\y2)$)  ;
\draw[dashed,color=red,very thick][-] let
    \p1=(sc1.east), \p2=(sc2.west) in
    ($(\x1,\y1+1em)$) -- ($(\x2,\y2+1.1em)$) node[above, midway] {Th. \ref{TheoDualitySpecCurves}};

\draw[->] (Lax1.south) .. controls +(down:35mm) and +(right:5mm) .. (sc1.north); 
\draw[->] (Lax2.south) .. controls +(down:35mm) and +(left:5mm) .. (sc2.north);
    \draw[->] (sc1) edge node[left] {$\hbar \to 0$} (25.center);
    \draw[->] (sc2) edge node[right] {$\hbar \to 0$} (26.center);
    \draw[<-] (sc1.east)  -- (27.center);
    \draw[<-] (sc2.west)  -- (28.center);
     \draw[-] (27.center)  edge node {$(q_0,p_0) \to (q,p)$} (29.center);
    \draw[-] (28.center)  edge node {$(Q_0,P_0) \to (Q,P)$} (30.center);
    \draw[color=blue][-] (JMU1.east)  -- (33.center);
    \draw[color=blue][-] (JMU2.west)  -- (34.center);
    
\draw[color=blue][-] let
    \p1=(31.center), \p2=(Ham1.east) in
    ($(\x1,\y1-1.5em)$) edge node[above] {$\hbar = 0$}  ($(\x2,\y2-1.5em)$) node[below, midway] {};
\draw[color=blue][-] let
    \p1=(31.center), \p2=(33.center) in
    ($(\x1,\y1-1.5em)$) -- ($(\x2,\y2)$) {};
\draw[color=blue][-] let
    \p1=(Ham2.west), \p2=(32.center) in
    ($(\x1,\y1-1.5em)$) edge node[above] {$\hbar = 0$} ($(\x2,\y2-1.5em)$) node[below, midway] {};
\draw[color=blue][-] let
    \p1=(32.center), \p2=(34.center) in
    ($(\x1,\y1-1.5em)$) -- ($(\x2,\y2)$) {};
\draw[color=blue][-] let
    \p1=(Ham2.west), \p2=(32.center) in
    ($(\x1,\y1-1.5em)$) edge node[below] {Prop. $\ref{ThJMUP4}$} ($(\x2,\y2-1.5em)$) node[below, midway] {};
    \draw[color=blue][-] let
    \p1=(31.center), \p2=(Ham1.east) in
    ($(\x1,\y1-1.5em)$) edge node[below] {Th. \ref{PropJMUDifferential}}  ($(\x2,\y2-1.5em)$) node[below, midway] {};

\draw[-] (p1.center) edge node[below] {Cor. \ref{SymplecticReduction}} (reduced2form1);
\draw[-] (p2.center) edge node[below] {Prop. \ref{PropReductionOmega2}} (reduced2form2);
\draw[->] (p1.center)-- (reduced2form1);
\draw[->] (p2.center) -- (reduced2form2);

\draw[densely dotted,color=green, very thick] (2MM) -- (1MM); 

\draw[color=black][-] let
    \p1=(f1.center), \p2=(p1.center) in
    ($(\x1,\y1)$) -- ($(\x2,\y2)$) ;
\draw[color=black][-] let
    \p1=(f2.center), \p2=(p2.center) in
    ($(\x1,\y1)$) -- ($(\x2,\y2)$) ;
\draw[color=black][-] let
    \p1=(f2.center), \p2=(Split2.north) in
    ($(\x1,\y1)$) -- ($(\x2+0.8cm,\y2)$) ;
\draw[color=black][-] let
    \p1=(f1.center), \p2=(Split1.north) in
    ($(\x1,\y1)$) -- ($(\x2-0.8cm,\y2)$) ;

\draw[dashed,color=white,very thick][-] let
    \p1=(2form1.east), \p2=(2form2.west) in
    ($(\x1,\y1+0.6cm)$) edge node[above]  {\textcolor{black}{\underline{\textbf{Section \ref{SecDuality}}}}} ($(\x2,\y2+0.6cm)$) ;
      
 \draw[color=black][-] let
    \p1=(s2.center), \p2=(s2.center) in
    ($(\x1,\y1)$) edge node[right,text width=20em] {\underline{\textbf{Section \ref{secGl3}}}} ($(\x2,\y2)$) ; 

 \draw[color=black][-] let
    \p1=(s4.center), \p2=(s4.center) in
    ($(\x1,\y1)$) edge node[right,text width=20em] {\underline{\textbf{Section \ref{SecP4JM}}}} ($(\x2,\y2)$) ; 

        \end{tikzpicture}   
}           
  \caption{\textit{Summary of the symplectic structures associated with the meromorphic connections on both sides and their relations by duality. Under the displayed conventions, the fundamental two-forms correspond with an overall minus sign.}}\label{Fig2Diagram}
\end{figure} 
\end{landscape}
\end{center}
\restoregeometry
\newpage

\section{Meromorphic connections in $\mathfrak{gl}_3(\mathbb{C})$ with an irregular pole at infinity}\label{secGl3}
Our starting point is the space of generic complete flat symplectic connections in $\mathfrak{gl}_3(\mathbb{C})$ with only one irregular unramified pole located at $\lambda=\infty$ of order $r_\infty=3$:

\begin{definition}[Space of connections] \label{DefMeroGl3}
We define
\beq
F_{\{\infty\},3} := \left\{\hat{L}(\lambda) = \sum_{k=0}^{1} \hat{L}^{[\infty,k]} \lambda^{k}
\,\, \text{ with }\,\, \{\hat{L}^{[\infty,k]}\} \in \left(\mathfrak{gl}_3\right)^2\right\}/\text{ GL}_3(\mathbb{C}) 
\eeq
where the quotient by $\text{GL}_3(\mathbb{C})$ indicates the conjugation action of the reductive group. We shall also define $\hat{F}_{\{\infty\},3} \subset F_{\{\infty\},3} $ the subset of ``generic" connections whose leading order at $\infty$ has distinct eigenvalues:
\begin{align}
\hat{F}_{\{\infty\},3} :=& \Big\{\hat{L}(\lambda) = \sum_{k=0}^{1} \hat{L}^{[\infty,k]} \lambda^{k}
\,\, \text{ with }\,\, \{\hat{L}^{[\infty,k]}\} \in \left(\mathfrak{gl}_3\right)^{2}\cr&
\,\,\,\, \text{ and }\hat{L}^{[\infty,1]} \text{ has distinct eigenvalues}
\Big\}/\text{ GL}_3(\mathbb{C})
\end{align}
Given $\hat{L}(\lambda)$ in an orbit of the space $F_{\{\infty\},3}$ , we define the ($\hbar$-dependent) horizontal sections $\hat{\Psi}(\lambda;\hbar)$ as solutions of the following linear differential system  
\beq \hbar\partial_\lambda \hat{\Psi}(\lambda)=\hat{L}(\lambda)\hat{\Psi}(\lambda)\eeq
for $\hbar\in \mathbb{C}^*$.
This system is seen as a trivialization of the connection over the trivial vector bundle over $\mathbb{P}^1$. 
\end{definition}

In this article, we shall restrict our construction to $r_\infty=3$ for the sake of simplicity. However, we believe that this construction can be extended to any value of $r_\infty$. One may also add different finite poles and study connections of higher pole orders defined on rank$-d$ vector bundles, this remains beyond the scope of the present article that can be seen as a proof of concept for these future works.

\begin{remark}In general, for a given pole order $r_\infty$, the space $\hat{F}_{\{\infty\},r_\infty}$ may be endowed with a \textit{Poisson structure} inherited from the Poisson structure of a corresponding loop algebra. This turns our space into a \textit{Poisson manifold} of dimension
    \begin{align}
        \text{dim } \hat{F}_{\{\infty\},r_\infty} = 9 r_\infty - 17
    \end{align}
    This dimension is computed by exerting the gauge action and observing the remaining free entries of the Lax matrix.
\end{remark}

\begin{remark}The parameter $\hbar$ is not necessary to derive the Hamiltonian structure of isomonodromic deformations associated with $\hat{F}_{\{\infty\},3}$ (in which case one can take it equal to $\hbar=1$). However, it plays an important role in the duality studied in this paper and in relation with the quantization of classical spectral curves using the topological recursion. We also recall that the formal parameter $\hbar$ can be introduced from the standard setup corresponding to $\hbar=1$ by a specific rescaling of the irregular times, monodromies and wave matrices as described in \cite{marchal2023hamiltonian} (See \autoref{Remark23} for the present case). Apart from places where it is important for the discussion, we shall often omit $\hbar$ in the set of parameters of the various functions defined in this article.
\end{remark}

Since the leading coefficient has distinct eigenvalues, formal diagonalization at infinity gives a formal normal solution (Turrittin-Levelt form) \cite{Birkhoff,Wasowbook}. More precisely, there is an invertible formal gauge
\beq G_\infty(\lambda)=\sum_{k=0}^{\infty}G_{\infty,k}\lambda^{-k}\in\mathrm{GL}_3(\mathbb{C}[[\lambda^{-1}]])\eeq
such that $\Psi_\infty(\lambda):=G_\infty(\lambda)\hat{\Psi}(\lambda)$ has the formal expression
\begin{align}\label{PsiTT} &\Psi_\infty(\lambda)= \Psi_\infty^{(\text{reg})}(\lambda)\exp\Big[\text{diag}\Bigg(\frac{1}{\hbar}\left(\frac{t_{\infty^{(1)},2}}{2}\lambda^2+t_{\infty^{(1)},1}\lambda+t_{\infty^{(1)},0}\ln \lambda\right),\cr&
\frac{1}{\hbar}\left(\frac{t_{\infty^{(2)},2}}{2}\lambda^2+t_{\infty^{(2)},1}\lambda+  t_{\infty^{(2)},0}\ln \lambda\right), \frac{1}{\hbar}\left(\frac{t_{\infty^{(3)},2}}{2}\lambda^2+t_{\infty^{(3)},1}\lambda+  t_{\infty^{(3)},0}\ln \lambda\right) \Big)\Bigg]\end{align}
where $\Psi_\infty^{(\text{reg})}(\lambda)$ is an invertible formal series in $\lambda^{-1}$. These formal series need not converge; analytic fundamental solutions realize the corresponding asymptotics on suitable sectors. This Turrittin-Levelt fundamental form provides a Lax matrix $L_\infty(\lambda)$ satisfying 
\begin{align} \label{2.7}
 L_\infty(\lambda) := &\, G_\infty(\lambda)\hat{L}(\lambda) G_\infty^{-1}(\lambda) +\hbar\partial_\lambda (G_\infty(\lambda))G_\infty^{-1}(\lambda)  \\ \nonumber
 = &\, \frac{d}{ d\lambda}Q_\infty( z_\infty(\lambda)) + \Lambda_\infty \frac{ z'_\infty(\lambda)}{ z_\infty(\lambda)} + O(\lambda^{-2})
 \,\,\text{where}\,\, Q_\infty( X) := \sum_{k=1}^{2} \frac{Q_{\infty,k}}{ X^{ k}} 
\end{align} 
with $z_\infty(\lambda):=\lambda^{-1}$ the local coordinate at infinity and
\beq 
\forall \, k\in \llbracket 1,2\rrbracket\,:\, Q_{\infty,k} :=  \diag \left( \frac{t_{\infty^{(1)},k}}{k}, \frac{t_{\infty^{(2)},k}}{k},  \frac{t_{\infty^{(3)},k}}{k} \right) \,\,\text{ , }\, \Lambda_\infty := -\diag(t_{\infty^{(1)},0},t_{\infty^{(2)},0},t_{\infty^{(3)},0})
\eeq
This implies that the local connection $L_\infty$ at infinity has the form:
\footnotesize{\bea
    L_\infty(\lambda )&=&  \text{diag}\left(t_{\infty^{(1)},2}\lambda+t_{\infty^{(1)},1}+\frac{t_{\infty^{(1)},0}}{\lambda}, t_{\infty^{(2)},2}\lambda+t_{\infty^{(2)},1}+  \frac{t_{\infty^{(2)},0}}{\lambda},t_{\infty^{(3)},2}\lambda+t_{\infty^{(3)},1}+  \frac{t_{\infty^{(3)},0}}{\lambda}\right) \cr
    &&+O\left(\lambda^{-2}\right)
\eea}

\normalsize{The} parameters $\mathbf{t}:=\left(t_{\infty^{(1)},2},t_{\infty^{(2)},2},t_{\infty^{(3)},2},t_{\infty^{(1)},1},t_{\infty^{(2)},1},t_{\infty^{(3)},1}\right)$ are usually called ``irregular times'' while the residues $\mathbf{t_0}:=\left(t_{\infty^{(1)},0},t_{\infty^{(2)},0},t_{\infty^{(3)},0}\right)$ are the formal monodromy exponents (called ``monodromies'' below for brevity). The irregular times determine the irregular type, whereas the generalized monodromy data also include Stokes data; the latter are preserved by isomonodromic deformations \cite{Boalch2001}. The formal monodromy exponents satisfy:
\beq \label{Monodromiescondition}
t_{\infty^{(1)},0}+t_{\infty^{(2)},0}+t_{\infty^{(3)},0}=0\eeq

The assumption that the connection is generic, i.e. that the eigenvalues of the leading order are distinct implies that $t_{\infty^{(1)},2}\neq t_{\infty^{(2)},2}$, $t_{\infty^{(1)},2}\neq t_{\infty^{(3)},2}$ and $t_{\infty^{(2)},2}\neq t_{\infty^{(3)},2}$. 

Our strategy consists in building the explicit Hamiltonian system corresponding to the isomonodromic deformations following the strategy of \cite{marchal2023hamiltonian} and using the apparent singularities and their dual partner on the spectral curve as Darboux coordinates as initiated by the Japanese school for the Painlev\'{e} equations \cite{JimboMiwa,JimboMiwaUeno}. This approach heavily relies on the geometric knowledge at the pole, described by the set of irregular times that provides the natural space for isomonodromic deformations. In order to have explicit formula for the connection, we shall first choose a specific representative (this choice does not change the symplectic structure). Then, we shall also rewrite the meromorphic connection in its \textit{oper gauge} which is equivalent to the so-called ``quantum curve" or equivalently to the irregular Garnier system associated with the meromorphic connection. This gauge has many advantages, the first one being that it does not depend on the choice of representative in the orbit, and the most important of which is that the zero-curvature equation (also known as the compatibility equation or the isomonodromy equation) becomes much easier to handle since only $3$ entries of the Lax matrix remain to determine instead of nine.

\begin{remark}\label{Remark23} One may also introduce (or remove) the parameter $\hbar$ by a simple rescaling of the irregular times, monodromies and $\lambda$:
\bea\label{Rescaling} \lambda &\to& \hbar^{-1} \lambda,\cr
\forall \, k\in \llbracket 0,2\rrbracket\,:\, t_{\infty^{(i)},k} &\to& \hbar^{k-1}t_{\infty^{(i)},k}.
\eea  
\end{remark}

\subsection{Representative normalized at infinity}\label{SectionNormalized}
In order to have explicit formulas, we shall choose a specific representative in the orbit although it is well-known that the symplectic structure is independent of this choice. After fixing the ordering of the eigenvalues and restricting to the chart where the $(1,2)$ and $(1,3)$ entries of the subleading coefficient in an eigenbasis are nonzero, using the conjugation action and the fact that the connection has a diagonalizable leading order, it is straightforward to see that there exists a unique element in the orbit admitting a diagonal leading order and a subleading order satisfying $\hat{L}^{[\infty,0]}[1,2]=\hat{L}^{[\infty,0]}[1,3]=1$. We shall denote by $\td{L}(\lambda)\in \hat{F}_{\{\infty\},3}$ this representative normalized at infinity. Combining this with the Turrittin-Levelt fundamental form, we obtain
\beq \td{L}(\lambda)=\diag\left(t_{\infty^{(1)},2},t_{\infty^{(2)},2},t_{\infty^{(3)},2}\right)\lambda +\td{L}^{[\infty,0]} \eeq
with 
\begin{align}
    \td{L}^{[\infty,0]} = & \begin{pmatrix} \beta_{1,1}&1&1\\  \beta_{2,1}& \beta_{2,2}& \beta_{2,3}\\  \beta_{3,1}& \beta_{3,2}&  \beta_{3,3} \end{pmatrix}
\end{align}
In the rest of the article, we shall use the notation $\td{L}(\lambda)$ to indicate the unique representative in the orbit having the above normalization at infinity.

Using the local diagonalization at infinity: $G_\infty(\lambda)=I_3+G_1\lambda^{-1}+O\left(\lambda^{-2}\right)$, it is easy to observe that
\beq \beta_{1,1}=t_{\infty^{(1)},1}\, \,,\,\, \beta_{2,2}=t_{\infty^{(2)},1}\, \,,\,\,\beta_{3,3}=t_{\infty^{(3)},1}
\eeq
so that only four entries in the Lax matrix remain to determine.

\subsection{Isomonodromic deformations and Darboux coordinates}
Let us define the base $\mathbb{B}$ as the set of irregular times satisfying the condition
\begin{align}
\mathbb{B} :=& \bigg\{\mathbf{t}=\left(t_{\infty^{(1)},2},t_{\infty^{(2)},2},t_{\infty^{(3)},2},t_{\infty^{(1)},1},t_{\infty^{(2)},1},t_{\infty^{(3)},1}\right)\in \mathbb{C}^{6} \cr  & \,\,/\,\,   t_{\infty^{(1)},2}\neq t_{\infty^{(2)},2} \,\,,\,\, t_{\infty^{(1)},2}\neq t_{\infty^{(3)},2} \,\,,\,\,t_{\infty^{(2)},2}\neq t_{\infty^{(3)},2}     
\bigg\}.
\end{align}
The space $\mathbb{B}$ is the space of isomonodromic deformations. Indeed, we recall that isomonodromic deformations are the horizontal lifts of parameter changes in the bundle of meromorphic connections, with flat connection defined by monodromy preservation. In our case, it means that any vector field $\partial_t \in T_{\mathbf t} \mathbb{B}$ gives rise to a deformation of $\tilde{L}(\lambda)$ preserving its generalized monodromy data.\footnote{Note that in the case of poles located at positions different than $\infty$ one needs to consider, in addition to deformations relatively to the irregular times, the standard deformation relatively to the position of the poles.}\\
Indeed, defining 
\beq \hat{F}_{\{\infty\},3,\mathbf{t_0}}:= \left\{\hat{L}(\lambda)\in \hat{F}_{\{\infty\},3}\,/\, \hat{L}(\lambda) \text{ has monodromies } \mathbf{t_0}\right\} \subset \hat{F}_{\{\infty\},3}\eeq
we have that above a point $\mathbf{t} \in \mathbb{B}$ the fiber is given by $\hat{\mathcal{M}}_{\{\infty\},3,\mathbf{t},\mathbf{t}_0}$ corresponding to elements of $\hat{F}_{\{\infty\},3,\mathbf{t_0}}$ with irregular times $\mathbf{t}$. Moreover, let us recall \cite{HarnadHurtubise,BoalchThesis,Harnad_1994} that
\beq
\hat{\mathcal{M}}_{\{\infty\},3,\mathbf{t},\mathbf{t_0}} =\left\{\hat{L}(\lambda) \in \hat{F}_{\{\infty\},3}\,\,/\,\, \hat{L}(\lambda) \,\text{ has irregular times } \mathbf{t}  \,\text{ and monodromies } \mathbf{t_0}\right\}
\eeq
is a symplectic manifold of dimension $2g$, where $g$ is the genus of the spectral curve defined by the algebraic equation $\det(yI_3-\hat{L}(\lambda) ) = 0$, that can be equipped with spectral Darboux coordinates $\left(q_i,p_i\right)_{1\leq i\leq g}$. The symplectic space is the phase space of the evolution equations of the Darboux coordinates \cite{Boalch2001}. In our case since $r_\infty=3$, we have $g=1$ hence generating a non-trivial symplectic structure for the fiber and we shall denote $(q,p)$ the Darboux coordinates for simplicity. For generic values of the irregular times $\mathbf{t}$ and monodromies $\mathbf{t_0}$, we define the Darboux coordinates $(q,p)$ in a natural way.
Let $\tilde{L}(\lambda) \in \hat{\mathcal{M}}_{\{\infty\},3,\mathbf{t},\mathbf{t}_0}$ be the unique representative normalized at infinity as explained previously. This representative admits a unique gauge transformation $\Psi(\lambda)=G(\lambda) \td{\Psi}(\lambda)$ giving rise to a Lax matrix $L(\lambda):= G(\lambda)\td{L}(\lambda)G(\lambda)^{-1}+\hbar (\partial_\lambda G(\lambda))G(\lambda)^{-1}$ that takes an oper form i.e. $L(\lambda)$ is of the form
\begin{align} \label{oper}
    L(\lambda)=\begin{pmatrix}0&1&0\\ 0&0&1\\ L_{3,1}(\lambda) & L_{3,2}(\lambda)& L_{3,3}(\lambda) \end{pmatrix}
\end{align}
which is equivalent to the so-called ``quantum curve" for the first line of $\td{\Psi}(\lambda)$:
\beq \label{quantumcurve} \left[(\hbar\partial_\lambda)^3-L_{3,3}(\lambda)(\hbar\partial_\lambda)^2- L_{3,2}(\lambda)\hbar\partial_\lambda-L_{3,1}(\lambda)\right]\Psi_{1,j}(\lambda)=0 \,,\,\, \forall \, j\in \llbracket 1,3\rrbracket
\eeq
Note that the first line of the gauge matrix $G(\lambda)$ is necessarily $(1,0,0)$ since the gauge transformation preserves the first line of $\td{\Psi}(\lambda)$.

\begin{remark}\label{RemarkWeyl}One may instead use the second or third row of $\td{\Psi}(\lambda)$ as a cyclic vector. If $G_e[L]$ denotes the cyclic gauge constructed from the row $e$, and $S_{1,j}$ exchanges the first and $j$th basis vectors, then
\beq G_{e_j}[\td{L}]=G_{e_1}[S_{1,j}\td{L}S_{1,j}^{-1}]S_{1,j},\qquad j=2,3.\eeq
Thus the simultaneous permutation of the representative, its formal data and its cyclic vector preserves the underlying symplectic geometry. The cyclic gauges associated with different rows of a fixed representative need not have the same apparent singularities.
\end{remark}

In the process, the coefficients $(L_{3,j}(\lambda))_{1\leq j\leq 3}$ remain rational functions with a pole at infinity but also with apparent singularities (i.e. simple poles that are not singularities of the wave matrix). This suggests using the apparent singularities as natural Darboux coordinates. In our case, since $g=1$, there is only $1$ apparent singularity that we shall denote $q$. We complement it with its dual partner on the spectral curve by defining $p$ such that
 \beq p=\frac{1}{\hbar}\underset{\lambda \to q}{\Res} L_{3,2}(\lambda) +\left(t_{\infty^{(1)},2}+t_{\infty^{(2)},2}+t_{\infty^{(3)},2}\right)q+ t_{\infty^{(1)},1}+t_{\infty^{(2)},1}+t_{\infty^{(3)},1}\eeq
which is such that
\beq 
 \det(p\, I_3- \td{L}(q)) = 0.
\eeq
so that $(q,p)$ is a point on the spectral curve $S(\lambda,y):=\det(y\, I_3-\td{L}(\lambda))=0$.

\begin{remark}The apparent coordinate $q$ depends on the chosen cyclic row. Different rows of a fixed matrix generally give different zeros of the cyclic gauge determinant and hence different spectral coordinates. Throughout this article we use the first row.
\end{remark}

\begin{remark} Note that the introduction/removal of $\hbar$ is straightforward at the level of the apparent singularities since it corresponds to 
\beq (q,p) \to \left(\hbar^{-1}q,\hbar\, p\right)  
\eeq
\end{remark}

\subsection{Explicit expressions of the Lax matrices}
In this section, we provide the explicit gauge transformation to transform the Lax matrix $\td{L}(\lambda)$ normalized at infinity according to \autoref{SectionNormalized} into a companion-like (i.e. oper form) Lax matrix $L(\lambda)$. We also provide the explicit expressions for the matrices $L(\lambda)$ and $\td{L}(\lambda)$. We recall that we take a representative $\td{L}(\lambda)$ normalized at infinity and we perform the gauge transformation to obtain
\begin{align}
 L(\lambda) := & G(\lambda)\td{L}(\lambda)G^{-1}(\lambda) +\hbar\partial_\lambda (G(\lambda))G^{-1}(\lambda)
\end{align}
where $L(\lambda)$ has the oper form described in \eqref{oper}. Using the local diagonalization at infinity, it allows us to obtain the entries of the gauge matrix in terms of the entries of $\td{L}(\lambda)$ as well as the expression of the oper Lax matrices in terms of the Darboux coordinates $(q,p)$. We find:

\begin{proposition} [Lax matrix and gauge matrix]\label{LaxMatrixgl3}\sloppy{The gauge matrix $G(\lambda)$ such that $\Psi(\lambda)=G(\lambda) \td{\Psi}(\lambda)$ is given by
    \beq G(\lambda)=\begin{pmatrix} 1 & 0 & 0 \\ t_{\infty^{(1)},2} \lambda + t_{\infty^{(1)},1} & 1 & 1 \\ G_{3,1}(\lambda) &   G_{3,2}(\lambda) &   G_{3,3}(\lambda)
    \end{pmatrix}\eeq
    where the last line of the gauge matrix is given by
    \small{\begin{align}
        G_{3,1}(\lambda) =& \left(t_{\infty^{(1)},2}\lambda+t_{\infty^{(1)},1}\right)^2-p^2+ \left( \left( t_{\infty^{(2)},2}+t_{\infty^{(3)},2} \right)q+ t_{\infty^{(2)},1} +t_{\infty^{(3)},1}   \right)p\cr
        &- t_{\infty^{(2)},2}t_{\infty^{(3)},2}q^2-\left( t_{\infty^{(2)},2}t_{\infty^{(3)},1} +t_{\infty^{(3)},2}t_{\infty^{(2)},1}  \right)q\cr
        &+ t_{\infty^{(1)},2} t_{\infty^{(1)},0}+t_{\infty^{(2)},2} t_{\infty^{(2)},0}+t_{\infty^{(3)},2} t_{\infty^{(3)},0}- t_{\infty^{(2)},1} t_{\infty^{(3)},1}+\hbar t_{\infty^{(1)},2}
        \cr
        G_{3,2}(\lambda) =& \left( t_{\infty^{(1)},2} + t_{\infty^{(2)},2} \right) \lambda + t_{\infty^{(3)},2} q  -p +t_{\infty^{(1)},1}+t_{\infty^{(2)},1}+t_{\infty^{(3)},1} \cr
         G_{3,3}(\lambda)=&  \left( t_{\infty^{(1)},2} + t_{\infty^{(3)},2} \right) \lambda + t_{\infty^{(2)},2} q  -p +t_{\infty^{(1)},1}+t_{\infty^{(2)},1}+t_{\infty^{(3)},1}
    \end{align}}
\normalsize{The} Lax matrix normalized at infinity is given by:    
    \begin{align}
    \tilde{L}(\lambda) =  \diag\left(t_{\infty^{(1)},2},t_{\infty^{(2)},2},t_{\infty^{(3)},2}\right)\lambda +\td{L}^{[\infty,0]}
\end{align}
where $\td{L}^{[\infty,0]}$ has the following form
\begin{align}
    \td{L}^{[\infty,0]}=\begin{pmatrix} t_{\infty^{(1)},1}  & 1  & 1  \\ \left[\td{L}^{[\infty,0]}\right]_{2,1}
    & t_{\infty^{(2)},1} & -p+ t_{\infty^{(2)},2}q + t_{\infty^{(2)},1} \\  
   \left[\td{L}^{[\infty,0]}\right]_{3,1} &-p+  t_{\infty^{(3)},2}q + t_{\infty^{(3)},1}& t_{\infty^{(3)},1}  \end{pmatrix}
\end{align}
with
\begin{align}
\left[\td{L}^{[\infty,0]}\right]_{2,1}=&\frac{t_{\infty^{(1)},2} -t_{\infty^{(2)},2}}{t_{\infty^{(2)},2} - t_{\infty^{(3)},2}}  \bigg(p^2-\left((t_{\infty^{(2)},2}+t_{\infty^{(3)},2})q+t_{\infty^{(2)},1}+t_{\infty^{(3)},1}  \right)p+ t_{\infty^{(2)},2}t_{\infty^{(3)},2}q^2\cr
&+(t_{\infty^{(2)},2}t_{\infty^{(3)},1}+t_{\infty^{(3)},2}t_{\infty^{(2)},1})q+ t_{\infty^{(2)},1}t_{\infty^{(3)},1}-t_{\infty^{(2)},0}(t_{\infty^{(2)},2}-t_{\infty^{(3)},2})\bigg)\cr
\left[\td{L}^{[\infty,0]}\right]_{3,1}=& -\frac{t_{\infty^{(1)},2} -t_{\infty^{(3)},2}}{t_{\infty^{(2)},2} - t_{\infty^{(3)},2}}  \bigg(
p^2-\left((t_{\infty^{(2)},2}+t_{\infty^{(3)},2})q+t_{\infty^{(2)},1}+t_{\infty^{(3)},1}  \right)p+ t_{\infty^{(2)},2}t_{\infty^{(3)},2}q^2\cr
&+(t_{\infty^{(2)},2}t_{\infty^{(3)},1}+t_{\infty^{(3)},2}t_{\infty^{(2)},1})q+ t_{\infty^{(2)},1}t_{\infty^{(3)},1}+t_{\infty^{(3)},0}(t_{\infty^{(2)},2}-t_{\infty^{(3)},2})
\bigg)
\end{align}
The Lax matrix $L(\lambda)$ in the oper gauge is characterized by its non-trivial last line:
\begin{align}
 L_{3,3}(\lambda)=& \frac{\hbar}{\lambda-q} + P_1(\lambda)
 \cr
L_{3,2}(\lambda)=&\frac{\hbar\left( p - P_1(q) \right)}{\lambda-q}  - P_2(\lambda)  +\hbar t_{\infty^{(1)},2} \cr
L_{3,1}(\lambda)= & \frac{\hbar\left( p^2 - P_1(q) p + P_2(q) - \hbar t_{\infty^{(1)},2}  \right)}{\lambda-q}   + P_3(\lambda) + p^3  - P_1(q) p^2 
+ (P_2(q)-\hbar t_{\infty^{(1)},2}) p \cr&
- P_3(q)-\hbar(t_{\infty^{(2)},2}+t_{\infty^{(3)},2})t_{\infty^{(1)},2}(\lambda-q)
\end{align} 
where the polynomials $P_1(\lambda)$, $P_2(\lambda)$ and $P_3(\lambda)$ are independent of the Darboux coordinates and admit the following expressions 
\small{\begin{align}
P_1(\lambda):=&\left(t_{\infty^{(1)},2}+t_{\infty^{(2)},2}+t_{\infty^{(3)},2}\right)\lambda+ t_{\infty^{(1)},1}+t_{\infty^{(2)},1}+t_{\infty^{(3)},1}\cr  
P_2(\lambda) :=&\left(t_{\infty^{(1)},2}t_{\infty^{(2)},2}+t_{\infty^{(1)},2}t_{\infty^{(3)},2}+t_{\infty^{(2)},2}t_{\infty^{(3)},2}\right)\lambda^2\cr
&+ \left(t_{\infty^{(1)},2}(t_{\infty^{(2)},1}+t_{\infty^{(3)},1})+t_{\infty^{(2)},2}(t_{\infty^{(1)},1}+t_{\infty^{(3)},1})+t_{\infty^{(3)},2}(t_{\infty^{(1)},1}+t_{\infty^{(2)},1}) \right)\lambda\cr
&-\left(t_{\infty^{(1)},2}t_{\infty^{(1)},0}+t_{\infty^{(2)},2}t_{\infty^{(2)},0}+t_{\infty^{(3)},2}t_{\infty^{(3)},0}\right)+t_{\infty^{(1)},1}t_{\infty^{(2)},1}+t_{\infty^{(1)},1}t_{\infty^{(3)},1}+t_{\infty^{(2)},1}t_{\infty^{(3)},1}\cr
P_3(\lambda) :=&   t_{\infty^{(1)},2}t_{\infty^{(2)},2}t_{\infty^{(3)},2} \lambda^3  +\left( t_{\infty^{(1)},2} t_{\infty^{(2)},2} t_{\infty^{(3)},1} +t_{\infty^{(1)},2} t_{\infty^{(3)},2}t_{\infty^{(2)},1}   + t_{\infty^{(2)},2} t_{\infty^{(3)},2} t_{\infty^{(1)},1}   \right)\lambda^2\cr
    & +\bigg( t_{\infty^{(1)},2} t_{\infty^{(2)},2} t_{\infty^{(3)},0}+t_{\infty^{(1)},2} t_{\infty^{(3)},2} t_{\infty^{(2)},0}\cr
    &+t_{\infty^{(2)},2} t_{\infty^{(3)},2} t_{\infty^{(1)},0}
+ t_{\infty^{(1)},2} t_{\infty^{(2)},1} t_{\infty^{(3)},1}+t_{\infty^{(2)},2} t_{\infty^{(1)},1} t_{\infty^{(3)},1}+t_{\infty^{(3)},2} t_{\infty^{(1)},1} t_{\infty^{(2)},1}
\bigg)\lambda\cr
\end{align}}\normalsize{}}
\end{proposition}

\begin{proof}The proof for the gauge matrix $G(\lambda)$ follows from direct computations. The asymptotic expansion at infinity of the wave matrix $\Psi(\lambda)$ (whose second and third lines are just the first and second derivatives of the first line) together with our choice of Darboux coordinates determines the matrix $L(\lambda)$. Finally, the gauge transformation determines $\td{L}(\lambda)$ from $L(\lambda)$ and $G(\lambda)$.
\end{proof}

\begin{remark}\label{rem2.7}
The general form of the gauge matrix $G(\lambda)$ in arbitrary dimension was already derived in \cite{Quantization_2021}. Our result is compatible with this result keeping in mind that the normalization at infinity is different between both formalisms. Indeed, we have
\beq
G(\lambda)^{-1}=G_{\text{norm}}(\lambda) J(\lambda) \eeq
with $J(\lambda)$ derived in \cite{Quantization_2021} and given by
\beq J(\lambda):=\begin{pmatrix}1&0&0\\ 0&1&0\\
\frac{p^2-P_1(q)p+P_2(q)-\hbar t_{\infty^{(1)},2}}{\lambda-q}&  \frac{p-P_1(q)}{\lambda-q} & \frac{1}{\lambda-q}\end{pmatrix}
\eeq
while the matrix $G_{\text{norm}}(\lambda)$ stands for our different choice of normalization at infinity and is given by
\beq G_{\text{norm}}(\lambda):=\begin{pmatrix}1&0&0\\
\left[G_{\text{norm}}(\lambda)\right]_{2,1}& -\frac{t_{\infty^{(1)},2}+t_{\infty^{(3)},2}}{t_{\infty^{(2)},2}-t_{\infty^{(3)},2}} & \frac{1}{t_{\infty^{(2)},2}-t_{\infty^{(3)},2}} \\
\left[G_{\text{norm}}(\lambda)\right]_{3,1}&\frac{t_{\infty^{(1)},2}+t_{\infty^{(2)},2}}{t_{\infty^{(2)},2}-t_{\infty^{(3)},2}}& -\frac{1}{t_{\infty^{(2)},2}-t_{\infty^{(3)},2}}
\end{pmatrix}\eeq
with
\small{\begin{align}
    \left[G_{\text{norm}}(\lambda)\right]_{2,1}=&\frac{t_{\infty^{(1)},2}t_{\infty^{(3)},2}\lambda +\left( (t_{\infty^{(2)},2}+t_{\infty^{(3)},2})q-p+t_{\infty^{(2)},1}+t_{\infty^{(3)},1}\right)t_{\infty^{(1)},2}+t_{\infty^{(1)},1}t_{\infty^{(3)},2}}{t_{\infty^{(2)},2}-t_{\infty^{(3)},2}}\\
\left[G_{\text{norm}}(\lambda)\right]_{3,1}=&-\frac{t_{\infty^{(1)},2}t_{\infty^{(2)},2}\lambda +\left( (t_{\infty^{(2)},2}+t_{\infty^{(3)},2})q-p+t_{\infty^{(2)},1}+t_{\infty^{(3)},1}\right)t_{\infty^{(1)},2}+t_{\infty^{(1)},1}t_{\infty^{(2)},2}}{t_{\infty^{(2)},2}-t_{\infty^{(3)},2}} 
\end{align}}
\normalsize{}  
\end{remark}

\subsection{Monodromy-preserving deformations and solving the compatibility equation}\label{SecGL3Iso}
\subsubsection{General monodromy-preserving deformations}
As the name implies, monodromy-preserving deformations, or equivalently isomonodromic deformations, are deformations that preserve the monodromies of the system. The set of irregular times introduced in the previous sections is a natural set of deformation parameters in $F_{\{\infty\},3}$. Thus, we define deformations with respect to these parameters as follows
\begin{definition}[Space of isomonodromic deformations]\label{GL3DefIsoDeform}
    We define the following general deformation operators in $\hat{F}_{\{\infty\},3}$: 
    \begin{align}
    \label{GeneralDeformationsDefinition} \mathcal{L}_{\boldsymbol{\alpha}}=\hbar \sum_{i=1}^3\sum_{k=1}^{2} \alpha_{\infty^{(i)},k} \partial_{t_{\infty^{(i)},k}}
    \end{align}
where we define the vector $\boldsymbol{\alpha}\in \mathbb{C}^{6}$ by
\beq \boldsymbol{\alpha}= \sum_{i=1}^3\sum_{k=1}^{2} \alpha_{\infty^{(i)},k}\mathbf{e}_{\infty^{(i)},k}
\eeq
with $(\mathbf{e}_{\infty^{(i)},k})_{i,k}$ the canonical basis of $\mathbb{C}^6$. 
\end{definition}

Note that the coefficients of $\boldsymbol{\alpha}$ may also depend on the set of monodromies of the system, namely $\left(t_{\infty^{(i)},0} \right)_{1 \leq i \leq 3}$.

When considering isomonodromic deformations, we associate auxiliary matrices $\td{A}_{\boldsymbol{\alpha}}(\lambda)$ and $A_{\boldsymbol{\alpha}}(\lambda)$ (depending on the gauge) by 
\beq \label{Auxiliary}  \td{A}_{\boldsymbol{\alpha}}(\lambda):=\mathcal{L}_{\boldsymbol{\alpha}}[\td{\Psi}(\lambda)]\td{\Psi}(\lambda)^{-1}\,\,\,,\,\,\,
A_{\boldsymbol{\alpha}}(\lambda):=\mathcal{L}_{\boldsymbol{\alpha}}[\Psi(\lambda)]\Psi(\lambda)^{-1}
\eeq
In other words:
\bea \mathcal{L}_{\boldsymbol{\alpha}}[\td{\Psi}(\lambda)]=\td{A}_{\boldsymbol{\alpha}}(\lambda)\td{\Psi}(\lambda)\,\,,\,\, \mathcal{L}_{\boldsymbol{\alpha}}[\Psi(\lambda)]=A_{\boldsymbol{\alpha}}(\lambda)\Psi(\lambda)
\eea 

In particular, $\td{A}_{\boldsymbol{\alpha}}(\lambda)$ has entries that are rational functions of $\lambda$ with poles only at $\{ \infty \}$ while $A_{\boldsymbol{\alpha}}(\lambda)$ may also admit additional simple poles at the apparent singularities.  

The compatibility of the differential systems in $\lambda$ and $\mathcal{L}_{\boldsymbol{\alpha}}$ is equivalent to the compatibility equations or isomonodromy equations (also known as the zero-curvature equation):

\bea \label{CompatibilityEquation}\mathcal{L}_{\boldsymbol{\alpha}}[L(\lambda)]&=&\hbar\partial_\lambda A_{\boldsymbol{\alpha}}(\lambda)+\left[A_{\boldsymbol{\alpha}}(\lambda),L(\lambda)\right]\cr
\Leftrightarrow\, \mathcal{L}_{\boldsymbol{\alpha}}[\td{L}(\lambda)]&=&\hbar\partial_\lambda \td{A}_{\boldsymbol{\alpha}}(\lambda)+\left[\td{A}_{\boldsymbol{\alpha}}(\lambda),\td{L}(\lambda)\right].
\eea

The system (\eqref{Auxiliary}) allows us to extract information on the general form of the auxiliary matrices from the expansion of the wave matrix around the pole at infinity. This knowledge is then used and complemented by the zero-curvature equation, which encodes the flatness of the connections that we are investigating, to provide the evolutions of the Darboux coordinates. We shall recover that these evolutions are indeed Hamiltonian as proved in the general context in \cite{Boalch2001,Boalch2012,Yamakawa2017TauFA}.

\begin{remark}\label{EquivalenceAtdA}
    One should keep in mind that the Lax pairs $\left(\td{L}(\lambda),\td{A}_{\boldsymbol{\alpha}}(\lambda)\right)$ and $\left(L(\lambda),A_{\boldsymbol{\alpha}}(\lambda)\right)$ are equivalent, in the sense that they admit the same deformations and provide the same Hamiltonian system but are just expressed in two different gauges. In particular, one can recover $\td{A}_{\boldsymbol{\alpha}}(\lambda)$ from $A_{\boldsymbol{\alpha}}(\lambda)$ by the usual formula:
    \begin{align} \td{A}_{\boldsymbol{\alpha}}(\lambda)=&G(\lambda)^{-1}A_{\boldsymbol{\alpha}}(\lambda)G(\lambda)-G(\lambda)^{-1} \mathcal{L}_{\boldsymbol{\alpha}}[G(\lambda)]\cr\Leftrightarrow A_{\boldsymbol{\alpha}}(\lambda)=&G(\lambda)\td{A}_{\boldsymbol{\alpha}}(\lambda)G(\lambda)^{-1}+\mathcal{L}_{\boldsymbol{\alpha}}[G(\lambda)]G(\lambda)^{-1}
    \end{align}
\end{remark}

\subsubsection{Compatibility equations and Hamiltonian evolutions}
In this section, we consider the zero-curvature equation that provides the flows of the Darboux coordinates. There exist two equivalent ways to obtain the formulas. The first strategy is to use the local diagonalization at infinity to get that the auxiliary matrix $\td{A}_{\boldsymbol{\alpha}}(\lambda)$ is polynomial in $\lambda$ of the form
\beq \td{A}_{\boldsymbol{\alpha}}(\lambda)=\td{A}_{\boldsymbol{\alpha}}^{[\infty,2]}\lambda^2+\td{A}_{\boldsymbol{\alpha}}^{[\infty,1]}\lambda+\td{A}_{\boldsymbol{\alpha}}^{[\infty,0]}\eeq
and then use the compatibility equation 
\beq \mathcal{L}_{\boldsymbol{\alpha}}[\td{L}(\lambda)]=\hbar\partial_\lambda \td{A}_{\boldsymbol{\alpha}}(\lambda)+[\td{A}_{\boldsymbol{\alpha}}(\lambda),\td{L}(\lambda)]\eeq
together with the expression of $\td{L}(\lambda)$ of \autoref{LaxMatrixgl3} to determine entries of $\td{A}_{\boldsymbol{\alpha}}(\lambda)$ and the evolution of $(q,p)$ relatively to $\mathcal{L}_{\boldsymbol{\alpha}}$.

The second strategy is to study the compatibility equation in the oper gauge 
\beq \mathcal{L}_{\boldsymbol{\alpha}}[L(\lambda)]=\hbar\partial_\lambda A_{\boldsymbol{\alpha}}(\lambda)+\left[A_{\boldsymbol{\alpha}}(\lambda),L(\lambda)\right]\eeq
In particular, because the matrix $L(\lambda)$ is companion-like then one may obtain the second and third lines of $A_{\boldsymbol{\alpha}}(\lambda)$ in terms of its first line.
\begin{align}\label{SecondThirdLineA}
    [A_{\boldsymbol{\alpha}} (\lambda)]_{2,1} =& \hbar \partial_\lambda [A_{\boldsymbol{\alpha}}(\lambda)]_{1,1} + [A_{\boldsymbol{\alpha}}(\lambda)]_{1,3} L_{3,1}(\lambda) \cr
    [A_{\boldsymbol{\alpha}}(\lambda)]_{2,2} =& \hbar \partial_\lambda [A_{\boldsymbol{\alpha}}(\lambda)]_{1,2} +  [A_{\boldsymbol{\alpha}}(\lambda)]_{1,1} + [A_{\boldsymbol{\alpha}}(\lambda)]_{1,3} L_{3,2}(\lambda) \cr
     [A_{\boldsymbol{\alpha}}(\lambda)]_{2,3} =& \hbar \partial_\lambda [A_{\boldsymbol{\alpha}}(\lambda)]_{1,3} +  [A_{\boldsymbol{\alpha}}(\lambda)]_{1,2} + [A_{\boldsymbol{\alpha}}(\lambda)]_{1,3} L_{3,3}(\lambda)\cr
     [A_{\boldsymbol{\alpha}} (\lambda)]_{3,1} =& \hbar \partial_\lambda [A_{\boldsymbol{\alpha}}(\lambda)]_{2,1} + [A_{\boldsymbol{\alpha}}(\lambda)]_{2,3} L_{3,1}(\lambda) \cr
     [A_{\boldsymbol{\alpha}}(\lambda)]_{3,2} =& \hbar \partial_\lambda [A_{\boldsymbol{\alpha}}(\lambda)]_{2,2} +  [A_{\boldsymbol{\alpha}}(\lambda)]_{2,1} + [A_{\boldsymbol{\alpha}}(\lambda)]_{2,3} L_{3,2}(\lambda) \cr
     [A_{\boldsymbol{\alpha}}(\lambda)]_{3,3} =& \hbar \partial_\lambda [A_{\boldsymbol{\alpha}}(\lambda)]_{2,3} +  [A_{\boldsymbol{\alpha}}(\lambda)]_{2,2} + [A_{\boldsymbol{\alpha}}(\lambda)]_{2,3} L_{3,3}(\lambda) 
\end{align}
Thus, only the first line of $A_{\boldsymbol{\alpha}}(\lambda)$ remains to determine. These entries have simple pole at $\lambda=q$ (apparent singularity) and polynomial part of degree at most $3$. The remaining entries of the compatibility equation provide the evolution of our Darboux coordinates
\begin{align}
    \mathcal{L}_{\boldsymbol{\alpha}}[L_{3,1}(\lambda)] =& \hbar \partial_\lambda [A_{\boldsymbol{\alpha}}(\lambda)]_{3,1} +  [A_{\boldsymbol{\alpha}}(\lambda)]_{3,3} L_{3,1}(\lambda) - \cr & \left(  L_{3,1}(\lambda) [A_{\boldsymbol{\alpha}}(\lambda)]_{1,1} + L_{3,2}(\lambda) [A_{\boldsymbol{\alpha}}(\lambda)]_{2,1} +  L_{3,3}(\lambda) [A_{\boldsymbol{\alpha}}(\lambda)]_{3,1}    \right)  \cr
    \mathcal{L}_{\boldsymbol{\alpha}}[L_{3,2}(\lambda)] =& \hbar \partial_\lambda [A_{\boldsymbol{\alpha}}(\lambda)]_{3,2} + [A_{\boldsymbol{\alpha}}(\lambda)]_{3,1} + [A_{\boldsymbol{\alpha}}(\lambda)]_{3,3}  L_{3,2}(\lambda)  - \cr 
    & \left(  L_{3,1}(\lambda) [A_{\boldsymbol{\alpha}}(\lambda)]_{1,2} + L_{3,2}(\lambda) [A_{\boldsymbol{\alpha}}(\lambda)]_{2,2} +  L_{3,3}(\lambda) [A_{\boldsymbol{\alpha}}(\lambda)]_{3,2}  \right) \cr
    \mathcal{L}_{\boldsymbol{\alpha}}[L_{3,3}(\lambda)] =& \hbar \partial_\lambda [A_{\boldsymbol{\alpha}}(\lambda)]_{3,3} + [A_{\boldsymbol{\alpha}}(\lambda)]_{3,2} + [A_{\boldsymbol{\alpha}}(\lambda)]_{3,3}  L_{3,3}(\lambda) - \cr
    & \left(  L_{3,1}(\lambda) [A_{\boldsymbol{\alpha}}(\lambda)]_{1,3} + L_{3,2}(\lambda) [A_{\boldsymbol{\alpha}}(\lambda)]_{2,3} +  L_{3,3}(\lambda) [A_{\boldsymbol{\alpha}}(\lambda)]_{3,3}  \right) 
\end{align}

The expressions for the auxiliary matrices and Hamiltonian system for Darboux coordinates are rather long, but they may be easily summarized in specific directions that span the whole tangent space.

\begin{definition}\label{VectorsDeformations}We define the following vectors and their corresponding deformation operators
 \begin{itemize}
     \item The vectors $\mathbf{v}_{\infty,1}$ and $\mathbf{v}_{\infty,2}$ are defined by
\begin{align} \label{Defvinftyk} \mathcal{L}_{\mathbf{v}_{\infty,1}}:=&\hbar \left(\partial_{t_{\infty^{(1)},1}}+\partial_{t_{\infty^{(2)},1}}+\partial_{t_{\infty^{(3)},1}} \right)\cr
     \mathcal{L}_{\mathbf{v}_{\infty,2}}:=&\hbar \left(\partial_{t_{\infty^{(1)},2}}+\partial_{t_{\infty^{(2)},2}}+\partial_{t_{\infty^{(3)},2}} \right)
 \end{align}
     \item The vector $\mathbf{u}_{\infty,1}$:
      \beq \label{Defuinfty1} \mathcal{L}_{\mathbf{u}_{\infty,1}}:=\hbar \left( t_{\infty^{(1)},2}\partial_{t_{\infty^{(1)},1}} +t_{\infty^{(2)},2}  \partial_{t_{\infty^{(2)},1}} 
    +t_{\infty^{(3)},2}   \partial_{t_{\infty^{(3)},1}}\right)
\eeq
    \item The vector $\mathbf{u}_{\infty,2}$:
\begin{align} \label{Defuinfty2} \mathcal{L}_{\mathbf{u}_{\infty,2}}:=&\hbar\bigg(
t_{\infty^{(1)},1}\partial_{t_{\infty^{(1)},1}}+t_{\infty^{(2)},1}\partial_{t_{\infty^{(2)},1}}+t_{\infty^{(3)},1}\partial_{t_{\infty^{(3)},1}}\cr
&+2t_{\infty^{(1)},2}\partial_{t_{\infty^{(1)},2}}+2t_{\infty^{(2)},2}\partial_{t_{\infty^{(2)},2}}+2t_{\infty^{(3)},2}\partial_{t_{\infty^{(3)},2}}
\bigg) 
\end{align}
\item   The vectors $\mathbf{a}_1$, $\mathbf{a}_2$, $\mathbf{a}_3$:
\begin{align}
     \label{Defai} \mathcal{L}_{\mathbf{a}_1}:=&\hbar\bigg( 2(t_{\infty^{(1)},2}-t_{\infty^{(2)},2})(t_{\infty^{(1)},2}-t_{\infty^{(3)},2}) \partial_{t_{\infty^{(1)},2}}\cr
     &+ \bigg[t_{\infty^{(1)},2}(3t_{\infty^{(1)},1}-t_{\infty^{(2)},1}-t_{\infty^{(3)},1})-(t_{\infty^{(1)},2}+t_{\infty^{(2)},2}+t_{\infty^{(3)},2})t_{\infty^{(1)},1}\cr &+t_{\infty^{(2)},2}t_{\infty^{(3)},1}+t_{\infty^{(3)},2}t_{\infty^{(2)},1}\bigg]\partial_{t_{\infty^{(1)},1}}\bigg)\cr
 \mathcal{L}_{\mathbf{a}_2}:=&\hbar\bigg( 2(t_{\infty^{(2)},2}-t_{\infty^{(1)},2})(t_{\infty^{(2)},2}-t_{\infty^{(3)},2}) \partial_{t_{\infty^{(2)},2}}\cr 
     &+ \bigg[t_{\infty^{(2)},2}(3t_{\infty^{(2)},1}-t_{\infty^{(1)},1}-t_{\infty^{(3)},1})-(t_{\infty^{(1)},2}+t_{\infty^{(2)},2}+t_{\infty^{(3)},2})t_{\infty^{(2)},1}\cr &+t_{\infty^{(1)},2}t_{\infty^{(3)},1}+t_{\infty^{(3)},2}t_{\infty^{(1)},1}\bigg]\partial_{t_{\infty^{(2)},1}}\bigg)\cr
 \mathcal{L}_{\mathbf{a}_3}:=&\hbar\bigg( 2(t_{\infty^{(3)},2}-t_{\infty^{(1)},2})(t_{\infty^{(3)},2}-t_{\infty^{(2)},2}) \partial_{t_{\infty^{(3)},2}}\cr  
     &+ \bigg[t_{\infty^{(3)},2}(3t_{\infty^{(3)},1}-t_{\infty^{(1)},1}-t_{\infty^{(2)},1})-(t_{\infty^{(1)},2}+t_{\infty^{(2)},2}+t_{\infty^{(3)},2})t_{\infty^{(3)},1}\cr &+t_{\infty^{(1)},2}t_{\infty^{(2)},1}+t_{\infty^{(2)},2}t_{\infty^{(1)},1}\bigg]\partial_{t_{\infty^{(3)},1}}\bigg)
 \end{align}   
 \end{itemize} 
\end{definition}

Vectors $\mathbf{v}_{\infty,1}$ and $\mathbf{v}_{\infty,2}$ correspond to the trace reduction of the Lax matrix $\td{L}(\lambda)$ that is well-known to give trivial directions for the evolutions of the Darboux coordinates. Vector $\mathbf{u}_{\infty,1}$ corresponds to a translation of $\lambda$ and thus shall also provide a trivial direction for the Hamiltonian evolution after a suitable rescaling of the Darboux coordinates. Finally $\mathbf{u}_{\infty,2}$ corresponds to a dilatation of $\lambda$ and thus shall also provide a trivial direction for the Hamiltonian evolution after a suitable rescaling of the Darboux coordinates. The additional directions $\mathbf{a}_1$, $\mathbf{a}_2$ and $\mathbf{a}_3$ are new and they only involve derivatives in one sheet. In our case, these directions also provide trivial evolutions for the Darboux coordinates in the sense that the evolutions are linear in $(q,p)$, therefore, a suitable symplectic change of coordinates eliminates these evolutions. Note that $(\mathbf{v}_{\infty,1},\mathbf{v}_{\infty,2},\mathbf{a}_1,\mathbf{a}_2,\mathbf{a}_3)$ are linearly independent in $T\mathbb{B}$, alternatively $(\mathbf{v}_{\infty,1},\mathbf{v}_{\infty,2},\mathbf{u}_{\infty,1},\mathbf{u}_{\infty,2},\mathbf{a}_1)$ are also linearly independent. One may also complement the previous sub-space by adding $\mathbf{e}_{\infty^{(1)},1}$ as a final independent direction. 

As mentioned above, these directions in the tangent space are chosen so that we have the following theorem.

\begin{theorem}\label{Defs}
The evolutions of the Darboux coordinates are given by
    \begin{align}
\mathcal{L}_{\mathbf{v}_{\infty,1}}[q] =&0\cr
          \mathcal{L}_{\mathbf{v}_{\infty,1}}[p] =& \hbar \cr
\mathcal{L}_{\mathbf{v}_{\infty,2}}[q] =&0\cr
        \mathcal{L}_{\mathbf{v}_{\infty,2}}[p] =& \hbar q   \cr
          \mathcal{L}_{\mathbf{u}_{\infty,1}}[q]  =&  - \hbar   \cr
       \mathcal{L}_{\mathbf{u}_{\infty,1}}[p] =&0 \cr
       \mathcal{L}_{\mathbf{u}_{\infty,2}}[q] = &-\hbar q \cr
      \mathcal{L}_{\mathbf{u}_{\infty,2}}[p] = & \hbar p\cr 
 \mathcal{L}_{\mathbf{a}_{1}}[q] = &\hbar\left( (t_{\infty^{(2)},2}+t_{\infty^{(3)},2})q-2p+t_{\infty^{(2)},1}+t_{\infty^{(3)},1}\right)  \cr
      \mathcal{L}_{\mathbf{a}_1}[p] = & \hbar\left[-(t_{\infty^{(2)},2}+t_{\infty^{(3)},2})p+2t_{\infty^{(2)},2}t_{\infty^{(3)},2}q +t_{\infty^{(3)},1}t_{\infty^{(2)},2}+t_{\infty^{(3)},2}t_{\infty^{(2)},1}\right]\cr
 \mathcal{L}_{\mathbf{a}_{2}}[q] = &\hbar\left( (t_{\infty^{(1)},2}+t_{\infty^{(3)},2})q-2p+t_{\infty^{(1)},1}+t_{\infty^{(3)},1}\right)  \cr
      \mathcal{L}_{\mathbf{a}_2}[p] = & \hbar\left[-(t_{\infty^{(1)},2}+t_{\infty^{(3)},2})p+2t_{\infty^{(1)},2}t_{\infty^{(3)},2}q +t_{\infty^{(3)},1}t_{\infty^{(1)},2}+t_{\infty^{(3)},2}t_{\infty^{(1)},1}\right]\cr
 \mathcal{L}_{\mathbf{a}_{3}}[q] = &\hbar\left( (t_{\infty^{(1)},2}+t_{\infty^{(2)},2})q-2p+t_{\infty^{(1)},1}+t_{\infty^{(2)},1}\right)  \cr
      \mathcal{L}_{\mathbf{a}_3}[p] = & \hbar\left[-(t_{\infty^{(1)},2}+t_{\infty^{(2)},2})p+2t_{\infty^{(1)},2}t_{\infty^{(2)},2}q +t_{\infty^{(2)},1}t_{\infty^{(1)},2}+t_{\infty^{(2)},2}t_{\infty^{(1)},1}\right]\cr
\mathcal{L}_{\mathbf{e}_{\infty^{(1)},1}}[q]=&\frac{1}{(t_{\infty^{(1)},2}-t_{\infty^{(2)},2})(t_{\infty^{(1)},2}-t_{\infty^{(3)},2})}\left[-3p^2+2P_1(q)p-P_2(q)  \right]\cr
\mathcal{L}_{\mathbf{e}_{\infty^{(1)},1}}[p]=&\frac{1}{(t_{\infty^{(1)},2}-t_{\infty^{(2)},2})(t_{\infty^{(1)},2}-t_{\infty^{(3)},2})}\left[-P_1'(q) p^2+P_2'(q)p-P_3'(q)  \right]      
  \end{align} 
The evolutions are Hamiltonian in each direction in the sense that
\beq \label{HamiltonianEvolutionsDef} \mathcal{L}_{\boldsymbol{\alpha}}[q]=\frac{\partial \text{Ham}_{(\boldsymbol{\alpha})}(q,p;\hbar)}{\partial p}\,,\, \mathcal{L}_{\boldsymbol{\alpha}}[p]=-\frac{\partial \text{Ham}_{(\boldsymbol{\alpha})}(q,p;\hbar)}{\partial q} \eeq
The explicit expression of the associated Hamiltonians is given in \autoref{DirectinHams}. 
\end{theorem}

For completeness, we also propose the formulas for the associated auxiliary matrices in \autoref{AppendixAuxiliaryGeneral}. 

\medskip

The Hamiltonians given in \autoref{DirectinHams} are determined up to purely time dependent terms (i.e. independent of the Darboux coordinates) that do not modify the evolutions of the Darboux coordinates in \autoref{Defs}. In \autoref{DirectinHams}, these purely time dependent terms are chosen (See \eqref{ConstantTermsHam}) so that the fundamental symplectic two-form $\Omega$ admits an easier reduction (\autoref{TheoNonTrivialEvolution}). However, this only determines the extra-terms up to an exact term $dG_0$ that we shall define later to match the Hamiltonian form evaluated at $\hbar=0$ with the Jimbo-Miwa-Ueno differential (\autoref{PropJMUDifferential}).

\medskip

One may regroup the Hamiltonians into a differential form.

\begin{definition}[Hamiltonian one-form]\label{DefHamiltonianOneFormGl3} We shall denote $\overline{\omega}(q,p;\hbar)$ the Hamiltonian one-form:
\beq \overline{\omega}(q,p;\hbar):= \underset{i=1}{\overset{3}{\sum}}\underset{k=1}{\overset{2}{\sum}}\,\text{Ham}_{(\mathbf{e}_{\infty^{(i)},k})}(q,p;\hbar) dt_{\infty^{(i)},k}\eeq
\end{definition}

\normalsize{Let} us finally mention that the existence of Hamiltonian systems of the form \eqref{HamiltonianEvolutionsDef} is equivalent to the fact that our choice of Darboux coordinates $(q,p)$ provides a fundamental symplectic two-form $\Omega$ as defined by D. Yamakawa in \cite{Yamakawa2019FundamentalTwoForms}.

\begin{definition}\label{DefinitionFundamentalTwoForm}The symplectic structure \eqref{HamiltonianEvolutionsDef} is characterized by the fundamental two-form
\beq
     \Omega :=\hbar  dq \wedge dp -\sum_{i=1}^3\sum_{k=1}^{2} dt_{\infty^{(i)},k}\wedge d\text{Ham}_{(\mathbf{e}_{\infty^{(i)},k})}(q,p;\hbar)
\eeq  
\end{definition}

Adding purely time-dependent terms to the Hamiltonians does not change the Hamiltonian evolutions. It may change $\Omega$: if the added Hamiltonian one-form is $\eta(\mathbf{t})$, then $\Omega$ changes by $d\eta$. In particular, adding an exact one-form leaves $\Omega$ unchanged. The JMU differential is fixed by the connection and its formal normalization; the choice of Hamiltonian normalization affects its expression in terms of the Hamiltonians. We choose the terms in \eqref{ConstantTermsHam} to simplify these relations.

\subsection{Symplectic reduction of the tangent space}\label{SecGL3Reduction}
\subsubsection{Shifted Darboux coordinates and definition of trivial and non-trivial times}
In the previous section we obtained five directions in which the Hamiltonian system is linear in the Darboux coordinates. We now introduce time-dependent coordinates that trivialize these directions and an explicit change of time coordinates. The transformation below is anti-symplectic on each fixed-time fibre; this accounts for the overall sign in the reduced two-form. The first step consists of removing the trace of the Lax matrices and performing translations and dilatations, as in \cite{marchal2023hamiltonian}. The second step trivializes the direction $\mathbf{a}_1$ (or either of the analogous directions $\mathbf{a}_2$ and $\mathbf{a}_3$). The resulting coordinates depend on a single non-trivial time, denoted $\tau$.

\begin{definition}[Shifted Darboux coordinates]\label{DefShif}Work locally with a fixed branch
\beq \Delta:=\sqrt{(t_{\infty^{(2)},2}-t_{\infty^{(1)},2})(t_{\infty^{(1)},2}-t_{\infty^{(3)},2})(t_{\infty^{(3)},2}-t_{\infty^{(2)},2})}.\eeq
All square roots below are chosen consistently with this branch. Define
\begin{align}
\check{q}:=&\frac{t_{\infty^{(1)},2}-t_{\infty^{(3)},2}}{\Delta}(-p+t_{\infty^{(2)},2}q+t_{\infty^{(2)},1})\cr
\check{p}:=&\frac{t_{\infty^{(2)},2}-t_{\infty^{(3)},2}}{\Delta}(p-t_{\infty^{(1)},2}q-t_{\infty^{(1)},1})
\end{align}
as well as the new time coordinates:
\begin{align}
\tau:=& \frac{(t_{\infty^{(2)},1}-t_{\infty^{(3)},1})t_{\infty^{(1)},2}
+ (t_{\infty^{(3)},1}-t_{\infty^{(1)},1})t_{\infty^{(2)},2}+(t_{\infty^{(1)},1}-t_{\infty^{(2)},1})t_{\infty^{(3)},2}}{\Delta}\cr
T_1:=&t_{\infty^{(1)},2}+t_{\infty^{(2)},2}+t_{\infty^{(3)},2}\cr
T_2:=&t_{\infty^{(1)},1}+t_{\infty^{(2)},1}+t_{\infty^{(3)},1}\cr
T_3:=&t_{\infty^{(2)},2} \cr
T_4:=&t_{\infty^{(1)},1}\cr
T_5:=&t_{\infty^{(3)},2}
\end{align}
\end{definition}
As we will see below, the crucial property of these shifted Darboux coordinates is that they satisfy \autoref{SymplecticReduction} and shall provide simpler Hamiltonian structures. Note that \autoref{DefShif} may be seen as a generalization of the shift introduced in \cite{marchal2023hamiltonian} for connections in $\mathfrak{gl}_3(\mathbb{C})$. 

\medskip

Let us now mention that one may easily invert the map from irregular times to the times defined in \autoref{DefShif}. Indeed, we have the following proposition.
\begin{proposition}\label{InversionTimes}On the chart $2T_1-3T_3-3T_5\neq0$, the inverse map from $(\mathbf{T},\tau)$ to $\mathbf{t}$ is given by
\begin{align}t_{\infty^{(1)},2}=&T_1-T_3-T_5\cr
t_{\infty^{(2)},2}=&T_3\cr
t_{\infty^{(3)},2}=&T_5\cr
t_{\infty^{(1)},1}=&T_4\cr
t_{\infty^{(2)},1}=& \frac{\sqrt{(T_1-2T_3-T_5)(T_1-T_3-2T_5)(T_3-T_5)}\tau +T_1T_2-T_1T_4-2T_2T_3-T_2T_5+3T_3T_4}{2T_1-3T_3-3T_5}\cr
t_{\infty^{(3)},1}=&-\frac{\sqrt{(T_1-2T_3-T_5)(T_1-T_3-2T_5)(T_3-T_5)}\tau -T_1T_2+T_1T_4+T_2T_3+2T_2T_5-3T_4T_5}{2T_1-3T_3-3T_5}\cr
\end{align}   
The square root in these formulas equals the branch $\Delta$ chosen in \autoref{DefShif}.
\end{proposition}

\begin{proof}
The proof follows from direct verification.
\end{proof}

\subsubsection{Reduced Hamiltonian evolutions}

The definition of the times $(\mathbf{T},\tau)$ and shifted Darboux coordinates $(\check{q},\check{p})$ is such that the following theorem holds.

\begin{corollary}\label{TheoNonTrivialEvolution}The Hamiltonian evolution of $(\check{q},\check{p})$ in the unique non-trivial direction is given by:
\begin{align}  \hbar \partial_\tau \check{q}=&-2\check{q}\check{p}-\check{q}^2 +\tau \check{q}+t_{\infty^{(2)},0}\cr
    \hbar \partial_\tau \check{p}=&\check{p}^2+2\check{q}\check{p}-\tau \check{p}+t_{\infty^{(1)},0}+\hbar 
\end{align}
so that the corresponding Hamiltonians are chosen as
\beq \label{DefHamtau1}\text{Ham}_{(\tau)}(\check{q},\check{p};\hbar)=-\left(\check{q}\check{p}^2+\check{q}^2\check{p} -\tau \check{q}\check{p}-t_{\infty^{(2)},0}\check{p} +(t_{\infty^{(1)},0}+\hbar)\check{q}\right) 
\eeq
and
\beq\label{DefHamtau2} \text{Ham}_{(T_i)}(\check{q},\check{p};\hbar)=0\,\,,\,\, \forall\, i\in\llbracket 1,5\rrbracket\eeq
where we have noted $\text{Ham}_{(\tau)}(\check{q},\check{p};\hbar)$ and $\text{Ham}_{(T_i)}(\check{q},\check{p};\hbar)$, the Hamiltonians corresponding to directions $\hbar \partial_\tau$ and $\hbar \partial_{T_i}$ for $i\in \llbracket 1,5\rrbracket$.
\end{corollary}

The proof can be done by direct computations following from the evolutions of $(q,p)$ of \autoref{Defs} and the definition of the shifted Darboux coordinates. In fact, we shall prove the following stronger result.

\begin{corollary}\label{SymplecticReduction}The fundamental symplectic two-form $\Omega$ (\autoref{DefinitionFundamentalTwoForm}) reduces to
\begin{align} \Omega=&-\hbar d\check{q}\wedge d\check{p} +d\tau \wedge d\text{Ham}_{(\tau)}(\check{q},\check{p};\hbar) 
\end{align} 
In particular it implies that  
\beq \forall\, k\in \llbracket 1,5\rrbracket\,:\, \partial_{T_k}\check{q}=\partial_{T_k}\check{p}=0\eeq
so $(\check{q},\check{p})$ only depend on $\tau$ but not on $(T_k)_{1\leq k\leq 5}$.    
\end{corollary}

\begin{proof}
    The proof is carried out via direct (computer-assisted) computations and is presented in \autoref{ProofSymplecticReduction}.
\end{proof}

\sloppy{\begin{remark}Let us remark that the equality in \autoref{SymplecticReduction} includes the purely time-dependent terms. In fact, the simple form of the Theorem is achieved when the purely time-dependent terms \eqref{ConstantTermsHam} in the definition of the Hamiltonians are included. Note however that the addition of any exact differential $df(\mathbf{t})$ of the irregular times to the Hamiltonian one-form $\overline{\omega}(q,p;\hbar)$ preserves the fundamental symplectic two-form $\Omega$, and as we shall see below, these exact terms play a role in the identification with the Jimbo-Miwa-Ueno isomonodromic tau-function.
\end{remark}}

\begin{remark}\label{RemarkHalfReducedCoordinates}One may partially reduce the Hamiltonian structure by only removing the trace and performing translation/dilatation as done for connections in $\mathfrak{gl}_2(\mathbb{C})$ in \cite{MarchalAlameddineP1Hierarchy2023,marchal2023hamiltonian}. This corresponds to defining 
\begin{align}
   \tau_1:=& \frac{t_{\infty^{(3)},2}-t_{\infty^{(2)},2}}{t_{\infty^{(2)},2}-t_{\infty^{(1)},2}}\cr
   \tau_2:=&\frac{(t_{\infty^{(2)},1}-t_{\infty^{(3)},1})t_{\infty^{(1)},2}
+ (t_{\infty^{(3)},1}-t_{\infty^{(1)},1})t_{\infty^{(2)},2}+(t_{\infty^{(1)},1}-t_{\infty^{(2)},1})t_{\infty^{(3)},2}
}{(t_{\infty^{(2)},2}-t_{\infty^{(1)},2})^{\frac{3}{2}}}\cr
\td{q}:=&(t_{\infty^{(2)},2}-t_{\infty^{(1)},2})^{\frac{1}{2}}\left(q +\frac{t_{\infty^{(2)},1}-t_{\infty^{(1)},1}}{t_{\infty^{(2)},2}-t_{\infty^{(1)},2}}\right)\cr
 \td{p}:=&(t_{\infty^{(2)},2}-t_{\infty^{(1)},2})^{-\frac{1}{2}}\left( p-t_{\infty^{(2)},2} q -t_{\infty^{(2)},1}\right)
\end{align}
and use $\mathbf{t}\mapsto (T_1,T_2,T_3,T_4,\tau_1,\tau_2)$ as a one-to-one change of time coordinates. In particular, $(\td{q},\td{p})$ are such that $\partial_{T_k}\td{q}=\partial_{T_k}\td{p}=0$ for $k\in \llbracket 1,4\rrbracket$ so that $(\td{q},\td{p})$ only depend on $\tau_1$ and $\tau_2$. Equivalently this corresponds to $\mathcal{L}_{\mathbf{v}_{\infty,k}}[\td{q}]=\mathcal{L}_{\mathbf{v}_{\infty,k}}[\td{p}]=\mathcal{L}_{\mathbf{u}_{\infty,k}}[\td{q}]=\mathcal{L}_{\mathbf{u}_{\infty,k}}[\td{p}]=0$ for $k\in \llbracket 1,2\rrbracket$.
\end{remark}

\autoref{SymplecticReduction} indicates that only one deformation is non-trivial for $(\check{q},\check{p})$: $\hbar\partial_{\tau}$. We denote $\td{A}_{\boldsymbol{\tau}}(\lambda)$ and $A_{\boldsymbol{\tau}}(\lambda)$ the corresponding deformation matrices. Moreover, the corresponding Hamiltonian evolutions of $(\check{q},\check{p})$ in this direction depend only on $\tau$ and contain no dependence on the trivial times $(T_k)_{1\leq k\leq 5}$. Note however that the Lax matrices $\td{L}(\lambda), \td{A}_{\boldsymbol{\tau}}(\lambda)$ may depend on the trivial times as we shall see below.

For completeness, we shall provide the expression of the Lax matrices $\td{L}(\lambda)$ using the shifted Darboux coordinates.

\begin{align}
    \left[\td{L}(\lambda)\right]_{1,1}=&t_{\infty^{(1)},2}\lambda+t_{\infty^{(1)},1}\cr
    \left[\td{L}(\lambda)\right]_{1,2}=&1\cr
    \left[\td{L}(\lambda)\right]_{1,3}=&1\cr
    \left[\td{L}(\lambda)\right]_{2,1}=&(t_{\infty^{(2)},2}-t_{\infty^{(1)},2})\left(-\check{q}\check{p} -\check{q}^2+\tau\check{q}+t_{\infty^{(2)},0}\right)\cr
    \left[\td{L}(\lambda)\right]_{2,2}=&t_{\infty^{(2)},2}\lambda+t_{\infty^{(2)},1}\cr
    \left[\td{L}(\lambda)\right]_{2,3}=&\frac{\Delta}{t_{\infty^{(1)},2}-t_{\infty^{(3)},2}}\check{q}\cr
    \left[\td{L}(\lambda)\right]_{3,1}=&(t_{\infty^{(1)},2}-t_{\infty^{(3)},2})\left(-\check{q}\check{p} -\check{q}^2+\tau\check{q}-t_{\infty^{(3)},0}\right)\cr
    \left[\td{L}(\lambda)\right]_{3,2}=&\frac{\Delta}{t_{\infty^{(1)},2}-t_{\infty^{(2)},2}}(\check{p}+\check{q}-\tau)\cr
    \left[\td{L}(\lambda)\right]_{3,3}=&t_{\infty^{(3)},2}\lambda+t_{\infty^{(3)},1}
\end{align}

\subsection{The Jimbo-Miwa-Ueno isomonodromic $\tau$-function}\label{SecGl3JMU}
In \cite{JimboMiwaUeno}, Jimbo, Miwa and Ueno extended the theory of isomonodromic deformations, developed for Fuchsian systems by Schlesinger \cite{schlesinger1912klasse}, to systems with irregular singularities and constructed the associated isomonodromic tau-function.
 In our context, their definition is given by 
\beq \label{JMUGL3Def}\omega_{\text{JMU}}=-\Res_{\lambda \to \infty}\Tr\Big[ \td{\Psi}^{(\text{reg})}(\lambda)^{-1} (\hbar\partial_\lambda \td{\Psi}^{(\text{reg})}(\lambda)) dT(\lambda)\Big]\eeq
where $T(\lambda)=\text{diag}\left(t_{\infty^{(i)},2}\frac{\lambda^2}{2}+t_{\infty^{(i)},1}\lambda+t_{\infty^{(i)},0}\ln \lambda\right)_{1\leq i\leq 3}$ and $\td{\Psi}^{(\text{reg})}(\lambda)$ is given by the formal asymptotic expansion at infinity of $\td{\Psi}(\lambda)$:
\beq \td{\Psi}(\lambda)\overset{\lambda\to \infty}{\sim} \td{\Psi}^{(\text{reg})}(\lambda) e^{\frac{1}{\hbar}T(\lambda)}\eeq
The last formal asymptotic expansion can be rewritten using our former notation regarding the local diagonalization at infinity \eqref{PsiTT}:
\beq \td{\Psi}(\lambda)=G_\infty(\lambda)^{-1} \Psi_{\infty}(\lambda)=G_{\infty}(\lambda)^{-1}\Psi_{\infty}^{(\text{reg})}(\lambda)e^{\frac{1}{\hbar} T(\lambda)} \,\,\Rightarrow\,\, \td{\Psi}^{(\text{reg})}(\lambda)=G_{\infty}(\lambda)^{-1}\Psi_{\infty}^{(\text{reg})}(\lambda) 
\eeq
Let us write the fundamental normal form at infinity as
\beq \td{\Psi}(\lambda):=\left(I_3+\sum_{k=1}^{\infty} \frac{F_k(\mathbf{t})}{\lambda^k}\right) e^{\frac{1}{\hbar}T(\lambda)}\eeq
Then, the Jimbo-Miwa-Ueno differential defined by \eqref{JMUGL3Def} is given by
\begin{align}
    \omega_{\text{JMU}}=&-\hbar[F_1]_{1,1} dt_{\infty^{(1)},1} - \hbar[F_1]_{2,2} dt_{\infty^{(2)},1}-\hbar[F_1]_{3,3} dt_{\infty^{(3)},1}\cr& 
+\frac{\hbar}{2} \left(([F_1]_{1,1})^2+ [F_1]_{1,2}[F_1]_{2,1}+ [F_1]_{1,3}[F_1]_{3,1} -2[F_2]_{1,1}\right) dt_{\infty^{(1)},2}\cr
    &+\frac{\hbar}{2}\left(([F_1]_{2,2})^2+ [F_1]_{1,2}[F_1]_{2,1}+ [F_1]_{2,3}[F_1]_{3,2} -2[F_2]_{2,2}\right) dt_{\infty^{(2)},2}\cr
    &+\frac{\hbar}{2}\left(([F_1]_{3,3})^2+ [F_1]_{1,3}[F_1]_{3,1}+ [F_1]_{2,3}[F_1]_{3,2} -2[F_2]_{3,3}\right) dt_{\infty^{(3)},2}   
\end{align}
From the knowledge of $\td{L}(\lambda)$, one can easily compute the first terms $F_1$ and $F_2$ of the formal asymptotic expansion at infinity using $\hbar \partial_{\lambda} \td{\Psi}(\lambda)=\td{L}(\lambda) \td{\Psi}(\lambda)$. Inserting these results into the former expression provides the explicit expression for the Jimbo-Miwa-Ueno differential.

\begin{theorem}[Expression of the Jimbo-Miwa-Ueno differential]\label{PropJMUDifferential} 
The Jimbo-Miwa-Ueno differential is given by
\beq
  \omega_{\text{JMU}}=\text{Ham}_{(\tau)}(\check{q},\check{p};\hbar=0)d\tau +dG_0(\mathbf{T})=\left(\text{Ham}_{(\tau)}(\check{q},\check{p};\hbar) +\hbar \check{q}\right)d\tau +dG_0(\mathbf{T})
\eeq
or equivalently
\beq \omega_{\text{JMU}}=-\overline{\omega}(q,p;\hbar=0)+dG_0(\mathbf{t})=-\sum_{i=1}^3\sum_{k=1}^2 \text{Ham}_{(\mathbf{e}_{\infty^{(i)},k})}(q,p;\hbar=0) dt_{\infty^{(i)},k} +dG_0(\mathbf{t})\eeq
where the $\hbar=0$ condition is understood for $(q,p,\mathbf{t},\mathbf{t}_0)$ fixed or equivalently for $(\check{q},\check{p},\mathbf{T},\mathbf{t}_0)$ fixed in \eqref{Hamcertain}, \eqref{OtherHamiltonians} and \eqref{DefHamtau1} while  $G_0$ is defined by either
\scriptsize{\begin{align}\label{DefG0}
    G_0(\mathbf{t}):=&\frac{1}{2}\left(t_{\infty^{(1)},0}t_{\infty^{(3)},0}\ln(t_{\infty^{(3)},2}-t_{\infty^{(1)},2})+t_{\infty^{(1)},0}t_{\infty^{(2)},0}\ln(t_{\infty^{(2)},2}-t_{\infty^{(1)},2})+t_{\infty^{(2)},0}t_{\infty^{(3)},0}\ln(t_{\infty^{(2)},2}-t_{\infty^{(3)},2})\right)\cr&
-\frac{t_{\infty^{(1)},0}(t_{\infty^{(1)},1})^2}{2(t_{\infty^{(1)},2}-t_{\infty^{(3)},2})}-\frac{t_{\infty^{(2)},0}(t_{\infty^{(2)},1})^2}{2(t_{\infty^{(2)},2}-t_{\infty^{(3)},2})}
+\frac{t_{\infty^{(3)},0}(t_{\infty^{(3)},1})^2}{2(t_{\infty^{(1)},2}-t_{\infty^{(3)},2})}\cr&
-\frac{(t_{\infty^{(1)},2}-t_{\infty^{(2)},2})(t_{\infty^{(3)},1})^2t_{\infty^{(2)},0}}{2(t_{\infty^{(2)},2}-t_{\infty^{(3)},2})(t_{\infty^{(1)},2}-t_{\infty^{(3)},2})}
+\frac{t_{\infty^{(3)},1}t_{\infty^{(1)},1}t_{\infty^{(1)},0}}{(t_{\infty^{(1)},2}-t_{\infty^{(3)},2})}+\frac{t_{\infty^{(2)},0}t_{\infty^{(3)},1}t_{\infty^{(2)},1}}{t_{\infty^{(2)},2}-t_{\infty^{(3)},2}}
\end{align}}
\normalsize{or} in terms of the coordinates $\mathbf{T}:=(T_1,\dots,T_5,\tau)$ by:
\footnotesize{\begin{align}\label{G0T}G_0(\mathbf{T}):=&\frac{1}{2}\left(t_{\infty^{(1)},0}t_{\infty^{(3)},0}\ln(2T_5-T_1+T_3)+t_{\infty^{(1)},0}t_{\infty^{(2)},0}\ln(2T_3+T_5-T_1)+t_{\infty^{(2)},0}t_{\infty^{(3)},0}\ln(T_3-T_5)\right)\cr&
+\frac{(T_2-3T_4)^2((T_1-2T_3-T_5)t_{\infty^{(2)},0}+t_{\infty^{(3)},0}(T_1-T_3-2T_5))}{2(2T_1-3T_3-3T_5)^2}\cr&
+\frac{\sqrt{(T_1-T_3-2T_5)(T_1-2T_3-T_5)(T_3-T_5)}(T_2-3T_4)(t_{\infty^{(2)},0}-t_{\infty^{(3)},0})\tau}{(2T_1-3T_3-3T_5)^2}\cr&
-\frac{((4T_1-5T_3-7T_5)t_{\infty^{(2)},0}-t_{\infty^{(3)},0}(T_3-T_5))(T_1-2T_3-T_5)\tau^2}{2(2T_1-3T_3-3T_5)^2}
\end{align}}
\normalsize{}
\end{theorem}

\begin{remark}Note that there is a sign difference in the formal evaluation at $\hbar=0$ of the Hamiltonian differential between the reduced and the non-reduced case:
\beq\omega_{\text{JMU}}=-\overline{\omega}(q,p;\hbar=0)+dG_0(\mathbf{t})=\text{Ham}_{(\tau)}(\check{q},\check{p};\hbar=0)d\tau +dG_0(\mathbf{T})\eeq
This is because the change of Darboux coordinates $(q,p)\leftrightarrow(\check{q},\check{p})$ is not symplectic but rather $d_{\mathrm{fib}}\check{q}\wedge d_{\mathrm{fib}}\check{p}=-d_{\mathrm{fib}}q\wedge d_{\mathrm{fib}}p$ at fixed times; the full exterior differentials also contain time terms.
\end{remark}

\begin{proof}The proof follows from direct but lengthy computations. A corresponding Maple worksheet is available at \url{http://math.univ-lyon1.fr/~marchal/AdditionalRessources/index.html}. Note that the exact term $dG_0$ (that does not modify $\Omega$) added is necessary and could have been included in the definition of the Hamiltonians without changing either the dynamics of the Darboux coordinates or $\Omega$.
\end{proof}

We also recall that the Jimbo-Miwa-Ueno differential is a closed form \cite{JimboMiwaUeno} so that one can define (up to a time-independent multiplicative constant) the Jimbo-Miwa-Ueno isomonodromic tau-function $\tau_{\text{JMU}}$ by:
\begin{align}\label{JMUtaufunction} \hbar^2\, d (\ln \tau_{\text{JMU}}(\mathbf{t})):=& \omega_{\text{JMU}}=-\sum_{i=1}^3\sum_{k=1}^2 \text{Ham}_{(\mathbf{e}_{\infty^{(i)},k})}(q,p;\hbar=0) dt_{\infty^{(i)},k} +dG_0(\mathbf{t})\cr
=&\text{Ham}_{(\tau)}(\check{q},\check{p};\hbar=0)d\tau +dG_0(\mathbf{T})\end{align}
Here $d$ denotes the exterior differential on the deformation base after pullback along an isomonodromic solution. The relation is $\omega_{\mathrm{JMU}}=-\overline{\omega}(q,p;\hbar=0)+dG_0$; in the formal evaluation only the explicit $\hbar$ dependence of the Hamiltonians is set to zero, while the solution $(q,p)$ is retained. Closedness therefore concerns this evaluated form on the deformation base, not a form on the full phase space. Explicit time dependence alone does not imply that a Hamiltonian one-form is non-closed.

\subsection{Canonical choice of trivial times and reduced Lax matrices}
Since the shifted Darboux coordinates are independent of the five trivial times, we may choose convenient values for them. We set the traces to zero, use dilatation and translation to fix $t_{\infty^{(1)},2}=\frac12$ and $t_{\infty^{(2)},2}=0$, and use the remaining trivial direction to set $t_{\infty^{(1)},1}=0$. Thus $\td{L}^{[\infty,1]}=\mathrm{diag}(\frac12,0,-\frac12)$.

\begin{definition}[Canonical choice of trivial times] \label{DefCanonicalTimesGl3}We shall define \beq\left(t_{\infty^{(1)},2},t_{\infty^{(2)},2},t_{\infty^{(3)},2},t_{\infty^{(1)},1},t_{\infty^{(2)},1},t_{\infty^{(3)},1}\right)=\left(\frac{1}{2},0,-\frac{1}{2},0,\frac{\tau}{3},-\frac{\tau}{3}\right)\eeq
as the canonical choice of trivial times. It is equivalent to setting
\beq\left(T_1,T_2,T_3,T_4,T_5\right)=\left(0,0,0,0,-\frac{1}{2}\right)\eeq and we get $\tau=3t_{\infty^{(2)},1}=-3t_{\infty^{(3)},1}$ (with $\Delta=\frac12$) and $\partial_\tau=\frac{1}{3}\left(\partial_{t_{\infty^{(2)},1}}-\partial_{t_{\infty^{(3)},1}}\right) $.
\end{definition}

We have the following proposition
\begin{proposition}\label{PropLaxMatricesReduced}As mentioned in \autoref{TheoNonTrivialEvolution} the only non-trivial evolution of $(\check{q},\check{p})$ is characterized by the Hamiltonian system:
\begin{align*} \hbar \partial_{\tau}\check{q}=&-2\check{q}\check{p}-\check{q}^2+\tau \check{q}+t_{\infty^{(2)},0}\cr
\hbar \partial_{\tau}\check{p}=&\check{p}^2+2\check{q}\check{p}-\tau \check{p}+t_{\infty^{(1)},0}+\hbar
\end{align*}
whose corresponding Hamiltonian is
\beqq\text{Ham}_{(\tau)}(\check{q},\check{p};\hbar)=-\left(\check{q}\check{p}^2+\check{q}^2\check{p} -\tau \check{q}\check{p}-t_{\infty^{(2)},0}\check{p} +(t_{\infty^{(1)},0}+\hbar)\check{q}\right)\eeqq
Moreover, under the choice of canonical trivial times of \autoref{DefCanonicalTimesGl3}, we have
\beqq \check{q}=-2p+\frac{2}{3}\tau\,,\, \check{p}=p-\frac{1}{2}q  \,\Leftrightarrow\, q=-\check{q}-2\check{p}+\frac{2}{3}\tau\,,\, p=-\frac{1}{2}\check{q}+\frac{\tau}{3}\eeqq
and the Lax matrices simplify into
\begin{align*}
    \td{L}(\lambda)=&\begin{pmatrix}
        \frac{\lambda}{2}&1&1\\
\frac{1}{2}\left(\check{q}\check{p}+\check{q}^2-\tau \check{q}-t_{\infty^{(2)},0}\right)& \frac{\tau}{3}&\frac{\check{q}}{2}\\
-\check{q}\check{p}-\check{q}^2+\tau \check{q}-t_{\infty^{(3)},0}& \check{p}+\check{q}-\tau&-\frac{\lambda}{2}-\frac{\tau}{3}
    \end{pmatrix}\cr
\td{A}_{\boldsymbol{\tau}}(\lambda)=&
\frac{1}{3}\begin{pmatrix}0&-2&1\\
-(\check{q}\check{p} +\check{q}^2-\tau \check{q}-t_{\infty^{(2)},0}) & 
\lambda -3(\check{p}+\check{q}) +\frac{7\tau}{3}& 
2\check{q}\\
-(\check{q}\check{p} +\check{q}^2-\tau \check{q}+t_{\infty^{(3)},0})& 4(\check{p}+\check{q}-\tau)& -\lambda -3\check{q}-\frac{\tau}{3}
\end{pmatrix}
\end{align*}
and
\begin{align*}
    L_{3,1}(\lambda)=&\frac{\hbar\left(p^2+P_2(q)-\frac{\hbar}{2}\right)}{\lambda-q}+P_3(\lambda)+p^3+\left(P_2(q)-\frac{\hbar}{2}\right)p-P_3(q)+\frac{\hbar}{4}(\lambda-q)\notag\\
L_{3,2}(\lambda)=&\frac{\hbar p}{\lambda-q}-P_2(\lambda)+\frac{\hbar}{2}\notag\\
L_{3,3}(\lambda)=&\frac{\hbar}{\lambda-q} \,,\,\text{ with }\notag\\
P_1(\lambda)=&0\notag\\
P_2(\lambda)=&-\frac{1}{4}\lambda^2-\frac{1}{6}\tau \lambda +\frac{1}{2}t_{\infty^{(2)},0}+t_{\infty^{(3)},0}-\frac{\tau^2}{9}\notag\\ 
P_3(\lambda)=&-\frac{\lambda}{4}\left(\frac{1}{3}\tau \lambda +t_{\infty^{(2)},0}+\frac{2}{9}\tau^2 \right)
\end{align*}
and
\begin{align*}
\left[A_{\boldsymbol{\tau}}(\lambda)\right]_{1,1}=&\,\frac{(-2p^2-2P_2(q)+\hbar)}{\lambda-q}-\frac{1}{3}\lambda+p-\frac{1}{2}q\cr 
\left[A_{\boldsymbol{\tau}}(\lambda)\right]_{1,2}=&-\frac{2p}{\lambda-q}+\frac{1}{3}\cr
\left[A_{\boldsymbol{\tau}}(\lambda)\right]_{1,3}=&-\frac{2}{\lambda-q}\cr
\left[A_{\boldsymbol{\tau}}(\lambda)\right]_{2,1}=&\, \frac{p(-2p^2-2P_2(q)+\hbar)}{\lambda-q}+\frac{1}{6}\tau\lambda+\frac{1}{6}\tau q+\frac{1}{9}\tau^2+\frac{1}{2}t_{\infty^{(2)},0}-\frac{\hbar}{6}\cr
\left[A_{\boldsymbol{\tau}}(\lambda)\right]_{2,2}=&-\frac{2p^2}{\lambda-q}-\frac{\lambda}{6} -p+\frac{\tau}{3} \cr
\left[A_{\boldsymbol{\tau}}(\lambda)\right]_{2,3}=&-\frac{2p}{\lambda-q}+\frac{1}{3}\cr
\left[A_{\boldsymbol{\tau}}(\lambda)\right]_{3,1}=&-\frac{(6p^2-\hbar)(p^2+P_2(q)-\frac{\hbar}{2})}{3(\lambda-q)}-\frac{1}{36}\tau \lambda^2+\frac{1}{6}\left(\tau p -\frac{\tau^2}{9}-\frac{t_{\infty^{(2)},0}}{2}+\frac{\hbar}{2}\right)\lambda\cr
&+\frac{1}{6}\Big( 2p^3-\frac{1}{2}\left(q^2-\frac{4}{3}\tau q-\frac{8}{9}\tau^2+8\hbar-8t_{\infty^{(2)},0}-4t_{\infty^{(3)},0}\right)p \cr&
+\frac{\tau^2 q}{9}+\frac{\tau q^2}{6}+\hbar \tau -\frac{\hbar}{2}q +\frac{t_{\infty^{(2)},0}
}{2}q\Big) \cr 
\left[A_{\boldsymbol{\tau}}(\lambda)\right]_{3,2}=&-\frac{p(6p^2-\hbar)}{3(\lambda-q)}+\frac{1}{12}\lambda^2+\left(\frac{2\tau}{9}- \frac{p}{2}\right)\lambda \cr&+\frac{1}{6}\left(-3pq+ (q-2p)\tau +\frac{8\tau^2}{9} + 2t_{\infty^{(2)},0}-2t_{\infty^{(3)},0}-\hbar \right) \cr
\left[A_{\boldsymbol{\tau}}(\lambda)\right]_{3,3}=&-\frac{6p^2-\hbar}{3(\lambda-q)}-\frac{\lambda}{6}+p-\frac{\tau}{3}
\end{align*}
\end{proposition}

\subsection{A Hermitian two-matrix model for the $\mathfrak{gl}_3$ side and relation with the JMU tau-function}\label{SecGl32MM}
\subsubsection{Classical spectral curve and $\mathfrak{gl}_3$ duality gauge}\label{Sectiongl3duality}
Let us first define the spectral curve and classical spectral curve associated with $\td{L}(\lambda)$.

\begin{definition}[Spectral curves and classical spectral curves in the $\mathfrak{gl}_3$ setting]\label{DefGl3SepctralCurves}We shall define
\beq \mathcal{S}=\left\{(\lambda,y)\in \overline{\mathbb{C}}\times \overline{\mathbb{C}} \,\,\text{ such that }\,\, \det(y I_3-\td{L}(\lambda))=0\right\}
\eeq
as the spectral curve associated with the Lax matrix $\td{L}(\lambda) \in \mathfrak{gl}_3(\mathbb{C})$ defined in  \autoref{secGl3}.
Moreover, we shall define
\beq \mathcal{S}_0=\left\{(\lambda,y)\in \overline{\mathbb{C}}\times \overline{\mathbb{C}} \,\,\text{ such that }\,\, \underset{\hbar\to 0}{\lim} \det(y I_3-\td{L}(\lambda))=0\right\}
\eeq
as the classical spectral curve associated with the Lax matrix $\td{L}(\lambda) \in \mathfrak{gl}_3(\mathbb{C})$ defined in \autoref{secGl3}. 
where the limit as $\hbar\to 0$ is to be understood as a formal evaluation as $\hbar \to 0$ for fixed Darboux coordinates. 
\end{definition}

Let us remark that the notion of spectral curve (characteristic polynomial of the Lax matrix) does not depend on the choice of normalization of $\td{L}$ and therefore is a well-defined quantity on the space of connections. It is also independent of the choice of Darboux coordinates but it is modified by a $\lambda$-dependent gauge transformation. In particular, the spectral curve is not equal to the characteristic polynomial of the Lax matrix in the oper gauge $L(\lambda)$.

On the contrary, the classical spectral curve is equal to the characteristic polynomial of the Lax matrix in the oper gauge, i.e. to the quantum curve where $\hbar\partial_\lambda$ is formally replaced by $y$. On the down side, the classical spectral curve may depend on the choice of Darboux coordinates and may be badly defined if one performs a Darboux coordinates transformation explicitly depending on $\hbar$. Fortunately, the shifts introduced in \autoref{DefShif} (and later in \autoref{ShiftGl2} for the $\mathfrak{gl}_2$ setting) do not formally depend on $\hbar$ so that we may use both sets to define the classical spectral curves. Finally, we recall that
\beq \label{SpecCurvesGl3}
    \det\left(yI_3-\td{L}(\lambda)\right)=y^3-P_1(\lambda)y^2+P_2(\lambda)y-P_3(\lambda)- \left(p^3-P_1(q)p^2+P_2(q)p-P_3(q)\right)
\eeq
As we shall see in \autoref{SecDualitySpecCurves}, in order to identify the spectral curves by spectral duality and because the spectral curve is not invariant by gauge transformations, we need to adapt the gauges between both sides to obtain the duality. In particular, if we fix the matrix $\td{L}_{\text{P4}}(\xi)\in \mathfrak{gl}_2(\mathbb{C})$ as defined in  \autoref{SecP4JM}, then we need to perform a trivial diagonal gauge transformation to obtain the most general form of spectral duality.

\begin{definition}[$\mathfrak{gl}_3(\mathbb{C})$ duality gauge]\label{DefDualityGauge} Let $G_d(\lambda)=\exp\left(-\frac{t_{\infty^{(2)},2}}{2\hbar}\lambda^2\right) I_3$ and define $\hat{\Psi}_d(\lambda):=G_d(\lambda) \td{\Psi}(\lambda)$ so that 
\begin{align*}\hbar \partial_\lambda \hat{\Psi}_d(\lambda)=&\left(\td{L}(\lambda)-t_{\infty^{(2)},2}\lambda\right)\hat{\Psi}_d(\lambda):=\hat{L}_d(\lambda) \hat{\Psi}_d(\lambda)\cr
   \mathcal{L}_{\boldsymbol{\alpha}}[\hat{\Psi}_d(\lambda)]=&\left(\td{A}_{\boldsymbol{\alpha}}(\lambda)-\alpha_{\infty^{(2)},2}\frac{\lambda^2}{2}\right)\hat{\Psi}_d(\lambda):=\hat{A}_{{\boldsymbol{\alpha}},d}(\lambda) \hat{\Psi}_d(\lambda)
\end{align*}
The dual Darboux coordinate $p_d$ associated with $\hat{L}_d(\lambda)$ and $q$ is defined by
\beq p_d:=p-t_{\infty^{(2)},2}q\eeq
so that $\det\left(p_d I_3-\hat{L}_d(q)\right)=0$ and the corresponding (classical) spectral curve is defined by:
\begin{align}\mathcal{S}_d=&\left\{(\lambda,y)\in \overline{\mathbb{C}}\times \overline{\mathbb{C}} \,\,\text{ such that }\,\, \det(y I_3-\hat{L}_d(\lambda))=0\right\}\cr
\mathcal{S}_{d,0}=&\left\{(\lambda,y)\in \overline{\mathbb{C}}\times \overline{\mathbb{C}} \,\,\text{ such that }\,\, \underset{\hbar\to 0}{\lim}\det(y I_3-\hat{L}_d(\lambda))=0\right\}
\end{align}
\end{definition}

Note that we have $\det\left( (y+t_{\infty^{(2)},2}\lambda) I_3-\td{L}(\lambda)\right) =\det\left( yI_3-\hat{L}_d(\lambda)\right)$. In other words, the gauge transformation by a diagonal matrix is equivalent to a shift of $y$ in the (classical) spectral curve. Moreover, in this duality gauge and using the Darboux coordinates $(q,p_d)$, we have $\hat{L}_d(\lambda)=\text{diag}(t_{\infty^{(1)},2}-t_{\infty^{(2)},2},0,t_{\infty^{(3)},2}-t_{\infty^{(2)},2})\lambda+O(1)$. More precisely:   
   \begin{align}
    \hat{L}_d(\lambda) =  \text{diag}(t_{\infty^{(1)},2}-t_{\infty^{(2)},2},0,t_{\infty^{(3)},2}-t_{\infty^{(2)},2})\lambda +\td{L}^{[\infty,0]}
\end{align}
where $\td{L}^{[\infty,0]}$ is unchanged and has the following form
\begin{align*}
    \td{L}^{[\infty,0]}=\begin{pmatrix} t_{\infty^{(1)},1}  & 1  & 1  \\ \left[\td{L}^{[\infty,0]}\right]_{2,1} 
    & t_{\infty^{(2)},1} & -p_d + t_{\infty^{(2)},1} \\  
   \left[\td{L}^{[\infty,0]}\right]_{3,1} &-p_d+  (t_{\infty^{(3)},2}-t_{\infty^{(2)},2} )q + t_{\infty^{(3)},1}& t_{\infty^{(3)},1}  \end{pmatrix}
\end{align*}
which is expressed in the Darboux coordinates $(q,p_d)$ by
\footnotesize{\begin{align*}
\left[\td{L}^{[\infty,0]}\right]_{2,1}=&\frac{t_{\infty^{(1)},2} -t_{\infty^{(2)},2}}{t_{\infty^{(2)},2} - t_{\infty^{(3)},2}}  \bigg((p_d+t_{\infty^{(2)},2}q)^2-\left((t_{\infty^{(2)},2}+t_{\infty^{(3)},2})q+t_{\infty^{(2)},1}+t_{\infty^{(3)},1}  \right)(p_d+t_{\infty^{(2)},2}q)\cr&+ t_{\infty^{(2)},2}t_{\infty^{(3)},2}q^2
+(t_{\infty^{(2)},2}t_{\infty^{(3)},1}+t_{\infty^{(3)},2}t_{\infty^{(2)},1})q+ t_{\infty^{(2)},1}t_{\infty^{(3)},1}-t_{\infty^{(2)},0}(t_{\infty^{(2)},2}-t_{\infty^{(3)},2})\bigg)\cr
\left[\td{L}^{[\infty,0]}\right]_{3,1}=& -\frac{t_{\infty^{(1)},2} -t_{\infty^{(3)},2}}{t_{\infty^{(2)},2} - t_{\infty^{(3)},2}}  \bigg(
(p_d+t_{\infty^{(2)},2}q)^2-\left((t_{\infty^{(2)},2}+t_{\infty^{(3)},2})q+t_{\infty^{(2)},1}+t_{\infty^{(3)},1}  \right)(p_d+t_{\infty^{(2)},2}q)\cr&+ t_{\infty^{(2)},2}t_{\infty^{(3)},2}q^2
+(t_{\infty^{(2)},2}t_{\infty^{(3)},1}+t_{\infty^{(3)},2}t_{\infty^{(2)},1})q+ t_{\infty^{(2)},1}t_{\infty^{(3)},1}+t_{\infty^{(3)},0}(t_{\infty^{(2)},2}-t_{\infty^{(3)},2})
\bigg)
\end{align*}}
\normalsize{Consequently}, the spectral curve associated with $\hat{L}_d(\lambda)$ is given by
\begin{align} \label{SpecCurvesGl3shift}
    \det(yI_3-\hat{L}_d(\lambda))=&(y+t_{\infty^{(2)},2}\lambda)^3-P_1(\lambda)(y+t_{\infty^{(2)},2}\lambda)^2+P_2(\lambda)(y+t_{\infty^{(2)},2}\lambda)-P_3(\lambda)\cr&
    - \left(p^3-P_1(q)p^2+P_2(q)p-P_3(q)\right)
\end{align}

\subsubsection{A Hermitian two-matrix model for the $\mathfrak{gl}_3(\mathbb{C})$ duality gauge}
In this section, we shall derive a Hermitian two-matrix model whose classical spectral curve matches with
\begin{align}\label{SpecCurveFor2MM} 0={}&(y+t_{\infty^{(2)},2}\lambda)^3-P_1(\lambda) (y+t_{\infty^{(2)},2}\lambda)^2+P_2(\lambda) (y+t_{\infty^{(2)},2}\lambda)\cr&-P_3(\lambda)-\left(p^3-P_1(q)p^2+P_2(q)p-P_3(q)\right)\end{align}
given by \eqref{SpecCurvesGl3shift}. In order to achieve this, we look for a two-matrix model defined by the measure
\beq d\mu_{N}(M_1,M_2)\propto e^{-N\Tr\left(V_1(M_1)+V_2(M_2)-M_1M_2\right)}\eeq
where $(M_1,M_2)$ are two Hermitian matrices of size $N$ with $\hbar=N^{-1}$ and $V_1$, $V_2$ are two potentials whose derivatives will be rational functions in the present case. We shall also assume that there are no hard edges in the measure (i.e. we do not introduce an indicator function of some intervals for the eigenvalues). Finally, we shall denote $Z^{(\text{2MM})}(\mathbf{t};N)$ the associated partition function. The derivation of the classical spectral curve for the two-matrix models using the loop equations can be found in \cite{Eynard:2002kg} and we shall only use results for compactness. Let us define:
\beq P(\lambda,y):=\frac{1}{N}\left<\Tr \frac{V_1'(\lambda)-V_1'(M_1)}{\lambda-M_1}\frac{V_2'(y)-V_2'(M_2)}{y-M_2}\right>\eeq
and $P^{(0)}(x,y)$ the limit as $\hbar \to 0$ of $P(x,y)$. We recall that the bracket means the expected value relatively to the probability measure. The classical spectral curve is then given by:
\beq E^{(0)}(\lambda,y):=\left(V_1'(\lambda)-y\right)\left(V_2'(y)-\lambda\right)-P^{(0)}(\lambda,y)+1=0\eeq
Let us now take
\beq V_1'(\lambda)=a_{1}\lambda+a_2 \,\,,\,\, V_2'(y)=b_0y+b_1+\frac{b_2}{y-X_1}\eeq
for some unknowns $(a_1,a_2,b_0,b_1,b_2)\in \mathbb{C}^5$. We get that
\beq P(\lambda,y):=a_1b_0-\frac{a_1b_2}{N(y-X_1)}\left\langle\mathrm{Tr}\,(M_2-X_1I_N)^{-1}\right\rangle.\eeq
Therefore we obtain
\beq P^{(0)}(\lambda,y)=a_1\left( b_0-\frac{b_2}{y-X_1}C_0\right)  \,\,,\,\, C_0:=\left<\frac{1}{N}\Tr \frac{1}{M_2-X_1}\right>^{(0)} \eeq
The classical spectral curve is thus
{\small
\begin{align}  E^{(0)}(\lambda,y)=&\frac{-b_0}{y-X_1}\Big[ y^3+\frac{1}{b_0}\left(-(1+a_1b_0)\lambda+b_1-b_0(X_1+a_2)\right)y^2\cr&+ \frac{1}{b_0}\left(a_1\lambda^2+((X_1b_0-b_1)a_1+X_1+a_2)\lambda+(X_1a_2+a_1)b_0-b_1X_1-a_2b_1+b_2-1\right)y\cr
&+\frac{1}{b_0}\left(-X_1a_1\lambda^2+((a_1b_1-a_2)X_1-a_1b_2)\lambda+(-a_1b_0+a_2b_1+1)X_1-a_2b_2-a_1b_2C_0\right)\Big]\cr
\end{align}}

Identifying with the limit $\hbar\to 0$ of \eqref{SpecCurveFor2MM} provides
\begin{align}
&t_{\infty^{(1)},0}=-1\cr
&a_1=t_{\infty^{(1)},2}- t_{\infty^{(2)},2} \cr
&a_2=t_{\infty^{(1)},1}\cr
&b_0=\frac{1}{t_{\infty^{(3)},2}-t_{\infty^{(2)},2}}\cr
&b_1=-\frac{t_{\infty^{(3)},1}}{t_{\infty^{(3)},2}-t_{\infty^{(2)},2}}\cr
&b_2=t_{\infty^{(2)},0}\cr
&X_1=t_{\infty^{(2)},1}\cr
&t_{\infty^{(2)},0}(t_{\infty^{(1)},2}-t_{\infty^{(2)},2})(t_{\infty^{(3)},2}-t_{\infty^{(2)},2})
C_0=\left(p_0^3-P_1(q_0)p_0^2+P_2(q_0)p_0-P_3(q_0)\right)\cr
&- t_{\infty^{(1)},1}t_{\infty^{(2)},1}t_{\infty^{(3)},1} +t_{\infty^{(1)},1}t_{\infty^{(2)},0}(t_{\infty^{(2)},2}-t_{\infty^{(3)},2})-t_{\infty^{(2)},1}(t_{\infty^{(1)},2}-t_{\infty^{(3)},2})
\cr
\end{align}

Let us remark that the identification to recover the two-matrix models provides a specific value for $t_{\infty^{(1)},0}$ which does not play any role in the spectral duality. Moreover, it gives the potentials:
\begin{align}\label{PotentialsV1V2}V_1(\lambda)=&\frac{1}{2}(t_{\infty^{(1)},2}-t_{\infty^{(2)},2})\lambda^2 +t_{\infty^{(1)},1}\lambda \cr
V_2(y)=&\frac{y^2}{2(t_{\infty^{(3)},2}-t_{\infty^{(2)},2})} -\frac{ t_{\infty^{(3)},1} y }{t_{\infty^{(3)},2}-t_{\infty^{(2)},2}}+t_{\infty^{(2)},0}\ln(y-t_{\infty^{(2)},1})
\end{align}

We may summarize the results in the following proposition.

\begin{proposition}\label{Prop2MMnew}The two-matrix model with potentials $V_1$ and $V_2$ given by \eqref{PotentialsV1V2} has the same classical spectral curve as the one given by $\hat{L}_d(\lambda)$ for $t_{\infty^{(1)},0}=-1$. 
\end{proposition}

\section{The dual side: Isomonodromic deformations of $\mathfrak{gl}_2(\mathbb{C})$ meromorphic connections associated with the Painlev\'{e} IV equation}\label{SecP4JM}
In this section, we review the specific case of isomonodromic deformations in $\mathfrak{gl}_2(\mathbb{C})$ associated with the Painlev\'{e} IV equation as historically obtained by Jimbo-Miwa-Ueno \cite{JimboMiwa,JimboMiwaUeno}. The general theory and formulas for isomonodromic deformations associated with meromorphic connections in $\mathfrak{gl}_2(\mathbb{C})$ have been developed in \cite{MarchalAlameddineP1Hierarchy2023,marchal2023hamiltonian} and we shall refer to them for technical details. For the sake of the present paper, we shall only mention the results that are useful for the duality.

\subsection{Lax matrices in the oper and geometric gauge}
The Painlev\'{e} IV equation corresponds to isomonodromic deformations in $\mathfrak{gl}_2(\mathbb{C})$ with one irregular pole at infinity of order $r_\infty=3$ and one regular pole at $\xi=X_1$. In other words:
\beq \hbar \partial_\xi \td{\Psi}_{\text{P4}}(\xi)=\td{L}_{\text{P4}}(\xi)\td{\Psi}_{\text{P4}}(\xi)\,\, \text{ with }\,\,  \td{L}_{\text{P4}}(\xi)= \frac{\td{L}_{\text{P4}}^{[X_1,0]}}{\xi-X_1}+\td{L}_{\text{P4}}^{[\infty,1]}\xi+ \td{L}_{\text{P4}}^{[\infty,0]}\eeq
As explained in \cite{marchal2023hamiltonian} one may always choose a representative such that $\td{L}_{\text{P4}}^{[\infty,1]}$ is diagonal and $\left[\td{L}_{\text{P4}}^{[\infty,0]}\right]_{1,2}=1$. Using the local diagonalization at each pole, one defines irregular times $\mathbf{s}:=\left(s_{\infty^{(1)},2},s_{\infty^{(2)},2},s_{\infty^{(1)},1},s_{\infty^{(2)},1}\right)$ and monodromies $\mathbf{s}_0:=\left(s_{\infty^{(1)},0},s_{\infty^{(2)},0},s_{X_1^{(1)},0},s_{X_1^{(2)},0}\right)$ such that $s_{\infty^{(1)},0}+s_{\infty^{(2)},0}+s_{X_1^{(1)},0}+s_{X_1^{(2)},0}=0$. The base is defined as follows
\beq \mathbb{B}_{\text{P4}}:=\left\{ \mathbf{s}\cup \{X_1\} \in \mathbb{C}^5 \,\text{ such that } \, s_{\infty^{(1)},2}\neq s_{\infty^{(2)},2} \right\} \eeq
and we shall always consider irregular times belonging to $\mathbb{B}_{\text{P4}}$. General isomonodromic deformations in this case correspond to
\beq \mathcal{L}_{\boldsymbol{\beta}}:=\hbar\left( \beta_{\infty^{(1)},2} \partial_{s_{\infty^{(1)},2}} + \beta_{\infty^{(2)},2} \partial_{s_{\infty^{(2)},2}}+\beta_{\infty^{(1)},1} \partial_{s_{\infty^{(1)},1}}+\beta_{\infty^{(2)},1} \partial_{s_{\infty^{(2)},1}}+ \beta_{X_1}\partial_{X_1}\right):=\hbar \partial_{\boldsymbol{\beta}} \eeq
where $\boldsymbol{\beta}:=\left( \beta_{\infty^{(1)},2},\beta_{\infty^{(2)},2},\beta_{\infty^{(1)},1},\beta_{\infty^{(2)},1},\beta_{X_1}\right)\in \mathbb{C}^5$. 

In order to avoid confusion with the previous section, we shall denote $(Q,P)$ the Darboux coordinates corresponding to apparent singularity and its dual partner on the spectral curve in this setup. Moreover, we shall define $\left(\mathbf{u}_{\infty^{(i)},j}\right)_{1\leq i,j\leq 2}$ the deformation vectors corresponding to $\mathcal{L}_{\mathbf{u}_{\infty^{(i)},j}}=\hbar \partial_{s_{\infty^{(i)},j}}$ for all $(i,j)\in\llbracket 1,2\rrbracket^2$. We also define:
\begin{align}\label{DefR1R2}
    R_1(\xi):=&\frac{s_{X_1^{(1)},0} +s_{X_1^{(2)},0}}{\xi-X_1}-(s_{\infty^{(1)},2}+s_{\infty^{(2)},2})\xi -s_{\infty^{(1)},1}-s_{\infty^{(2)},1}\cr
    R_2(\xi):=& \frac{s_{X_1^{(1)},0}s_{X_1^{(2)},0}}{(\xi-X_1)^2}+s_{\infty^{(1)},2}s_{\infty^{(2)},2}\xi^2+(s_{\infty^{(2)},2}s_{\infty^{(1)},1}+s_{\infty^{(1)},2}s_{\infty^{(2)},1})\xi\cr
&+s_{\infty^{(1)},0}s_{\infty^{(2)},2}+s_{\infty^{(1)},1}s_{\infty^{(2)},1}+s_{\infty^{(1)},2}s_{\infty^{(2)},0}
\end{align}
Results of \cite{marchal2023hamiltonian} indicate that the associated Lax matrix in the oper gauge is given by
\footnotesize{\begin{align}
    L_{\text{P4}}(\xi)=\begin{pmatrix}
        0&1\\
-\frac{\hbar P }{\xi -Q}-R_2(\xi)-\hbar s_{\infty^{(1)},2}+\frac{(Q-X_1)(P^2-R_1(Q)P+\frac{\hbar}{Q-X_1}P+R_2(Q)+\hbar s_{\infty^{(1)},2})}{\xi-X_1} & \frac{\hbar}{\xi-Q} +R_1(\xi)-\frac{\hbar}{\xi-X_1} \end{pmatrix}
\end{align}}
\normalsize{or} equivalently in the initial geometric gauge
\small{\begin{align}
  \left[\td{L}_{\text{P4}}(\xi)\right]_{1,1}=&\frac{(Q-X_1)(s_{\infty^{(1)},2}Q+P+s_{\infty^{(1)},1})}{\xi-X_1} -s_{\infty^{(1)},2}\xi-s_{\infty^{(1)},1}\cr
  \left[\td{L}_{\text{P4}}(\xi)\right]_{1,2}=&1+\frac{X_1-Q}{\xi-X_1}\cr
  \left[\td{L}_{\text{P4}}(\xi)\right]_{2,1}=&\frac{1}{(Q-X_1)(\xi-X_1)} \left(s_{\infty^{(1)},2}Q^2+(s_{\infty^{(1)},1}+P-X_1 s_{\infty^{(1)},2})Q -(P+s_{\infty^{(1)},1})X_1-s_{X_1^{(1)},0}\right)\cr& \left(s_{\infty^{(1)},2}Q^2+(s_{\infty^{(1)},1}+P-X_1 s_{\infty^{(1)},2})Q -(P+s_{\infty^{(1)},1})X_1-s_{X_1^{(2)},0}\right)\cr&
+(s_{\infty^{(1)},2}-s_{\infty^{(2)},2})(Q-X_1)P+(s_{\infty^{(1)},2}-s_{\infty^{(2)},2})s_{\infty^{(1)},2}Q^2\cr&
-(s_{\infty^{(1)},2}-s_{\infty^{(2)},2})(s_{\infty^{(1)},2}X_1-s_{\infty^{(1)},1})Q
-(s_{\infty^{(1)},1}X_1-s_{\infty^{(1)},0})(s_{\infty^{(1)},2}-s_{\infty^{(2)},2})\cr
\left[\td{L}_{\text{P4}}(\xi)\right]_{2,2}=&\frac{-s_{\infty^{(1)},2}Q^2+(s_{\infty^{(1)},2}X_1-P-s_{\infty^{(1)},1})Q+(s_{\infty^{(1)},1}+P)X_1+s_{X_1^{(1)},0}+s_{X_1^{(2)},0}}{\xi-X_1}\cr&
-s_{\infty^{(2)},2}\xi -s_{\infty^{(2)},1}
\end{align}}
\normalsize{}

\subsection{Isomonodromic deformations and auxiliary matrices}\label{SecGl2Iso}
In this setup, there are $4$ directions leading to trivial deformations (in the sense linear in $(Q,P)$). As in the previous $\mathfrak{gl}_3(\mathbb{C})$ case, one may perform a shift of the Darboux coordinates to have vanishing flows in these directions. It turns out that in $\mathfrak{gl}_2(\mathbb{C})$, the trivial directions only correspond to setting the trace of the Lax matrices to zero and use dilatation/translation on $\xi$. More precisely, the trivial directions are
\begin{align}\label{TrivialVectorsP4}
    \mathcal{L}_{\boldsymbol{\beta}_{\infty,2}}:=&\hbar (\partial_{s_{\infty^{(1)},2}}  +\partial_{s_{\infty^{(2)},2}})\cr
\mathcal{L}_{\boldsymbol{\beta}_{\infty,1}}:=&\hbar (\partial_{s_{\infty^{(1)},1}}  +\partial_{s_{\infty^{(2)},1}})\cr
\mathcal{L}_{\boldsymbol{\beta}_{\text{dil}}}:=&\hbar(
2 s_{\infty^{(1)},2} \partial_{s_{\infty^{(1)},2}}+2 s_{\infty^{(2)},2} \partial_{s_{\infty^{(2)},2}}+s_{\infty^{(1)},1} \partial_{s_{\infty^{(1)},1}}+ s_{\infty^{(2)},1} \partial_{s_{\infty^{(2)},1}}- X_1\partial_{X_1}  )\cr
\mathcal{L}_{\boldsymbol{\beta}_{\text{transl}}}:=&\hbar(s_{\infty^{(1)},2}\partial_{s_{\infty^{(1)},1}}+s_{\infty^{(2)},2}\partial_{s_{\infty^{(2)},1}}-\partial_{X_1})
\end{align} 
for which we find
\begin{align}\label{TrivialEvolutionsP4}\mathcal{L}_{\boldsymbol{\beta}_{\infty,2}}[Q]=&0 \,,\,\mathcal{L}_{\boldsymbol{\beta}_{\infty,2}}[P] =-\hbar Q\cr
\mathcal{L}_{\boldsymbol{\beta}_{\infty,1}}[Q]=&0 \,,\,\mathcal{L}_{\boldsymbol{\beta}_{\infty,1}}[P] =-\hbar \cr
\mathcal{L}_{\boldsymbol{\beta}_{\text{dil}}}[Q]=&-\hbar Q \,,\,\mathcal{L}_{\boldsymbol{\beta}_{\text{dil}}}[P] =\hbar P\cr
 \mathcal{L}_{\boldsymbol{\beta}_{\text{transl}}}[Q]=&-\hbar  \,,\,\mathcal{L}_{\boldsymbol{\beta}_{\text{transl}}}[P] =0   
\end{align}
More generally, the evolutions are provided by the following theorem of \cite{marchal2023hamiltonian}.

\begin{proposition}[General Hamiltonian evolutions (See \cite{marchal2023hamiltonian})]\label{TheoHamP4} The evolutions of the Darboux coordinates $(Q,P)$ are Hamiltonian and the Hamiltonian for a general isomonodromic deformation is given by
\begin{align}\label{HamilP4General} &\text{Ham}^{(\text{P4})}_{(\boldsymbol{\beta})}(Q,P;\hbar):=\mu^{(\boldsymbol{\beta})}\left(P^2-R_1(Q)P+\frac{\hbar P}{Q-X_1}+R_2(Q)+\hbar s_{\infty^{(1)},2}\right)\cr&
-\hbar \nu_0^{(\boldsymbol{\beta})}P -\hbar\nu_{-1}^{(\boldsymbol{\beta})} QP -\hbar c_{1}^{(\boldsymbol{\beta})}Q-\hbar c_{2}^{(\boldsymbol{\beta})} Q^2+ \sum_{i=1}^2\sum_{k=1}^2 \beta_{\infty^{(i)},k} u_{\infty^{(i)},k}(\mathbf{s}) +\beta_{X_1}u_{X_1}(\mathbf{s}) 
\end{align}
where
\begin{align}  \nu_{-1}^{(\boldsymbol{\beta})}:=&\frac{\beta_{\infty^{(1)},2}-\beta_{\infty^{(2)},2} }{2(s_{\infty^{(1)},2}- s_{\infty^{(2)},2})} \cr
\nu_{0}^{(\boldsymbol{\beta})}:=&\frac{2(\beta_{\infty^{(1)},1}-\beta_{\infty^{(2)},1}) (s_{\infty^{(1)},2}-s_{\infty^{(2)},2})-(\beta_{\infty^{(1)},2}-\beta_{\infty^{(2)},2})(s_{\infty^{(1)},1}-s_{\infty^{(2)},1})}{2(s_{\infty^{(1)},2}- s_{\infty^{(2)},2})^2}\cr
c_{1}^{(\boldsymbol{\beta})}:=&\frac{(s_{\infty^{(1)},2}s_{\infty^{(2)},1}-s_{\infty^{(1)},1}s_{\infty^{(2)},2})(\beta_{\infty^{(1)},2}-\beta_{\infty^{(2)},2})}{2(s_{\infty^{(1)},2}- s_{\infty^{(2)},2})^2 }+\frac{\beta_{\infty^{(1)},1}s_{\infty^{(2)},2}- \beta_{\infty^{(2)},1}s_{\infty^{(1)},2}}{s_{\infty^{(1)},2}- s_{\infty^{(2)},2}}\cr 
c_{2}^{(\boldsymbol{\beta})}:=& \frac{\beta_{\infty^{(1)},2}s_{\infty^{(2)},2}-\beta_{\infty^{(2)},2}s_{\infty^{(1)},2} }{2(s_{\infty^{(1)},2}- s_{\infty^{(2)},2})}\cr
\mu^{(\boldsymbol{\beta})}:=&(Q-X_1)\bigg(\beta_{X_1}+\frac{\beta_{\infty^{(1)},1}-\beta_{\infty^{(2)},1}}{s_{\infty^{(1)},2}- s_{\infty^{(2)},2}}\cr&
+\frac{(X_1(s_{\infty^{(1)},2}-s_{\infty^{(2)},2})+s_{\infty^{(2)},1}-s_{\infty^{(1)},1})(\beta_{\infty^{(1)},2}-\beta_{\infty^{(2)},2}) }{2(s_{\infty^{(1)},2}- s_{\infty^{(2)},2})^2}\bigg)
\end{align}
We shall also define the Hamiltonian differential:
\beq\overline{\omega}^{(\text{P4})}(Q,P;\hbar)=\sum_{i=1}^2\sum_{j=1}^2\text{Ham}^{(\text{P4})}_{(\boldsymbol{u}_{\infty^{(i)},j})}(Q,P;\hbar) ds_{\infty^{(i)},j}+ \text{Ham}^{(\text{P4})}_{(\beta_{X_1})}(Q,P;\hbar) dX_1\eeq
\end{proposition}

Note that the purely time-dependent terms $\left(u_{\infty^{(i)},k}(\mathbf{s})\right)_{1\leq i,k\leq 2}$, $u_{X_1}(\mathbf{s})$ 
do not modify the evolutions of $(Q,P)$ and were omitted in \cite{marchal2023hamiltonian}. The dependence of the symplectic fundamental two-form and the Hamiltonian differential on these terms forces their presence, even if they are irrelevant for the dynamics of the Darboux coordinates. They admit the following expressions:
\scriptsize{\begin{align}\label{uterms}
u_{\infty^{(1)},1}(\mathbf{s}):=&\frac{(s_{\infty^{(2)},2}X_1+s_{\infty^{(2)},1})s_{\infty^{(1)},0}+(s_{\infty^{(1)},2}X_1+s_{\infty^{(1)},1})s_{\infty^{(2)},0}}{s_{\infty^{(1)},2}-s_{\infty^{(2)},2}}+ \frac{s_{\infty^{(2)},2} X_1 \hbar}{s_{\infty^{(1)},2}-s_{\infty^{(2)},2}}\cr
u_{\infty^{(2)},1}(\mathbf{s}):=&-\frac{(s_{\infty^{(2)},2}X_1+s_{\infty^{(2)},1})s_{\infty^{(1)},0}+(s_{\infty^{(1)},2}X_1+s_{\infty^{(1)},1})s_{\infty^{(2)},0}}{s_{\infty^{(1)},2}-s_{\infty^{(2)},2}}- \frac{s_{\infty^{(1)},2} X_1 \hbar}{s_{\infty^{(1)},2}-s_{\infty^{(2)},2}}\cr
u_{\infty^{(1)},2}(\mathbf{s}):=&\frac{(X_1(s_{\infty^{(1)},2}-s_{\infty^{(2)},2})-s_{\infty^{(1)},1}+s_{\infty^{(2)},1})((s_{\infty^{(1)},0}s_{\infty^{(2)},2}+s_{\infty^{(1)},2}s_{\infty^{(2)},0})X_1+s_{\infty^{(2)},1}s_{\infty^{(1)},0}+s_{\infty^{(1)},1}s_{\infty^{(2)},0})}{2(s_{\infty^{(1)},2}-s_{\infty^{(2)},2})^2}\cr&
+\frac{(-X_1(s_{\infty^{(2)},2})^2+(X_1s_{\infty^{(1)},2}-s_{\infty^{(1)},1})s_{\infty^{(2)},2}+s_{\infty^{(1)},2}s_{\infty^{(2)},1})X_1\hbar}{2(s_{\infty^{(1)},2}-s_{\infty^{(2)},2})^2}\cr
u_{\infty^{(2)},2}(\mathbf{s}):=&-\frac{(X_1(s_{\infty^{(1)},2}-s_{\infty^{(2)},2})-s_{\infty^{(1)},1}+s_{\infty^{(2)},1})((s_{\infty^{(1)},0}s_{\infty^{(2)},2}+s_{\infty^{(1)},2}s_{\infty^{(2)},0})X_1+s_{\infty^{(2)},1}s_{\infty^{(1)},0}+s_{\infty^{(1)},1}s_{\infty^{(2)},0})}{2(s_{\infty^{(1)},2}-s_{\infty^{(2)},2})^2}\cr&
-\frac{(X_1(s_{\infty^{(1)},2})^2+(-X_1s_{\infty^{(2)},2}+s_{\infty^{(2)},1})s_{\infty^{(1)},2}-s_{\infty^{(1)},1}s_{\infty^{(2)},2})X_1\hbar}{2(s_{\infty^{(1)},2}-s_{\infty^{(2)},2})^2}\cr
u_{X_1}(\mathbf{s}):=&(X_1s_{\infty^{(2)},2}+s_{\infty^{(2)},1})s_{\infty^{(1)},0}+(X_1s_{\infty^{(1)},2}+s_{\infty^{(1)},1})s_{\infty^{(2)},0}
\end{align}}\normalsize{}

\subsection{Shift of Darboux coordinates and reduced system}\label{SecGl2Reduction}
One may shift the Darboux coordinates to remove the trivial evolutions by performing a simple change of coordinates.

\begin{definition}[Shifted Darboux coordinates and change of times]\label{ShiftGl2} One may map $(\mathbf{s},X_1)$ to $\mathbf{S}:=(S_{\infty,1},S_{\infty,2},S_1,S_2,\td{X}_1)$ with
\begin{align}
    S_{\infty,1}:=&s_{\infty^{(1)},1}+s_{\infty^{(2)},1}\,\,,\,\, S_{\infty,2}:=s_{\infty^{(1)},2}+s_{\infty^{(2)},2}\cr
S_1:=& \frac{s_{\infty^{(1)},1}-s_{\infty^{(2)},1}}{\sqrt{2(s_{\infty^{(1)},2}-s_{\infty^{(2)},2})}}\,\,,\,\,
S_2:=\sqrt{\frac{s_{\infty^{(1)},2}-s_{\infty^{(2)},2}}{2}}\cr
\td{X}_1:=&\sqrt{\frac{s_{\infty^{(1)},2}-s_{\infty^{(2)},2}}{2}} X_1+\frac{s_{\infty^{(1)},1}-s_{\infty^{(2)},1}}{\sqrt{2(s_{\infty^{(1)},2}-s_{\infty^{(2)},2})}}
\end{align}
and define the shifted Darboux coordinates $(\check{Q},\check{P})$ by the symplectic change of coordinates
\beq \check{Q}=S_2 Q+S_1 \,\,,\,\, \check{P}:=S_2^{-1}\left(P-\frac{1}{2}R_1(Q)\right)\eeq   
\end{definition}

As proved in \cite{marchal2023hamiltonian}, these shifted Darboux coordinates have only non-trivial evolutions in direction $\partial_{\td{X}_1}$.

\begin{proposition}[Reduced Hamiltonian system for the Painlev\'{e} IV equation]\label{PropReducedP4} The shifted Darboux coordinates $(\check{Q},\check{P})$ have only one non-trivial evolution in direction $\hbar \partial_{\td{X}_1}$. The Hamiltonians are given by:
\begin{align} 
\text{Ham}^{(\text{P4})}_{(\td{X}_1)}(\check{Q},\check{P};\hbar)=&(\check{Q}-\td{X}_1) \check{P}^2 +\hbar \check{P}
-\check{Q}^3+\td{X}_1\check{Q}^2\cr&
-(s_{\infty^{(1)},0}-s_{\infty^{(2)},0} -\hbar)\check{Q}-\frac{(s_{X_1^{(1)},0}-s_{X_1^{(2)},0})^2}{4(\check{Q}-\td{X}_1)}\cr
\text{Ham}^{(\text{P4})}_{(S_{\infty,1})}(\check{Q},\check{P};\hbar)=&\, 0\cr
\text{Ham}^{(\text{P4})}_{(S_{\infty,2})}(\check{Q},\check{P};\hbar)=&\, 0\cr
\text{Ham}^{(\text{P4})}_{(S_{1})}(\check{Q},\check{P};\hbar)=&\, 0\cr
\text{Ham}^{(\text{P4})}_{(S_{2})}(\check{Q},\check{P};\hbar)=& \,0
\end{align}
\end{proposition}

\normalsize{The} last proposition is a consequence of the fact that the fundamental symplectic two-form characterizing the symplectic structure is reduced with this shift.

\begin{proposition}[See \cite{marchal2023hamiltonian}]\label{PropReductionOmega2}The fundamental symplectic two-form of the Jimbo-Miwa Painlev\'{e} IV system is
\beq \Omega^{(\text{P4})}=\hbar dQ\wedge dP - \sum_{i=1}^2\sum_{k=1}^2 ds_{\infty^{(i)},k}\wedge d\text{Ham}_{(\mathbf{u}_{\infty^{(i)},k})}(Q,P;\hbar) -dX_1\wedge d \text{Ham}_{(X_1)}(Q,P;\hbar)\eeq
and reduces to an Arnold-Liouville form
\beq \Omega^{(\text{P4})}=\hbar d\check{Q}\wedge d\check{P}- d\td{X}_1\wedge d\text{Ham}_{(\td{X}_1)}(\check{Q},\check{P};\hbar)\eeq
In particular $(\check{Q},\check{P})$ are independent of $(S_1,S_2,S_{\infty,1},S_{\infty,2})$.
\end{proposition}

\begin{proof}The proof was done in \cite{marchal2023hamiltonian} without the purely time-dependent terms. But it is then a straightforward computation to check that from \eqref{uterms} we have 
\beq \sum_{i=1}^2\sum_{k=1}^2ds_{\infty^{(i)},k}\wedge d u_{\infty^{(i)},k}(\mathbf{s})+dX_1\wedge du_{X_1}(\mathbf{s})=0 \eeq
to complete the proof.    
\end{proof}

\subsection{Jimbo-Miwa-Ueno isomonodromic tau-function}\label{SecGl2JMU}
Similarly to the previous section, one may compute the Jimbo-Miwa-Ueno differential $\omega_{\text{JMU}}^{(\text{P4})}$ associated with $\td{L}_{\text{P4}}(\xi)$. The definition \cite{JimboMiwaUeno} is equivalent to
\begin{align} \label{JMUGL2Def}\omega_{\text{JMU}}^{(\text{P4})}:=&-\Res_{\xi \to \infty}\Tr\Big[ \td{\Psi}_{\text{P4},\infty}^{(\text{reg})}(\xi)^{-1} (\hbar\partial_\xi \td{\Psi}_{\text{P4},\infty}^{(\text{reg})}(\xi)) d\Lambda_\infty(\xi)\Big]\cr&
-\Res_{\xi \to X_1}\Tr\Big[ \td{\Psi}_{\text{P4},X_1}^{(\text{reg})}(\xi)^{-1} (\hbar\partial_\xi \td{\Psi}_{\text{P4},X_1}^{(\text{reg})}(\xi)) d\Lambda_{X_1}(\xi)\Big]
\end{align}
where
\beq\label{FormalAsymptP4}
\td{\Psi}_{\text{P4}}(\xi)\overset{\xi\to \infty}{\sim}\td{\Psi}_{\text{P4},\infty}^{(\text{reg})}(\xi)e^{\frac{1}{\hbar} \Lambda_\infty(\xi)} \,\,\text{ and }\,\, 
\td{\Psi}_{\text{P4}}\overset{\xi\to X_1}{\sim} N_0 \td{\Psi}_{\text{P4},X_1}^{(\text{reg})}(\xi) e^{\frac{1}{\hbar}\Lambda_{X_1}(\xi)}
\eeq
with
\begin{align}
\Lambda_\infty:=&\text{diag}\left(-s_{\infty^{(1)},0}\ln(\xi)-s_{\infty^{(1)},1}\xi-\frac{s_{\infty^{(1)},2}}{2}\xi^2,-s_{\infty^{(2)},0}\ln(\xi)-s_{\infty^{(2)},1}\xi-\frac{s_{\infty^{(2)},2}}{2}\xi^2\right)\cr
\Lambda_{X_1}(\xi):=&\text{diag}\left(s_{X_1^{(1)},0}\ln(\xi-X_1),s_{X_1^{(2)},0}\ln(\xi-X_1)\right)\cr
\td{\Psi}_{\text{P4},\infty}^{(\text{reg})}(\xi):=&I_2+\sum_{k=1}^\infty J_k\xi^{-k}\cr
\td{\Psi}_{\text{P4},X_1}^{(\text{reg})}(\xi) :=&I_2+ \sum_{k=1}^\infty N_k(\xi-X_1)^{k}\cr
N_0:=&\begin{pmatrix}
    Q-X_1&Q-X_1\\
\left[N_0\right]_{2,1}& \left[N_0\right]_{2,2}\end{pmatrix}\cr
\left[N_0\right]_{2,1}:=&P(Q-X_1)+s_{\infty^{(1)},2}Q^2+(s_{\infty^{(1)},1}-s_{\infty^{(1)},2}X_1)Q-s_{\infty^{(1)},1}X_1-s_{X_1^{(1)},0}\cr
\left[N_0\right]_{2,2}:=&P(Q-X_1)+s_{\infty^{(1)},2}Q^2+(s_{\infty^{(1)},1}-s_{\infty^{(1)},2}X_1)Q-s_{\infty^{(1)},1}X_1-s_{X_1^{(2)},0}
\end{align}
Note that $N_0$ is a matrix diagonalizing the leading order of the singular part of $\td{L}(\xi)$ at $\xi=X_1$.

Computing \eqref{JMUGL2Def} provides:
\begin{align}
  \omega_{\text{JMU}}^{(\text{P4})}=&\,\hbar(s_{X_1^{(1)},0} \left[N_1\right]_{1,1} +s_{X_1^{(2)},0} \left[N_1\right]_{2,2})  dX_1 +\hbar \left[J_1\right]_{1,1} d s_{\infty^{(1)},1}+\hbar \left[J_1\right]_{2,2} d s_{\infty^{(2)},1}  \cr&
-\frac{\hbar}{2}\left( \left(\left[J_1\right]_{1,1}\right)^2+ \left[J_1\right]_{1,2}\left[J_1\right]_{2,1}-2\left[J_2\right]_{1,1} \right) ds_{\infty^{(1)},2}\cr&
-\frac{\hbar}{2}\left( \left(\left[J_1\right]_{2,2}\right)^2+ \left[J_1\right]_{1,2}\left[J_1\right]_{2,1}-2\left[J_2\right]_{2,2} \right) ds_{\infty^{(2)},2}
\end{align}
Solving the differential equation with the formal asymptotic expansions \eqref{FormalAsymptP4} provides the first orders required to compute the last formula. We find the following theorem.

\begin{proposition}[Expression of the Jimbo-Miwa-Ueno isomonodromic differential]\label{ThJMUP4} The Jimbo-Miwa-Ueno differential is given by
\begin{align}
\omega_{\text{JMU}}^{(\text{P4})}=& \overline{\omega}^{(\text{P4})}(Q,P;\hbar=0) +dK_0(\mathbf{s})\cr
=&\sum_{i,k=1}^2\text{Ham}_{(\mathbf{u}_{\infty^{(i)},k})}^{(\text{P4})}(Q,P;\hbar=0) ds_{\infty^{(i)},k} +\text{Ham}_{(\mathbf{u}_{X_1})}^{(\text{P4})}(Q,P;\hbar=0) dX_1 +dK_0(\mathbf{s})\cr
=&\text{Ham}_{(\td{X}_1)}^{(\text{P4})}(\check{Q},\check{P};\hbar=0)d\td{X}_1+ dK_0(\mathbf{S})\cr
=&\left(\text{Ham}_{(\td{X}_1)}^{(\text{P4})}(\check{Q},\check{P};\hbar)-\hbar \check{P}-\hbar \check{Q}\right)d\td{X}_1+ dK_0(\mathbf{S}) 
\end{align}
where $K_0$ is given by either
\begin{align}\label{DefKgl2}
K_0(\mathbf{s}):=&\frac{1}{2}(s_{X_1^{(1)},0}+s_{\infty^{(2)},0})(s_{X_1^{(1)},0}+s_{\infty^{(1)},0}) \ln\left(\frac{s_{\infty^{(1)},2}-s_{\infty^{(2)},2}}{2}\right)\cr& +\frac{X_1}{2}\left((X_1s_{\infty^{(1)},2}+2s_{\infty^{(1)},1})s_{\infty^{(1)},0}+s_{\infty^{(2)},0}(X_1s_{\infty^{(2)},2}+2s_{\infty^{(2)},1})\right)
\end{align}
or in terms of the $\mathbf{S}$ coordinates by
\begin{align}\label{KNewCoordinates}&K_0(\mathbf{S})=(s_{X_1^{(1)},0}+s_{\infty^{(2)},0})(s_{X_1^{(1)},0}+s_{\infty^{(1)},0})\ln(S_2)-\frac{(S_1-\td{X_1})(s_{\infty^{(1)},0}+s_{\infty^{(2)},0})S_{\infty,1}}{2S_2}\cr&
+\frac{(S_1-\td{X_1})^2(s_{\infty^{(1)},0}+s_{\infty^{(2)},0})S_{\infty,2}}{4S_2^2}-\frac{1}{2}(S_1-\td{X_1})(s_{\infty^{(1)},0}-s_{\infty^{(2)},0})(S_1+\td{X_1})
\end{align}
In other words, up to the exact term $dK_0$, the Jimbo-Miwa-Ueno differential is given by the Hamiltonian one form evaluated formally at $\hbar=0$.
\end{proposition}

\begin{proof}The proof follows from direct but lengthy computations. A corresponding Maple worksheet is available at \url{http://math.univ-lyon1.fr/~marchal/AdditionalRessources/index.html} for completeness.  
\end{proof}

Similarly to \autoref{PropJMUDifferential}, we find that the formal evaluation at $\hbar=0$ of the Hamiltonian one form provides the Jimbo-Miwa-Ueno differential. In particular, it provides the specific terms (in the reduced case $\hbar (\check{Q}+\check{P})$ in the present example) that one needs to subtract in order to obtain a closed differential from the naive Hamiltonian one-form. We shall develop this observation in  \autoref{SectionOutlooks} and formulate a conjecture. Since the Jimbo-Miwa-Ueno differential is a closed form, we may define the associated Jimbo-Miwa-Ueno isomonodromic tau-function:
\begin{align}\label{JMUtaufunctionP4} &\hbar^2\, d (\ln \tau_{\text{JMU}}^{(\text{P4})}(\mathbf{s})):= \omega_{\text{JMU}}^{(\text{P4})}\cr
&=\sum_{i=1}^2\sum_{k=1}^2 \text{Ham}^{(\text{P4})}_{(\mathbf{u}_{\infty^{(i)},k})}(Q,P;\hbar=0) ds_{\infty^{(i)},k}+ \text{Ham}_{(\mathbf{u}_{X_1})}^{(\text{P4})}(Q,P;\hbar=0) dX_1+dK_0(\mathbf{s})\cr
&=\text{Ham}^{(\text{P4})}_{(\td{X}_1)}(\check{Q},\check{P};\hbar=0)d\td{X}_1+ dK_0(\mathbf{S})
\end{align}
Note that by construction the one-form $\text{Ham}^{(\text{P4})}_{(\td{X}_1)}\left(\check{Q},\check{P};\hbar=0\right)d\td{X}_1$ is also a closed form.
As above, closedness is understood on the deformation base after pullback along an isomonodromic solution.

\subsection{Canonical choice of trivial times}
As discussed in \cite{marchal2023hamiltonian}, there are two canonical choices of trivial times that are used in the literature. The Jimbo-Miwa Lax system corresponds to the choice $s_{\infty^{(1)},2}=-s_{\infty^{(2)},2}=1$, $s_{\infty^{(1)},1}=-s_{\infty^{(2)},1}$ (i.e. restrict to $\mathfrak{sl}_2(\mathbb{C})$) and take $X_1=0$ (fixing the position of the finite pole). It is equivalent to $S_{\infty,1}=S_{\infty,2}=0$, $S_2=1$ and $\td{X}_1=S_1$.

\begin{definition}\label{CanonicalChoiceP4JM}The canonical choice of times to recover the Jimbo-Miwa Painlev\'{e} IV Lax pair is to take $s_{\infty^{(1)},2}=-s_{\infty^{(2)},2}=1$, $\sigma:=s_{\infty^{(1)},1}=-s_{\infty^{(2)},1}$ and $X_1=0$.
It gives $\check{P}=P-\frac{s_{X_1^{(1)},0}+s_{X_1^{(2)},0}}{2Q}$ and $\check{Q}=Q+\sigma$.
\end{definition}
For this particular choice we obtain the standard Jimbo-Miwa Lax matrices.

\begin{proposition}\label{ReducedEvolutionP4}Under the canonical choice of trivial times of \autoref{CanonicalChoiceP4JM}, the only non-trivial evolution is relatively to $\sigma:=s_{\infty^{(1)},1}$ and its Hamiltonian (which is independent of the choice of trivial times) is given by
\begin{align} \text{Ham}^{(\text{P4})}(\check{Q},\check{P})=&(\check{Q}-\sigma) \check{P}^2 +\hbar \check{P}
-\check{Q}^3+\sigma\check{Q}^2
-(s_{\infty^{(1)},0}-s_{\infty^{(2)},0} -\hbar)\check{Q}
-\frac{(s_{X_1^{(1)},0}-s_{X_1^{(2)},0})^2
}{4(\check{Q}-\sigma)}  \cr&
\end{align}

It corresponds to the Painlev\'{e} IV equation for $\check{Q}$ (See \cite{marchal2023hamiltonian} for details). The corresponding Lax pairs denoted $\left(\td{L}^{(\text{P4})}_0(\xi),\td{A}^{(\text{P4})}_0(\xi)\right)$ are
\small{\bea\left[\td{L}^{(\text{P4})}_0(\xi)\right]_{1,1}&=&-\xi-\sigma+\frac{\check{Q}(\check{P}+\check{Q}-\sigma) -\sigma \check{P}+\frac{1}{2}\left(s_{X_1^{(1)},0}+s_{X_1^{(2)},0}\right)}{\xi}\cr
\left[\td{L}^{(\text{P4})}_0(\xi)\right]_{1,2}&=&1+\frac{\sigma-\check{Q}}{\xi}\cr
\left[\td{L}^{(\text{P4})}_0(\xi)\right]_{2,1}&=&2(\check{Q}-\sigma)(\check{Q}+\check{P})+s_{\infty^{(1)},0}-s_{\infty^{(2)},0}+\frac{(\check{Q}+\check{P})^2(\check{Q}-\sigma)-\frac{\left(s_{X_1^{(1)},0}-s_{X_1^{(2)},0}\right)^2}{4(\check{Q}-\sigma)}}{\xi}\cr
[\td{L}^{(\text{P4})}_0(\xi)]_{2,2}&=&\xi+\sigma+\frac{-\check{Q}(\check{P}+\check{Q}-\sigma) +\sigma \check{P}+\frac{1}{2}\left(s_{X_1^{(1)},0}+s_{X_1^{(2)},0}\right)}{\xi}\cr
\td{A}^{(\text{P4})}_0(\xi)&=&\begin{pmatrix}-\xi-\check{Q}& 1\\ 2(\check{Q}-\sigma)(\check{P}+\check{Q})+s_{\infty^{(1)},0}-s_{\infty^{(2)},0}
&\xi+\check{Q}
\end{pmatrix}
\eea}
\normalsize{or} equivalently in the oper gauge:
\scriptsize{\begin{align} L^{(\text{P4})}_0(\xi)=&\begin{pmatrix}
    0&1\\ -\frac{\hbar P}{\xi-Q}-R_2(\xi)-\hbar +\frac{QP^2+Q R_2(Q) +\hbar (Q+P)-(s_{X_1^{(1)},0}+s_{X_1^{(2)},0})P }{\xi}& \frac{\hbar}{\xi-Q}+\frac{s_{X_1^{(1)},0}+s_{X_1^{(2)},0}- \hbar}{\xi}\end{pmatrix}\cr
   A^{(\text{P4})}_0(\xi)=&\begin{pmatrix}-\frac{PQ}{\xi-Q}& 
   1+\frac{Q}{\xi-Q}\\
\frac{P(PQ-s_{X_1^{(1)},0}-s_{X_1^{(2)},0})}{\xi-Q}+ \frac{s_{X_1^{(1)},0}s_{X_1^{(2)},0}}{Q \xi}+\xi^2+(Q+2\sigma)\xi+(Q+ \sigma)^2 +s_{\infty^{(1)},0}-s_{\infty^{(2)},0}-\hbar  
&\frac{s_{X_1^{(1)},0}+s_{X_1^{(2)},0}-PQ}{\xi-Q}
   \end{pmatrix} \cr&&
\end{align}} 
\normalsize{with} 
\begin{align}
R_1(\xi)=&\frac{s_{X_1^{(1)},0}+s_{X_1^{(2)},0}}{\xi}\cr
    R_2(\xi)=& \frac{s_{X_1^{(1)},0}s_{X_1^{(2)},0}  }{\xi^2}-(\xi+\sigma)^2 -s_{\infty^{(1)},0}+s_{\infty^{(2)},0}
\end{align} 
\end{proposition}

\subsection{A Hermitian one-matrix model for the $\mathfrak{gl}_2$ side and relation with the JMU tau-function}\label{SecGl21MM}
Let us first define the spectral and classical spectral curve associated with $\td{L}_{\text{P4}}(\xi)$.

\begin{definition}[Spectral curves and classical spectral curves associated with $\td{L}_{\text{P4}}(\xi)$]\label{DefSepctralCurvesGl2 }We shall define
\beq \mathcal{S}_{\text{P4}}=\left\{(\xi,Y)\in \overline{\mathbb{C}}\times \overline{\mathbb{C}} \,\,\text{ such that }\,\, \det(Y I_2-\td{L}_{\text{P4}}(\xi))=0\right\}
\eeq
as the spectral curve associated with the Lax matrix $\td{L}_{\text{P4}}(\xi)$.
Moreover, we shall define
\beq \mathcal{S}_{\text{P4},0}=\left\{(\xi,Y)\in \overline{\mathbb{C}}\times \overline{\mathbb{C}} \,\,\text{ such that }\,\, \underset{\hbar\to 0}{\lim}\det(Y I_2-\td{L}_{\text{P4}}(\xi))=0\right\}
\eeq
as the classical spectral curve associated with the Lax matrix $\td{L}_{\text{P4}}(\xi) \in \mathfrak{gl}_2(\mathbb{C})$
where the limit as $\hbar\to 0$ is to be understood as a formal evaluation at $\hbar=0$ for fixed Darboux coordinates. 
\end{definition}

In the present Painlev\'{e} IV case, we have
\beq \label{ExplicitSpecCurvegl2}0=Y^2-R_1(\xi)Y+R_2(\xi)-\frac{Q-X_1}{\xi-X_1}\left(P^2-R_1(Q)P+R_2(Q)\right)\eeq

\medskip

We shall now produce a one-matrix model with $\hbar=N^{-1}$ whose classical spectral curve recovers the one given by \eqref{ExplicitSpecCurvegl2} when  $s_{X_1^{(1)},0}s_{X_1^{(2)},0}=0$. Indeed, let us first reduce the spectral curve to its standard hyper-elliptic form by setting $Y_0:=Y-\frac{1}{2}R_1(\xi)$. We get that the spectral curve is reduced to
\beq \label{Y00}Y_0^2= \frac{1}{4}R_1(\xi)^2-R_2(\xi)+\frac{U}{\xi-X_1}\,\,\text{ with } U:=(Q-X_1)\left(P^2-R_1(Q)P+R_2(Q)\right)\eeq
Let us now look at the $N\times N$ Hermitian matrix model given by
\beq d\mu_N(M)\propto e^{-N\Tr V(M)} dM\eeq
where the potential $V$ is given by
\beq V(\xi):=\frac{a}{2}\xi^2+b\xi +c\ln(\xi-X_1)\eeq
We shall denote $Z^{(\text1MM)}(\mathbf{s};N)$ the associated partition function.
Defining $y^{(\text{1MM})}_0(\xi)=W_{1}^{(0)}(\xi) -\frac{1}{2}V'(\xi)$ the standard shifted one-point correlation function with $W_1^{(0)}(\xi)=\left<\frac{1}{N}\underset{i=1}{\overset{N}{\sum}}\frac{1}{\xi-\lambda_i}\right>^{(0)}$ where the superscript $\,^{(0)}$ indicates that we take the large $N$ limit, we obtain the classical spectral curve given by \cite{AKEMANN1997475,AMBJORN1993127,E1MM}:
\beq \label{Y01} y_0^{(\text{1MM})}(\xi)^2=\frac{1}{4}V'(\xi)^2-\left<\frac{1}{N}\sum_{i=1}^N \frac{V'(\xi)-V'(\lambda_i)}{\xi-\lambda_i}\right>^{(0)}\eeq
The general form of $V'(\xi)=a\xi+b+\frac{c}{\xi-X_1}$ implies that the last term is of the form $a-\frac{c\,c_0}{\xi-X_1}$ where $c_0=\left<\frac{1}{N}\underset{i=1}{\overset{N}{\sum}}\frac{1}{\lambda_i-X_1}\right>^{(0)}$. In the case $s_{X_1^{(1)},0}s_{X_1^{(2)},0}=0$, we have 
\begin{align}
    R_1(\xi)=&-(s_{\infty^{(1)},2}+s_{\infty^{(2)},2})\xi-(s_{\infty^{(1)},1}+s_{\infty^{(2)},1})-\frac{s_{\infty^{(1)},0}+s_{\infty^{(2)},0}}{\xi-X_1}\cr
    R_2(\xi)=&s_{\infty^{(1)},2}s_{\infty^{(2)},2}\xi^2+(s_{\infty^{(1)},1}s_{\infty^{(2)},2}+s_{\infty^{(1)},2}s_{\infty^{(2)},1})\xi\cr&
+s_{\infty^{(1)},0}s_{\infty^{(2)},2}+s_{\infty^{(1)},1}s_{\infty^{(2)},1}+s_{\infty^{(1)},2}s_{\infty^{(2)},0}
\end{align}
so that an immediate computation shows that one may identify the classical spectral curve associated with \eqref{Y00} and the classical spectral curve given by \eqref{Y01} by taking
\begin{align} 
s_{\infty^{(2)},0}=&-1\cr
a=&s_{\infty^{(2)},2}-s_{\infty^{(1)},2}\cr
b=&s_{\infty^{(2)},1}-s_{\infty^{(1)},1}\cr
c=&-(s_{\infty^{(1)},0}+s_{\infty^{(2)},0})\cr
U_0:=&(Q_0-X_1)\left(P_0^2-R_1(Q_0)P_0+R_2(Q_0)\right)=c\,c_0-(s_{\infty^{(1)},0}+s_{\infty^{(2)},0})(X_1s_{\infty^{(2)},2}+s_{\infty^{(2)},1}) \cr
\end{align}
In other words, we may identify the classical spectral curve associated with $\td{L}_{\text{P4}}(\xi)$ with $s_{X_1^{(1)},0}s_{X_1^{(2)},0}=0$ and $s_{\infty^{(2)},0}=-1$ with a Hermitian one-matrix model with potential
\beq V'(\xi)=(s_{\infty^{(2)},2}-s_{\infty^{(1)},2})\xi+s_{\infty^{(2)},1}-s_{\infty^{(1)},1}+\frac{1-s_{\infty^{(1)},0}}{\xi-X_1}\eeq
In a formal Hermitian matrix model, the filling fraction fixes the remaining spectral-curve parameter $c_0$; it is not fixed by the loop equation alone. Matching with the isomonodromic curve requires the relation for $U_0$ above. On a regular formal branch with $Q_0\neq X_1$, the leading Hamiltonian $H_{X_1}^{(0)}=(Q-X_1)K(Q,P)$, where $K(Q,P)=P^2-R_1(Q)P+R_2(Q)$, gives the stationary equations
\begin{align}
0=&2P_0-R_1(Q_0),\cr
0=&K(Q_0,P_0)+(Q_0-X_1)(-R_1'(Q_0)P_0+R_2'(Q_0)).
\end{align}

This can be summarized into the following proposition.

\begin{proposition}\label{PartitionFunctionP4new}The classical spectral curve associated with the Hermitian one-matrix model with potential $V(\xi)=\frac{1}{2}(s_{\infty^{(2)},2}-s_{\infty^{(1)},2})\xi^2+(s_{\infty^{(2)},1}-s_{\infty^{(1)},1})\xi+(1-s_{\infty^{(1)},0})\ln(\xi-X_1)$ matches with the classical spectral curve associated with $\td{L}_{\text{P4}}(\xi)$ for $s_{X_1^{(1)},0}s_{X_1^{(2)},0}=0$ and $s_{\infty^{(2)},0}=-1$.
The filling fraction is chosen so that $c_0$ satisfies the coefficient-matching relation above; the identification is an identity of formal spectral curves.
\end{proposition}

Note that the conditions to identify the matrix model only imply conditions on the monodromies that do not change isomonodromic deformations.

\section{Duality at different levels}\label{SecDuality}

In this section, we propose several perspectives to understand the duality of the two isomonodromic deformations of the meromorphic connections introduced in \autoref{secGl3} and  \autoref{SecP4JM}. Note that the duality is non-trivial because both the ranks of the connections differ ($\mathfrak{gl}_3(\mathbb{C})$ versus $\mathfrak{gl}_2(\mathbb{C})$) and the number of irregular times and monodromies: there are $6$ irregular times $\mathbf{t}=\left(t_{\infty^{(i)},k}\right)_{1\leq i\leq 3,1\leq k\leq 2}$ and $2$ independent monodromies $(t_{\infty^{(1)},0},t_{\infty^{(2)},0})$ in the $\mathfrak{gl}_3(\mathbb{C})$ setting of \autoref{secGl3} while we have $4$ irregular times $\left(s_{\infty^{(i)},k}\right)_{1\leq i,k\leq 2}$, one position of a pole $X_1$ and $3$ independent monodromies $(s_{\infty^{(1)},0},s_{\infty^{(2)},0},s_{X_1^{(1)},0})$ in the $\mathfrak{gl}_2(\mathbb{C})$ setting of \autoref{SecP4JM}. We shall propose several identifications of the duality. 

First, at the level of the reduced Hamiltonian systems and the shifted Darboux coordinates, the identification is simpler because there is only one non-trivial evolution. 

Then, the duality can be observed directly at the level of (classical) spectral curves in the exchange $\lambda \leftrightarrow y$. This gives rise to dual classical spectral curves (that are independent of the choice of gauge) of genus $1$ that can be identified up to terms proportional to $\hbar$ that disappear in the classical limit. These terms can be obtained at the level of spectral curves with the main limitation that these terms depend on the gauge choice. As we shall see, the identification of the duality at the level of spectral curves is possible but requires to turn one of the two Lax matrices in a different gauge where one eigenvalue of the diagonal leading coefficient at infinity is shifted to zero. 

Finally, one can try a dual identification at the level of quantum curves, i.e. at the level of the wave functions using the Chekhov-Eynard-Orantin topological recursion \cite{EO07} and its $x-y$ symmetry as well as the recent works \cite{Quantization_2021} building the wave functions and the Lax matrices starting from the classical spectral curves.

Let us also mention that all these perspectives should be part of some Harnad's duality picture \cite{Harnad_1994,woodhouse2007duality} that we hope to unveil for any untwisted meromorphic connections.

\subsection{Duality at the level of reduced Hamiltonian systems}
The first natural identification on both sides can be made from the Hamiltonian evolutions of $(\check{q},\check{p})$ and $(\check{Q},\check{P})$. Indeed, in both cases the evolutions of these coordinates are trivial except for one direction, namely $\hbar \partial_\tau$ for $(\check{q},\check{p})$ and $\hbar \partial_{\td{X}_1}$ for $(\check{Q},\check{P})$. The identification is given by the following theorem.

\begin{theorem}[Dual correspondence of the reduced Hamiltonian systems]\label{TheoDualReducedHamiltonian} The Hamiltonian system for $(\check{q},\check{p})$ in \autoref{TheoNonTrivialEvolution} can be identified with the Hamiltonian system for $(\check{Q},\check{P})$ in  \autoref{PropReducedP4} under the correspondence:
\begin{align}
    \tau=&\sqrt{2}\td{X}_1\cr
    \check{q}=&-\sqrt{2}\left(\check{Q}-\td{X}_1\right)\cr
    \check{p}=&-\frac{1}{\sqrt{2}}\left(\check{P}-\check{Q}+\frac{t_{\infty^{(2)},0}}{2(\check{Q}-\td{X}_1)}\right)\cr
    \left(s_{X_1^{(1)},0}-s_{X_1^{(2)},0}\right)^2=&\left(t_{\infty^{(2)},0}\right)^2\cr
    s_{\infty^{(1)},0}-s_{\infty^{(2)},0}=&-t_{\infty^{(2)},0}-2t_{\infty^{(1)},0}
\end{align}
which is equivalent to
\begin{align}\td{X}_1=&\frac{\tau}{\sqrt{2}}\cr
    \check{Q}=&-\frac{1}{\sqrt{2}}(\check{q}-\tau)\cr
    \check{P}=&-\sqrt{2}\left(\check{p}+\frac{1}{2}\check{q}-\frac{\tau}{2}-\frac{t_{\infty^{(2)},0}}{2\check{q}}\right)\cr
    \left(t_{\infty^{(2)},0}\right)^2=& \left(s_{X_1^{(1)},0}-s_{X_1^{(2)},0}\right)^2\cr
    t_{\infty^{(2)},0}+2t_{\infty^{(1)},0}=& -s_{\infty^{(1)},0}+s_{\infty^{(2)},0}
\end{align}
\end{theorem}

\begin{proof}The proof follows from straightforward computations. Note however that the identification between $(\check{q},\check{p})$ and $(\check{Q},\check{P})$ is time-dependent so that one cannot simply replace the Darboux coordinates in the Hamiltonians but one needs to track the flows $\hbar\partial_\tau \check{q}$ and $\hbar\partial_\tau \check{p}$ to get $\hbar \partial_\tau \check{Q}$ and $\hbar \partial_\tau \check{P}$ and then $\hbar \partial_{\td{X}_1}\check{Q}$ and $\hbar \partial_{\td{X}_1}\check{P}$ and match both sides.
\end{proof}

\autoref{TheoDualReducedHamiltonian} implies that the isomonodromic deformations of meromorphic connections in $\mathfrak{gl}_3(\mathbb{C})$ with an irregular pole at infinity of order $r_\infty=3$ essentially correspond to the Painlev\'{e} IV Hamiltonian systems. In particular, the Lax pair of \autoref{PropLaxMatricesReduced} provides a Lax representation in $\mathfrak{gl}_3(\mathbb{C})$ of the Painlev\'{e} IV Hamiltonian system. Therefore the Painlev\'{e} IV Hamiltonian system may be seen either as isomonodromic deformations of meromorphic connections in $\mathfrak{gl}_2(\mathbb{C})$ with one irregular pole at infinity of order $r_\infty=3$ and one regular finite pole or isomonodromic deformations of meromorphic connections in $\mathfrak{gl}_3(\mathbb{C})$ with only one irregular pole of order $r_\infty=3$.

\medskip

As we will see in the next section, \autoref{TheoDualReducedHamiltonian} follows from the more general identification of the spectral curves by duality that provides complete identification of all times and monodromies on both sides.

\subsection{Duality at the level of spectral curves}\label{SecDualitySpecCurves}
The second perspective to consider duality is to study the correspondence of the spectral curves under the exchange $\lambda\leftrightarrow y$. This point of view is the standard and historical one for the duality and it is the main subject of this section. As discussed in \autoref{Sectiongl3duality} and because the spectral curves are not invariant by gauge transformations, we need to adapt the gauge between both sides to obtain spectral duality. This was the main motivation to introduce the duality gauge $\hat{L}_d(\lambda)$ in the $\mathfrak{gl}_3(\mathbb{C})$ setting. Let us finally remind that in the $\mathfrak{gl}_2(\mathbb{C})$ case, the spectral curve is explicitly given by:
\beq\label{SpectralCurveGl2} \det(YI_2-\td{L}_{\text{P4}}(\xi))=Y^2-R_1(\xi)Y+R_2(\xi)-\frac{Q-X_1}{\xi-X_1}\left(P^2-R_1(Q)P+R_2(Q)\right)\eeq
while in the $\mathfrak{gl}_3(\mathbb{C})$ case the spectral curve associated with $\hat{L}_d(\lambda)$ is given by
\begin{align} \label{SpectralCurveGl3} \det(yI_3-\hat{L}_{d}(\lambda))=&(y+t_{\infty^{(2)},2}\lambda)^3-P_1(\lambda)(y+t_{\infty^{(2)},2}\lambda)^2+P_2(\lambda)(y+t_{\infty^{(2)},2}\lambda)\cr
&-P_3(\lambda) -\left(p^3-P_1(q)p^2+P_2(q)p-P_3(q)\right)\end{align}
These two spectral curves do not depend explicitly on $\hbar$ and thus one may obtain the classical spectral curves by merely replacing $(q,p,Q,P)$ by their formal limit $(q_0,p_0,Q_0,P_0)$ as $\hbar\to 0$. Both spectral curves are generically of genus $1$. However, it is important to notice that the shift in $y$ performed on the $\mathfrak{gl}_3(\mathbb{C})$ side modifies the Newton polygon by removing the lower-right triangle as presented in \autoref{figurepolygons}. In other words the Newton polygon associated with $\td{L}(\lambda)$ is not the same as the Newton polygon for $\hat{L}_d(\lambda)$.\footnote{The reason is that the point $(3,0)$ in the Newton polygon corresponds to the coefficient of $\lambda^3y^0$ in the polynomial $\det(yI_3-\hat{L}_d(\lambda))$, i.e. to the cubic term in $\det \hat{L}_d(\lambda)$. But this coefficient is obviously vanishing because $\hat{L}_d(\lambda)\overset{\lambda\to \infty}{=}\text{diag}\left(t_{\infty^{(1)},2}-t_{\infty^{(2)},2},0, t_{\infty^{(3)},2}-t_{\infty^{(2)},2}\right)\lambda+O(1)$, whence $\det\hat{L}_d(\lambda)=O(\lambda^2)$.}  In the $\mathfrak{gl}_2(\mathbb{C})$ case, the Newton polygon also changes drastically when $s_{X_1^{(1)},0}s_{X_1^{(2)},0}=0$, i.e. when one monodromy (at least) at $\xi=X_1$ is vanishing. Indeed, in this case, we have that both $\Tr \,\td{L}_{\text{P4}}(\xi)$ and $ \det \td{L}_{\text{P4}}(\xi)$ have only a simple pole at $\xi=X_1$ (otherwise $\det \td{L}_{\text{P4}}(\xi)$ would have a double pole). This is the only possible case for which the spectral curve associated with $\hat{L}_d$ and $\td{L}_{\text{P4}}$ have dual Newton polygons (i.e. symmetric relatively to the $x-y$ swap).

\begin{figure}[H]\centering
\includegraphics[width=\linewidth]{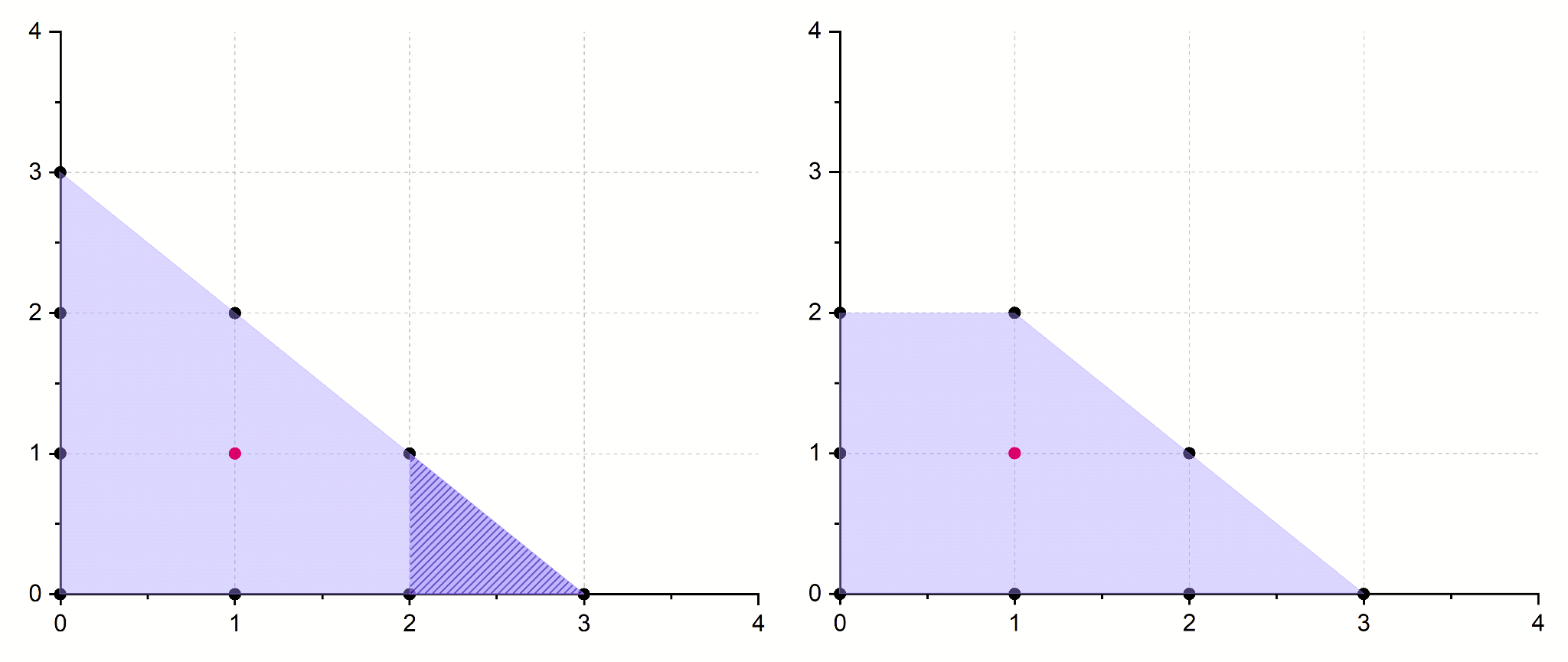}
\caption{\textit{Newton polygons on both sides. The left figure represents the Newton polygons of $\mathcal{S}_d$ (light blue) or $\mathcal{S}$ (light and dark blue). The right figure represents the Newton polygon of $\mathcal{S}_{\text{P4}}$ under the constraint $s_{X_1^{(1)},0}s_{X_1^{(2)},0}=0$.}}
\label{figurepolygons}
\end{figure}

Moreover, the duality gauge $\hat{L}_d(\lambda)$ is chosen so that the following theorem holds.

\begin{theorem}[Duality at the level of spectral curves]\label{TheoDualitySpecCurves} We have
\beq \frac{1}{(t_{\infty^{(3)},2}-t_{\infty^{(2)},2})(t_{\infty^{(1)},2}-t_{\infty^{(2)},2})}\det(\lambda I_3-\hat{L}_d(y))= (\lambda-X_1)\det(y I_2-\td{L}_{\text{P4}}(\lambda))\eeq
upon the identification of the times, Darboux coordinates and monodromies given by
We work on the chart $s_{\infty^{(1)},2}s_{\infty^{(2)},2}\neq0$.
\begin{align}
q=&P\cr
p_d=&Q\,\,\Leftrightarrow\,\, p=Q+t_{\infty^{(2)},2}P\,\, \Leftrightarrow \,\, Q=p-t_{\infty^{(2)},2}q \cr
t_{\infty^{(2)},1}=&X_1\cr
t_{\infty^{(1)},2}=& t_{\infty^{(2)},2}-\frac{1}{s_{\infty^{(2)},2}}\cr
t_{\infty^{(3)},2}=& t_{\infty^{(2)},2}-\frac{1}{s_{\infty^{(1)},2}}\cr
t_{\infty^{(1)},1}=&-\frac{s_{\infty^{(2)},1}}{s_{\infty^{(2)},2}}\cr
t_{\infty^{(3)},1}=&-\frac{s_{\infty^{(1)},1}}{s_{\infty^{(1)},2}}\cr
t_{\infty^{(1)},0}=&s_{\infty^{(2)},0}\cr
t_{\infty^{(2)},0}=&-s_{\infty^{(1)},0}-s_{\infty^{(2)},0}\cr
t_{\infty^{(3)},0}=&s_{\infty^{(1)},0}\cr
0=&s_{X_1^{(1)},0}s_{X_1^{(2)},0}
\end{align}
In other words, one may identify both spectral curves loci $\det(\lambda I_3-\hat{L}_d(y))=0$ and $\det(y I_2-\td{L}_{\text{P4}}(\lambda))=0$ upon the swap $(\lambda,y)\leftrightarrow (y,\lambda)$ corresponding to the standard terminology of spectral duality.
\end{theorem}

\begin{proof}The proof follows from straightforward computations of both spectral curves and matching of each term giving the identification of \autoref{TheoDualitySpecCurves}. It is worth noticing that both sides are polynomial in $(\lambda,y)$ of respective degrees $3$ and $2$. The l.h.s is because $\hat{L}_d$ is chosen such that $\det \hat{L}_d(\lambda)\overset{\lambda\to \infty}{=} t_{\infty^{(2)},1}(t_{\infty^{(3)},2}-t_{\infty^{(2)},2})(t_{\infty^{(1)},2}-t_{\infty^{(2)},2})\lambda^2 +O(\lambda)$ while the r.h.s. is because $s_{X_1^{(1)},0}s_{X_1^{(2)},0}=0$ so that $\det \td{L}_{\text{P4}}(\lambda)$ has only a simple pole at $\lambda=X_1$.
\end{proof}

Note that the parameter $t_{\infty^{(2)},2}$ remains free in the identification so that the spectral duality is valid for any value of $t_{\infty^{(2)},2}$. When $t_{\infty^{(2)},2}$ is set to $0$ as proposed in the canonical choice of trivial times of \autoref{DefCanonicalTimesGl3} then the spectral duality  directly happens at the level of the Lax matrices $\td{L}(\lambda)$ and $\td{L}_{\text{P4}}(\lambda)$ for the Darboux coordinates $(q,p)$ and $(Q,P)$.

\begin{remark}On the chart $\prod_{i=1}^3t_{\infty^{(i)},2}\neq0$, the spectral curve in the initial gauge is dual to another curve of the same family. Define
\beq (t'_{\infty^{(i)},2},t'_{\infty^{(i)},1},t'_{\infty^{(i)},0},q',p')=
\left(\frac{1}{t_{\infty^{(i)},2}},-\frac{t_{\infty^{(i)},1}}{t_{\infty^{(i)},2}},-t_{\infty^{(i)},0},p,q\right),\quad i=1,2,3.\eeq
Let $P'_k$ denote the polynomials $P_k$ evaluated at the primed times, and set
\beq F'(q',p')=(p')^3-P'_1(q')(p')^2+P'_2(q')p'-P'_3(q').\eeq
Then
\begin{align}
\det(\lambda I_3-\td{L}(y))={}&-\prod_{i=1}^3t_{\infty^{(i)},2}
\bigl[y^3-P'_1(\lambda)y^2+P'_2(\lambda)y-P'_3(\lambda)-F'(q',p')\bigr].
\end{align}
Thus spectral duality exchanges the Darboux coordinates and changes the formal data as above. This description is unavailable when a leading time vanishes, as happens after the shift to $\hat{L}_d$.
\end{remark}

\begin{remark}One may wonder what is the dual of the $\mathfrak{gl}_2(\mathbb{C})$ case when $s_{X_1^{(1)},0}s_{X_1^{(2)},0}\neq 0$. In this case, the dual spectral curve would be of the form $Y^4+ U_1(\xi)Y^3+U_2(\xi)Y^2+V_2(\xi)Y+W_2(\xi)=0$ where the degree of the polynomials $(U_1,U_2,V_2,W_2)$ is given by their index. This implies a meromorphic connection in $\mathfrak{gl}_4(\mathbb{C})$ with a single pole at infinity and with local diagonalization at infinity given by
\beq \td{L}^{\text{dual}}_{\text{P4}}(\xi)=\text{diag}(0,0,s_1,s_2)\xi +\td{L}_{0,\text{P4}}^{\text{dual}}\,\,\text{with}\,\, \det \td{L}_{0,\text{P4}}^{\text{dual}}=s_{X_1^{(1)},0}s_{X_1^{(2)},0}\eeq
Note that the special case $s_{X_1^{(1)},0}s_{X_1^{(2)},0}=0$ corresponds to a case where $\td{L}_{0,\text{P4}}^{\text{dual}}$ collapses to a block-diagonal form $\td{L}_{0,\text{P4}}^{\text{dual}}=\diag(0,M_3)$ so that the meromorphic connection $\td{L}^{\text{dual}}_{\text{P4}}(\xi)$ can be reduced to a $\mathfrak{gl}_3(\mathbb{C})$ block only.
\end{remark}

Let us finally explain how to guess the trivial directions $\left(\mathcal{L}_{\mathbf{a}_i}\right)_{1\leq i\leq 3}$ in the $\mathfrak{gl}_3(\mathbb{C})$ setting using spectral duality. Indeed, the spectral duality implies a one-to-one map $\mathbf{t}\to (\mathbf{s},t_{\infty^{(2)},2})$ that translates immediately to the tangent spaces. However, in the dual $\mathfrak{gl}_2(\mathbb{C})$ setting, we know that the directions corresponding to dilatation/translation/trace given by \eqref{TrivialVectorsP4} must give linear evolutions for the Darboux coordinates. Translating this fact by spectral duality (that only exchanges $(q,p_d)$ to $(P,Q)$) to the $\mathfrak{gl}_3(\mathbb{C})$ case provides from \eqref{CorrespondanceTrivialDirections} that $\mathcal{L}_{\mathbf{a}_2}$ must provide linear evolutions for the Darboux coordinates (and of course, a similar reasoning provides the same result for $\mathcal{L}_{\mathbf{a}_1}$ and $\mathcal{L}_{\mathbf{a}_3}$ because in \eqref{CorrespondanceTrivialDirections} one could have used any of these three operators). This observation may be useful for general meromorphic connections but shall not provide all the trivial directions if none of the sides happens to be in $\mathfrak{gl}_2(\mathbb{C})$. This is why it would be useful to have a proper understanding of the trivial directions for meromorphic connections in $\mathfrak{gl}_d(\mathbb{C})$ without using spectral duality and we leave this open question for future works.

\subsection{Duality at the level of Hamiltonian evolutions and fundamental symplectic two-forms}

The spectral duality of \autoref{TheoDualitySpecCurves} allows one to verify whether the Hamiltonian evolutions on both sides are also dual. 

\begin{theorem}[Dual Hamiltonian evolutions]\label{TheoDualHamiltonians}The general Hamiltonian evolutions of \autoref{Defs} and \autoref{TheoHamP4} are identical under the identification given by the duality at the level of spectral curves of \autoref{TheoDualitySpecCurves}.
\end{theorem}

\begin{proof}The proof is done by direct verification in each direction of the tangent space. It is presented in  \autoref{AppendixIdentificationHamiltonian}. We note in particular that the trivial directions on each side are mapped with each other. 
\end{proof}

\begin{remark}The spectral-curve matching equation of \autoref{TheoDualitySpecCurves}, with its displayed prefactors, provides two different correspondences of the times and monodromies. However, the second solution does not satisfy \autoref{TheoDualHamiltonians} because the evolutions on both sides differ in an additional term $\hbar(s_{\infty^{(2)},2} -s_{\infty^{(1)},2})$. Note also that matching the curve loci under $(\lambda,y)\leftrightarrow(y,-\lambda)$ is possible and provides two distinct identifications of times and monodromies. However, it does not allow one to satisfy  \autoref{TheoDualHamiltonians} because the evolutions are opposite and we get extra factors of type $\hbar( \alpha s_{\infty^{(2)},2} +\beta s_{\infty^{(1)},2})$.
\end{remark}

\begin{remark}The spectral identification of  \autoref{TheoDualitySpecCurves} may be seen as a change of times and Darboux coordinates. On the $\mathfrak{gl}_3(\mathbb{C})$ side we have $(\mathbf{t},q,p)$ while on the $\mathfrak{gl}_2(\mathbb{C})$ side we would have $(t_{\infty^{(2)},2},\mathbf{s})$. However, note that in order to have tangent spaces with the same dimension, we need to keep $t_{\infty^{(2)},2}$ in the $\mathfrak{gl}_2(\mathbb{C})$ side although it does not appear naturally there. In particular, one could transfer the evolutions of $(q,p)$ given by $\text{Ham}_{\mathbf{e}_{\infty^{(2)},2}}(q,p;\hbar)$ to obtain the Hamiltonian evolutions of $(Q,P)$ relatively to $t_{\infty^{(2)},2}$ on the $\mathfrak{gl}_2(\mathbb{C})$. However since $t_{\infty^{(2)},2}$ is not a natural parameter on the $\mathfrak{gl}_2(\mathbb{C})$ side, we do not see any interest in these evolutions and shall not reproduce them here.      
\end{remark}

Let us observe that the spectral duality identification of \autoref{TheoDualitySpecCurves} implies that the reduced Darboux coordinates $(\check{q},\check{p})$ and $(\check{Q},\check{P})$ are identified by:
\begin{align} \label{EquivalenceReduced}\tau=&\sqrt{2} \td{X}_1\cr
\check{q}=&-\sqrt{2}(\check{Q}-\td{X}_1)\cr
\check{p}=&-\frac{1}{\sqrt{2}}\left(\check{P}-\check{Q}-\frac{s_{\infty^{(1)},0}+s_{\infty^{(2)},0}}{2(\check{Q}-\td{X}_1)} \right)\cr
t_{\infty^{(1)},0}=&s_{\infty^{(2)},0}\cr
t_{\infty^{(2)},0}=&-s_{\infty^{(1)},0}-s_{\infty^{(2)},0}\cr
t_{\infty^{(3)},0}=&s_{\infty^{(1)},0}
\end{align}
or equivalently
\begin{align}
    \td{X}_1=&\frac{1}{\sqrt{2}}\tau\cr
    \check{Q}=&-\frac{1}{\sqrt{2}}(\check{q}-\tau)\cr
    \check{P}=&-\sqrt{2}\left(\check{p}+\frac{1}{2}\check{q}-\frac{\tau}{2}-\frac{t_{\infty^{(2)},0}}{2\check{q}}\right)\cr
    s_{X_1^{(1)},0}s_{X_1^{(2)},0}=&0\cr
    s_{\infty^{(1)},0}=& t_{\infty^{(3)},0}\cr
    s_{\infty^{(2)},0}=& t_{\infty^{(1)},0}
\end{align}
This is consistent with the identification made in  \autoref{TheoDualReducedHamiltonian} that can be seen as a by-product of \autoref{TheoDualHamiltonians}.

\medskip

\autoref{TheoDualHamiltonians} may be extended to the fundamental symplectic two-forms.

\begin{theorem}[Duality for the fundamental symplectic two-forms]\label{TheoCorrespondenceFundamentalTwoForms} The two fundamental symplectic two-forms
\begin{align}
    \Omega^{(\text{P4})}=&\hbar dQ\wedge dP - \sum_{i=1}^2\sum_{k=1}^2 ds_{\infty^{(i)},k}\wedge d\text{Ham}_{(\mathbf{u}_{\infty^{(i)},k})}(Q,P) -dX_1\wedge d \text{Ham}_{(X_1)}(Q,P)\cr
   \Omega =&\hbar  dq \wedge dp -\sum_{i=1}^3\sum_{k=1}^{2} dt_{\infty^{(i)},k}\wedge d\text{Ham}_{(\mathbf{e}_{\infty^{(i)},k})}(q,p)
\end{align}
 satisfy $\Omega=-\Omega^{(\text{P4})}$ under the identification of the spectral duality given by \autoref{TheoDualitySpecCurves}.
\end{theorem}

\begin{proof}From \autoref{SymplecticReduction} and \autoref{PropReductionOmega2} we already know that the fundamental symplectic two-forms reduce to their respective forms, with the minus sign in \autoref{SymplecticReduction}. Moreover, we have from \eqref{EquivalenceReduced}:
{\small
\begin{align}-\Omega\overset{\text{Th.} \ref{SymplecticReduction}}{=}&\hbar d\check{q}\wedge d\check{p}-d\tau\wedge d\text{Ham}_{(\tau)}(\check{q},\check{p})\cr
=& \hbar(d\check{Q}-d\td{X}_1)\wedge\left(d\check{P}-d\check{Q}+\frac{s_{\infty^{(1)},0}+s_{\infty^{(2)},0}}{2(\check{Q}-\td{X}_1)^2}(d\check{Q}-d\td{X}_1)\right)\cr
&+\hbar  \frac{d \check{p}}{d\td{X}_1} d\td{X}_1\wedge d\check{q}-\hbar \frac{d\check{q}}{d\td{X}_1} d\td{X}_1\wedge d\check{p}\cr
=&\hbar d\check{Q}\wedge d\check{P}-\hbar d\td{X}_1\wedge d\check{P}+\hbar d\td{X}_1\wedge d\check{Q}+ \hbar  \frac{d \check{p}}{d\td{X}_1} d\td{X}_1\wedge d\check{q}-\hbar \frac{d\check{q}}{d\td{X}_1} d\td{X}_1\wedge d\check{p}\cr
=& \hbar d\check{Q}\wedge d\check{P}-\hbar d\td{X}_1\wedge d\check{P}+\hbar d\td{X}_1\wedge d\check{Q}+\hbar \frac{d}{d\td{X}_1}\left[\check{P}-\check{Q}-\frac{s_{\infty^{(1)},0}+s_{\infty^{(2)},0}}{2(\check{Q}-\td{X}_1)}\right] d\td{X}_1\wedge d\check{Q} \cr
&-\hbar \frac{d}{d\td{X}_1}\left[\check{Q}-\td{X}_1\right] d\td{X}_1\wedge \left(d\check{P}-d\check{Q}+\frac{s_{\infty^{(1)},0}+s_{\infty^{(2)},0}}{2(\check{Q}-\td{X}_1)^2}(d\check{Q}-d\td{X}_1)\right)\cr
=&\hbar d\check{Q}\wedge d\check{P}+ \hbar \frac{d \check{P}}{d\td{X}_1} d\td{X}_1\wedge d\check{Q} -\hbar \frac{d \check{Q}}{d\td{X}_1} d\td{X}_1\wedge d\check{P}  -\hbar d\td{X}_1\wedge d\check{P}+\hbar d\td{X}_1\wedge d\check{Q}\cr
&-\hbar \frac{d \check{Q}}{d\td{X}_1} d\td{X}_1\wedge d\check{Q}+  \hbar \frac{s_{\infty^{(1)},0}+s_{\infty^{(2)},0}}{2(\check{Q}-\td{X}_1)^2} \left(\frac{d \check{Q}}{d\td{X}_1}-1\right) d\td{X}_1\wedge d\check{Q}\cr
&+\hbar  d\td{X}_1\wedge d\check{P}+\hbar\frac{d\check{Q}}{d\td{X}_1}d\td{X}_1\wedge d\check{Q}-\hbar d\td{X}_1\wedge d\check{Q}- \hbar \frac{s_{\infty^{(1)},0}+s_{\infty^{(2)},0}}{2(\check{Q}-\td{X}_1)^2}\left(\frac{d\check{Q}}{d\td{X}_1}-1\right)d\td{X}_1\wedge d\check{Q}\cr
=&\hbar d\check{Q}\wedge d\check{P}+ \hbar \frac{d \check{P}}{d\td{X}_1} d\td{X}_1\wedge d\check{Q} -\hbar \frac{d \check{Q}}{d\td{X}_1} d\td{X}_1\wedge d\check{P}\cr
=&\hbar d\check{Q}\wedge d\check{P}-d\td{X}_1\wedge d\text{Ham}_{(\td{X}_1)}(\check{Q},\check{P})\cr
\overset{\text{Prop.} \ref{PropReductionOmega2}}{=}&\Omega^{(\text{P4})}
\end{align}}

\end{proof}

\begin{remark}
The unreduced duality is anti-symplectic: at fixed times, $q=P$ and $p=Q+t_{\infty^{(2)},2}P$ give $dq\wedge dp=-dQ\wedge dP$. Including the time terms, the preceding theorem gives
\beq \Omega=-\bigl(\hbar d\check{q}\wedge d\check{p}-d\tau\wedge d\text{Ham}_{(\tau)}\bigr)=-\Omega^{(\text{P4})}.\eeq
Multiplication of the two-form by $-1$ leaves its kernel unchanged, so the isomonodromic horizontal distributions correspond under duality.
\end{remark}

\subsection{Duality at the level of the Jimbo-Miwa-Ueno isomonodromic tau-functions}
Since the Jimbo-Miwa-Ueno isomonodromic differentials and associated tau-functions are central objects in integrable systems, it is natural to ask if the spectral duality given by \autoref{TheoDualitySpecCurves} extends to these quantities. Let us first remark that the definition of $\omega_{\text{JMU}}$ depends on the gauge that is chosen for $\td{L}(\lambda)$. Indeed, even if $\omega_{\text{JMU}}$ holds no dependence on $\lambda$-independent gauge transformations, it may however depend on $\lambda$-dependent gauge transformations that would still preserve the meromorphic structure of $\td{L}(\lambda)$. In particular, this is the case for the shift to $\hat{L}_d(\lambda)$ required to obtain spectral duality that might modify the Jimbo-Miwa-Ueno isomonodromic tau-function on the $\mathfrak{gl}_3$ side. Indeed, let us denote $\omega_{\text{JMU,d}}$ the Jimbo-Miwa-Ueno differential associated with $\hat{L}_d(\lambda)$:
\beq \label{omegaJMUd}\omega_{\text{JMU,d}}:=-\Res_{\lambda \to \infty}\Tr\Big[ \hat{\Psi}_{\text{d}}^{(\text{reg})}(\lambda)^{-1} \left(\hbar\partial_\lambda \hat{\Psi}_{\text{d}}^{(\text{reg})}(\lambda)\right) dT_{\text{d}}(\lambda)\Big]
\eeq
where we have $\hbar\partial_\lambda \hat{\Psi}_{\text{d}}(\lambda)= \hat{L}_{\text{d}}(\lambda) \hat{\Psi}_{\text{d}}(\lambda)$ and the formal asymptotic expansion of $\hat{\Psi}_{\text{d}}(\lambda)$ at infinity is given by
\beq \hat{\Psi}_{\text{d}}(\lambda)\overset{\lambda \to \infty}{\sim} \hat{\Psi}_{\text{d}}^{(\text{reg})}(\lambda) e^{T_{\text{d}}(\lambda)/\hbar}\eeq
where $T_{\text{d}}(\lambda)=T_{d,2}\lambda^2+T_{d,1}\lambda+T_{d,0}\ln\lambda$ is diagonal. Since we have from \autoref{DefDualityGauge}
\beq \hat{\Psi}_{\text{d}}(\lambda)=\exp\left(-\frac{t_{\infty^{(2)},2} \lambda^2}{2\hbar}\right) I_3\td{\Psi}(\lambda)\eeq
we have $\hat{\Psi}_{\text{d}}^{(\text{reg})}(\lambda)=\td{\Psi}^{(\text{reg})}(\lambda)$ and $T_{\text{d}}(\lambda)=T(\lambda)-\frac{t_{\infty^{(2)},2} \lambda^2}{2} I_3$. From \eqref{omegaJMUd} we get:
\begin{align} \omega_{\text{JMU,d}}=&\omega_{\text{JMU}}+\frac{1}{2}dt_{\infty^{(2)},2}\Res_{\lambda \to \infty}\Tr\Big[ \td{\Psi}^{(\text{reg})}(\lambda)^{-1} (\hbar\partial_\lambda \td{\Psi}^{(\text{reg})}(\lambda)) \lambda^2 \Big]\cr
=&\omega_{\text{JMU}}+\frac{\hbar}{2}\Tr(2F_2-F_1^2)dt_{\infty^{(2)},2}
\end{align}
Ultimately, it is then a straightforward computation from the resolution of $(F_1,F_2)$ in $\hbar \partial_\lambda \td{\Psi}= \td{L}\td{\Psi}$ with $\td{\Psi}^{(\text{reg})}(\lambda)=I_3+F_1\lambda^{-1}+F_2\lambda^{-2}+\dots$ to get that 
$\Tr( 2F_2-F_1^2)=0$ so that
\beq \omega_{\text{JMU,d}}=\omega_{\text{JMU}}\eeq
Hence, the Jimbo-Miwa-Ueno isomonodromic $\tau$-function is invariant by the gauge transformation that is necessary to obtain the spectral duality and we may not worry about it for the dual identification of the isomonodromic tau-functions. Let us now state the duality theorem corresponding to Jimbo-Miwa-Ueno differentials.

\begin{theorem}[Duality at the level of Jimbo-Miwa-Ueno isomonodromic tau-functions]\label{TheoDualityOmegaJMU} Under the identification of the times, monodromies and Darboux coordinates given by spectral duality of \autoref{TheoDualitySpecCurves}, we have
\beq \omega_{\text{JMU}}=\omega_{\text{JMU}}^{(\text{P4})}+df\eeq
with
\beq f(\mathbf{s}):=\frac{1}{2}\left((s_{\infty^{(1)},0})^2\ln(s_{\infty^{(1)},2})+(s_{\infty^{(2)},0})^2\ln(s_{\infty^{(2)},2})+\frac{s_{\infty^{(1)},0}(s_{\infty^{(1)},1})^2}{s_{\infty^{(1)},2}}+\frac{s_{\infty^{(2)},0}(s_{\infty^{(2)},1})^2}{s_{\infty^{(2)},2}}
\right)\eeq
or equivalently
\begin{align}
    f(\mathbf{t}):=&-\frac{1}{2}\Bigg((t_{\infty^{(1)},0})^2\ln(t_{\infty^{(1)},2}-t_{\infty^{(2)},2})+(t_{\infty^{(3)},0})^2\ln(t_{\infty^{(2)},2}-t_{\infty^{(3)},2})\cr&
   +\frac{t_{\infty^{(3)},0}(t_{\infty^{(3)},1})^2}{t_{\infty^{(3)},2}-t_{\infty^{(2)},2}}+\frac{t_{\infty^{(1)},0}(t_{\infty^{(1)},1})^2}{t_{\infty^{(1)},2}-t_{\infty^{(2)},2}}\Bigg)
\end{align}
\end{theorem}

\begin{proof}We have already shown that the Jimbo-Miwa-Ueno differentials reduce nicely to $\text{Ham}_{(\tau)}(\check{q},\check{p};\hbar=0)d\tau$ or $\text{Ham}_{(\td{X}_1)}^{(\text{P4})}(\check{Q},\check{P};\hbar=0) d\td{X}_1$ up to some exact terms (\autoref{PropJMUDifferential} and \autoref{ThJMUP4}). Moreover, we have:
\footnotesize{\begin{align}
    &\text{Ham}_{(\tau)}(\check{q},\check{p};\hbar=0)d\tau= \left[-\left(\check{q}\check{p}^2+\check{q}^2\check{p} -\tau \check{q}\check{p}-t_{\infty^{(2)},0}\check{p} +t_{\infty^{(1)},0}\check{q}\right) 
    \right]\sqrt{2}\, d\td{X}_1\cr&
    =\Big[\sqrt{2}(\check{Q}-\td{X}_1)\left(-\frac{1}{\sqrt{2}}\left(\check{P}-\check{Q}+\frac{t_{\infty^{(2)},0}}{2(\check{Q}-\td{X}_1)}\right)\right)^2+ \left(\sqrt{2}(\check{Q}-\td{X}_1)\right)^2\frac{1}{\sqrt{2}}\left(\check{P}-\check{Q}+\frac{t_{\infty^{(2)},0}}{2(\check{Q}-\td{X}_1)}\right)\cr&
    +\sqrt{2}\td{X}_1(\check{Q}-\td{X}_1)\left(\check{P}-\check{Q}+\frac{t_{\infty^{(2)},0}}{2(\check{Q}-\td{X}_1)}\right)- t_{\infty^{(2)},0}\frac{1}{\sqrt{2}}\left(\check{P}-\check{Q}+\frac{t_{\infty^{(2)},0}}{2(\check{Q}-\td{X}_1)}\right)+t_{\infty^{(1)},0}\sqrt{2}(\check{Q}-\td{X}_1)
\Big]\sqrt{2}\, d\td{X}_1\cr
&= \Big[(\check{Q}-\td{X}_1)\left(\check{P}-\check{Q}+\frac{t_{\infty^{(2)},0}}{2(\check{Q}-\td{X}_1)}\right)^2+ 2(\check{Q}-\td{X}_1)^2\left(\check{P}-\check{Q}+\frac{t_{\infty^{(2)},0}}{2(\check{Q}-\td{X}_1)}\right)\cr&
    +2\td{X}_1(\check{Q}-\td{X}_1)\left(\check{P}-\check{Q}+\frac{t_{\infty^{(2)},0}}{2(\check{Q}-\td{X}_1)}\right)- t_{\infty^{(2)},0}\left(\check{P}-\check{Q}+\frac{t_{\infty^{(2)},0}}{2(\check{Q}-\td{X}_1)}\right)+2t_{\infty^{(1)},0}(\check{Q}-\td{X}_1)
\Big]\, d\td{X}_1\cr
&= \Big[(\check{Q}-\td{X}_1)\check{P}^2+ \left[-2\check{Q}(\check{Q}-\td{X}_1)+t_{\infty^{(2)},0}+2(\check{Q}-\td{X}_1)^2+2\td{X}_1(\check{Q}-\td{X}_1)-t_{\infty^{(2)},0}  \right]\check{P}\cr&   
+ \check{Q}^2(\check{Q}-\td{X}_1)+\frac{(t_{\infty^{(2)},0})^2}{4(\check{Q}-\td{X}_1)} -\check{Q}t_{\infty^{(2)},0}- 2(\check{Q}-\td{X}_1)^2\check{Q}+ t_{\infty^{(2)},0} (\check{Q}-\td{X}_1)\cr&
-2\td{X}_1(\check{Q}-\td{X}_1)\check{Q}+t_{\infty^{(2)},0} \td{X}_1+ t_{\infty^{(2)},0}\check{Q}-\frac{(t_{\infty^{(2)},0})^2}{2( \check{Q}-\td{X}_1)}+2t_{\infty^{(1)},0}(\check{Q}-\td{X}_1)\Big]d\td{X}_1\cr
&= \Big[(\check{Q}-\td{X}_1)\check{P}^2-\check{Q}^3+\td{X}_1\check{Q}^2+(2t_{\infty^{(1)},0} +t_{\infty^{(2)},0})\check{Q} -2t_{\infty^{(1)},0}\td{X}_1 -\frac{(t_{\infty^{(2)},0})^2}{4(\check{Q}-\td{X}_1)}  \Big]d\td{X}_1\cr
&= \Big[(\check{Q}-\td{X}_1)\check{P}^2-\check{Q}^3+\td{X}_1\check{Q}^2-(s_{\infty^{(1)},0} -s_{\infty^{(2)},0})\check{Q} -2s_{\infty^{(2)},0}\td{X}_1 -\frac{(s_{X_1^{(1)},0}+s_{X_1^{(2)},0})^2}{4(\check{Q}-\td{X}_1)}  \Big]d\td{X}_1\cr
&=\left(\text{Ham}_{(\td{X}_1)}^{(\text{P4})}(\check{Q},\check{P};\hbar=0)-2s_{\infty^{(2)},0}\td{X}_1\right)d\td{X}_1
\end{align}}
\normalsize{from} \autoref{PropReducedP4}. Thus, we have from the definition of the differentials:
\beq \omega_{\text{JMU}}-\omega_{\text{JMU}}^{(\text{P4})}=dG_0-dK_0-2s_{\infty^{(2)},0}\td{X}_1d\td{X}_1\eeq
It is then a straightforward computation using the map $\mathbf{t}\mapsto (\mathbf{s},t_{\infty^{(2)},2})$ given by  \autoref{TheoDualitySpecCurves} to get
\begin{align}
    dt_{\infty^{(1)},1}=&-\frac{1}{s_{\infty^{(2)},2}}d s_{\infty^{(2)},1}+\frac{s_{\infty^{(2)},1}}{(s_{\infty^{(2)},2})^2}d s_{\infty^{(2)},2}\cr
    dt_{\infty^{(2)},1}=&d X_1\cr
    dt_{\infty^{(3)},1}=&-\frac{1}{s_{\infty^{(1)},2}}d s_{\infty^{(1)},1}+\frac{s_{\infty^{(1)},1}}{(s_{\infty^{(1)},2})^2}d s_{\infty^{(1)},2}\cr
dt_{\infty^{(1)},2}=&dt_{\infty^{(2)},2}+\frac{1}{(s_{\infty^{(2)},2})^2}d s_{\infty^{(2)},2}\cr
dt_{\infty^{(2)},2}=&dt_{\infty^{(2)},2}\cr
dt_{\infty^{(3)},2}=&dt_{\infty^{(2)},2}+\frac{1}{(s_{\infty^{(1)},2})^2}d s_{\infty^{(1)},2}
\end{align}
so that we have
\beq  dG_0-dK_0-2s_{\infty^{(2)},0}\td{X}_1d\td{X}_1=df(\mathbf{s})\eeq 
with 
\beqq f(\mathbf{s}):=\frac{1}{2}\left((s_{\infty^{(1)},0})^2\ln(s_{\infty^{(1)},2})+(s_{\infty^{(2)},0})^2\ln(s_{\infty^{(2)},2})+\frac{s_{\infty^{(1)},0}(s_{\infty^{(1)},1})^2}{s_{\infty^{(1)},2}}+\frac{s_{\infty^{(2)},0}(s_{\infty^{(2)},1})^2}{s_{\infty^{(2)},2}}
\right)\eeqq
\end{proof}

\begin{remark}
    The exact term $df$ contributes at order $\hbar^0$ to $\omega_{\mathrm{JMU}}$, and therefore at order $\hbar^{-2}$ to $\ln\tau_{\mathrm{JMU}}$.
\end{remark}

Let us briefly mention that \autoref{TheoDualityOmegaJMU} is equivalent to saying
\begin{align}\label{IdentificationHamForm} &-\sum_{i=1}^3\sum_{k=1}^2 \text{Ham}_{(\mathbf{e}_{\infty^{(i)},k})}(q,p;\hbar=0) dt_{\infty^{(i)},k}+dG_0(\mathbf{t})=\text{Ham}_{(\tau)}(\check{q},\check{p};\hbar=0)d\tau +dG_0(\mathbf{T})\cr&=\text{Ham}_{(\td{X}_1)}^{(\text{P4})}(\check{Q},\check{P};\hbar=0) d\td{X}_1+dK_0(\mathbf{S})+df(\mathbf{S})\cr&
=\sum_{i=1}^2\sum_{k=1}^2 \text{Ham}_{(\mathbf{u}_{\infty^{(i)},k})}^{(\text{P4})}(Q,P;\hbar=0) ds_{\infty^{(i)},k} +\text{Ham}_{(\mathbf{u}_{X_1})}^{(\text{P4})}(Q,P;\hbar=0) dX_1+dK_0(\mathbf{s})+df(\mathbf{s})\cr
\end{align}
In particular, using \eqref{IdentificationHamForm}, one can interpret the duality of the Jimbo-Miwa-Ueno isomonodromic tau-function as a duality of the Hamiltonian differentials evaluated at $\hbar=0$ up to some exact normalizing one-forms.  

\medskip

We finally conclude that since the spectral duality implies the correspondence of the fundamental symplectic two-forms and of the Jimbo-Miwa-Ueno differentials, then it implies the correspondence of the entire symplectic structure as expected from Harnad's duality \cite{Harnad_1994} and results of P. Boalch \cite{Boalch2012}. In the next section, we shall see that these identities can be interpreted using Hermitian matrix models and topological recursion and may have applications towards combinatorics and enumerative geometry.

\subsection{Duality at the level of the Hermitian matrix integrals and topological recursion}
\subsubsection{Topological recursion and perturbative TR-partition function} 
In this section, we shall define the perturbative partition functions associated with the classical spectral curves using the Chekhov-Eynard-Orantin topological recursion (TR) defined in \cite{EO07}. We shall not rewrite the proper definition of the topological recursion and refer to \cite{EO07} for the definitions and properties as well as \cite{TRReview,Norbury_survey} for some reviews on topological recursion and quantum curves. In order to apply the topological recursion of \cite{EO07}, we need to choose a Torelli marking on each Riemann surface associated with the classical spectral curves on each side. Thus, we shall denote $\left(\mathcal{A},\mathcal{B}\right)$ a basis of cycles on $\Sigma$ and $\left(\mathcal{A}^{(\text{P4})},\mathcal{B}^{(\text{P4})}\right)$ a basis of cycles on $\Sigma_{\text{P4}}$. We shall also define:
\beq \epsilon:= \oint_{\mathcal{A}} y(\lambda)d\lambda \,\, \text{ and }\,\, \epsilon^{(\text{P4})}:=\oint_{\mathcal{A}^{(\text{P4})}} Y(\xi)d\xi\eeq
the filling fraction on each side. Finally, we shall denote
\beq \omega_{0,1}(\lambda):=y(\lambda)d\lambda \,\, \text{ and }\,\,  \omega^{(\text{P4})}_{0,1}(\xi):=Y(\xi)d\xi
\eeq
and $\omega_{0,2}$ (resp. $\omega^{(\text{P4})}_{0,2}$) as the unique symmetric $(1\boxtimes 1)$-form on $\Sigma^2$ (resp. $\Sigma_{\text{P4}}^2$) with a unique double pole on the diagonal, without residue, bi-residue equal to $1$ and normalized on the $\mathcal{A}$-cycle (resp. $\mathcal{A}^{(\text{P4})}$-cycle)  by
\beq
 \oint_{z_1 \in \mathcal{A}} \om_{0,2}(z_1,z_2) = 0 \,\, \text{ (resp. } \oint_{z_1 \in \mathcal{A}^{(\text{P4})} } \om^{(\text{P4})}_{0,2}(z_1,z_2) = 0 \text{  ).}
\eeq
Using these materials, one may apply the topological recursion on each side (See \cite{EO07} for definition) and obtain the so-called Eynard-Orantin differentials $\left(\omega_{h,n}\right)_{h\geq 0,n\geq 0}$ and $\left(\omega^{(\text{P4})}_{h,n}\right)_{h\geq 0,n\geq 0}$. In this notation, the free energies correspond to $F^{(h)}:=\omega_{h,0}$ for all $h\geq 0$. 

Let us also mention that the free energies produced by the topological recursion are usually stacked into a formal generating series to produce the ``perturbative TR-partition function".

\begin{definition}[Perturbative topological recursion partition functions]\label{DefTRPartitionFunctions}One may define the perturbative TR-partition functions on both sides using the free energies generated by topological recursion by
\begin{align}
    Z_{\text{TR}}(\mathbf{t};\hbar) :=&\exp\left(\sum_{h=0}^{\infty} \hbar^{2h-2} \omega_{h,0}\right)\cr
    Z_{\text{TR}}^{(\text{P4})}(\mathbf{s};\hbar) :=&\exp\left(\sum_{h=0}^{\infty} \hbar^{2h-2} \omega^{(\text{P4})}_{h,0}\right)
\end{align}
where $\ln  Z_{\text{TR}}$ and $\ln Z_{\text{TR}}^{(\text{P4})}$ are understood as formal Laurent series in $\hbar$.
\end{definition}

\subsubsection{Duality for degeneration of classical spectral curves to  genus $0$}
The formal behaviour as $\hbar\to0$ depends on the chosen solution branch as well as on the deformation parameters and monodromy data. A regular formal power-series branch and the degeneration of its leading spectral curve must therefore be specified together. In general, higher-genus quantization involves transseries and theta-function corrections \cite{Quantization_2021,BorotGuionnet,Guionnet}; the genus-zero situation considered here is the perturbative one.

\begin{definition}[Genus-zero formal branch]\label{DefSingularTimes}We write $\mathbf{t}_{\mathrm{deg}}$ (respectively $\mathbf{s}_{\mathrm{deg}}$) for deformation data together with a chosen regular formal solution branch in $\hbar$ whose leading spectral curve has genus zero. This notation includes the branch choice and does not assert that degeneration is determined by the times alone.
\end{definition}

Degeneration to genus $0$ classical spectral curves allows one to make connections between the tau-function and the free energies generated by topological recursion \cite{IwakiMarchalSaenz,MOsl2}. It is also well-known that in the one-cut case (i.e. genus $0$), the logarithm of the partition function of Hermitian matrix integrals, when properly normalized with a $\frac{N^{-\frac{1}{12}}}{N!}$ factor, has a simple power series expansion in $N^{-1}$  \cite{Borot2011AsymptoticEO,E1MM}.

\begin{proposition}[Identification of JMU tau-functions with TR for degenerate genus $0$ cases]\label{TheoremDegeneration} For any value of degenerate times and monodromies of \autoref{DefSingularTimes} providing degenerate genus $0$ classical spectral curves we have under the spectral duality of \autoref{TheoDualitySpecCurves}:
\bea d(\ln\tau_{\text{JMU}}(\mathbf{t};\hbar)) &=&d(\ln Z_{\text{TR}}(\mathbf{t};\hbar))\cr
d(\ln\tau^{(\text{P4})}_{\text{JMU}}(\mathbf{s};\hbar)) &=&d(\ln Z_{\text{TR}}^{(\text{P4})}(\mathbf{s};\hbar))\cr
\hbar^2 d(\ln\tau_{\text{JMU}}(\mathbf{t};\hbar)) =\omega_{\text{JMU}}(\mathbf{t};\hbar)&\overset{\text{Th. } \ref{TheoDualityOmegaJMU}}{=} & \omega_{\text{JMU}}^{(\text{P4})}(\mathbf{s};\hbar)+df(\mathbf{s})= \hbar^2d(\ln\tau^{(\text{P4})}_{\text{JMU}}(\mathbf{s};\hbar))  +df(\mathbf{s}) \cr&&
\eea
where $d$ is the differential relatively to the deformation parameters (i.e. irregular times and location of finite poles).\\
In other words, one can interpret the $x-y$ symmetry in topological recursion as equivalent to the duality of the Jimbo-Miwa-Ueno isomonodromic tau-function in this degenerate setup.
\end{proposition}

\begin{proof}The proof follows from the fact that in the genus $0$ degeneration case, the Lax systems satisfy the topological type property \cite{bergre2009determinantal,BergereBorotEynard} on both sides and thus coefficients of the JMU differentials are reconstructed by TR up to an overall normalization factor (independent of the deformation parameters). In particular, in genus $0$ case, the wave matrices admit a formal WKB expansion giving the fact that the JMU tau-function has an exponential formal expansion, with a Laurent series for its logarithm. The proof of the topological type property on the $\mathfrak{gl}_2(\mathbb{C})$ side was done in \cite{IwakiMarchalSaenz}. On the $\mathfrak{gl}_3(\mathbb{C})$ side, it follows from the general results on quantization of classical spectral curves \cite{Quantization_2021} that reconstruct the isomonodromic Lax system from topological recursion.
\end{proof}

We may now identify the partition functions of the matrix models when the spectral curves degenerate to genus $0$. 

\begin{proposition}[Duality at the level of Hermitian matrix models for degenerate genus $0$ curves]\label{TheoDualitynew} For any value of degenerate times and monodromies of \autoref{DefSingularTimes} providing degenerate genus $0$ classical spectral curves we have under spectral duality of \autoref{TheoDualitySpecCurves} and the compatible condition $t_{\infty^{(1)},0}=-1=s_{\infty^{(2)},0}$ together with $s_{X_1^{(1)},0}s_{X_1^{(2)},0}=0$ that
\bea
    d(\ln\tau_{\text{JMU}}(\mathbf{t};\hbar=N^{-1})) &=&d(\ln Z_{\text{TR}}(\mathbf{t};\hbar=N^{-1}))=d(\ln Z^{(\text{2MM})}(\mathbf{t};N))  \cr
    d(\ln\tau^{(\text{P4})}_{\text{JMU}}(\mathbf{s};\hbar=N^{-1})) &=&d(\ln Z_{\text{TR}}^{(\text{P4})}(\mathbf{s};\hbar))=d(\ln Z^{(\text{1MM})}(\mathbf{s};N)) \cr
    \hbar^2 d(\ln\tau_{\text{JMU}}(\mathbf{t};\hbar)) =\omega_{\text{JMU}}(\mathbf{t};\hbar)&\overset{\text{Th. } \ref{TheoDualityOmegaJMU}}{=} & \omega_{\text{JMU}}^{(\text{P4})}(\mathbf{s};\hbar)+df(\mathbf{s})= \hbar^2d(\ln\tau^{(\text{P4})}_{\text{JMU}}(\mathbf{s};\hbar))  +df(\mathbf{s}) \cr&&
\eea
where $d$ is the differential relatively to the deformation parameters (i.e. irregular times and location of finite poles). In particular, $\omega_{0,0}-\omega_{0,0}^{(\text{P4})}=f+C_0$, whereas $\omega_{h,0}-\omega_{h,0}^{(\text{P4})}=C_h$ for $h\geq1$, with $C_h$ independent of the deformation parameters.
\end{proposition}

\begin{remark} Let us note that the condition $t_{\infty^{(1)},0}=-1$ is equivalent to $s_{\infty^{(2)},0}=-1$ by the spectral duality so that the last theorem is consistent.
\end{remark}

\begin{proof}For genus $0$ spectral curves, the partition functions of Hermitian matrix models, when properly normalized in $N$ (i.e. with a factor $\frac{N^{-\frac{1}{12}}}{N!}$), admit a formal expansion of the form $\ln Z= \underset{k=-1}{\overset{\infty}{\sum}} Z_k N^{-2k}$ from \cite{Borot2011AsymptoticEO} and coefficients can be matched with TR because loop equations (Virasoro constraints) match the bilinear identities of the tau-function. Alternatively, it follows from the fact that the correlation functions of Hermitian matrix models satisfy the same loop equations as in TR and that the topological type property holds for these correlation functions.  
\end{proof}

\subsubsection{Conjecture for genus $1$ curve}

For stable indices $2h-2+n>0$, the Eynard-Orantin differentials scale as $a^{2-2h-n}$ under $\lambda\mapsto a\lambda+b$ with $y$ fixed. Invariance of the free energies under the $x-y$ swap requires additional hypotheses and attention to residue terms and normalization \cite{EO07,EO2MM,EOxy}. Genus zero alone is not a sufficient criterion: counterexamples already occur for genus-zero curves \cite{BouchardGeneralizedTR}. Thus the identities established for the present genus-zero family should not be read as a general assertion of symplectic invariance. For higher-genus quantization one must also account for theta-function and transseries contributions to the perturbative TR partition function. We formulate the following conjectures for the present dual family.

Let us first formulate the conjectures regarding the matrix models partition function and the corresponding tau-functions:

\begin{conjecture}[Identification of the 2MM partition function with the JMU tau-function]\label{Prop2MM} For $t_{\infty^{(1)},0}=-1$, the partition function of the two-matrix models with potentials $V_1$ and $V_2$ given by \eqref{PotentialsV1V2} identifies with the JMU tau-function
\beq d(\ln Z_N^{(\text{2MM})}(\mathbf{t}))=d(\ln \tau_{\text{JMU}}(\mathbf{t};\hbar=N^{-1})) \eeq       
\end{conjecture}

\begin{conjecture}[Identification of the 1MM partition function with the JMU tau-function]\label{PartitionFunctionP4} Let $Z^{(\text{1MM})}_{N}(\mathbf{s};N)$ be the partition function of the Hermitian one-matrix model with potential $V(\xi)=\frac{1}{2}(s_{\infty^{(2)},2}-s_{\infty^{(1)},2})\xi^2+(s_{\infty^{(2)},1}-s_{\infty^{(1)},1})\xi+(1-s_{\infty^{(1)},0})\ln(\xi-X_1)$. Then we have for $s_{X_1^{(1)},0}s_{X_1^{(2)},0}=0$ and $s_{\infty^{(2)},0}=-1$:
\beq  d(\ln Z_N^{(\text{1MM})}(\mathbf{s};N))=d(\ln \tau_{\text{JMU}}^{(P4)}(\mathbf{s};\hbar=N^{-1})) \eeq 
\end{conjecture}

Then, we conjecture the following duality relations:

\begin{conjecture}[Conjecture for duality at the level of TR and Hermitian matrix models for genus $1$ curves]\label{ConjTRHMM} Under the spectral duality of \autoref{TheoDualitySpecCurves} we have:
\begin{itemize}\item Duality of the free energies, with the normalization compatible with \autoref{TheoDualityOmegaJMU}: 
   \beq \omega_{0,0}-\omega_{0,0}^{(\text{P4})}=f+C_0,\qquad
   \omega_{h,0}-\omega_{h,0}^{(\text{P4})}=C_h\quad(h\geq1),\eeq
   where $f$ is the function defined in \autoref{TheoDualityOmegaJMU} and the $C_h$ are independent of the deformation parameters.
   \item Identification of the non-perturbative TR-partition function with the JMU tau-functions:
   \bea  d(\ln \tau_{\text{JMU}}(\mathbf{t};\hbar))&=& d(\ln Z_{\text{TR, NP}}(\mathbf{t};\hbar))\cr
   d(\ln \tau_{\text{JMU}}^{(P4)}(\mathbf{s};\hbar))&=& d(\ln Z_{\text{TR, NP}}^{(\text{P4})}(\mathbf{s};\hbar))
   \eea
   where the non-perturbative TR-partition functions are defined in \cite{Quantization_2021} by adding Theta-functions terms to the perturbative TR-partition functions.
\end{itemize}  
\end{conjecture}

\section{Conclusion and outlooks}\label{SectionOutlooks}
In this article, we analyzed using explicit formulas two non-trivial examples of isomonodromic deformations of meromorphic connections that we related using Harnad's duality at different levels (See \autoref{Fig2Diagram}). The main advantage of our approach compared to abstract settings developed for example in \cite{Boalch2012,yamakawa2014fourierlaplace,Yamakawa2017TauFA,Yamakawa2019FundamentalTwoForms} is that we obtain explicit expressions for the Lax matrices and Hamiltonians at each step that can be used directly for people interested in applications. 
There are also many interesting questions associated with the present important example that we list below for future works
\begin{itemize}
    \item If \autoref{TheoDualReducedHamiltonian} identifies both reduced Hamiltonian systems without restrictions, the Lax matrices, the spectral curves and JMU differentials can only be identified on both sides in a specific gauge on the $\mathfrak{gl}_3(\mathbb{C})$ side and with the condition $s_{X_1^{(1)},0}s_{X_1^{(2)},0}=0$ on the monodromies at $X_1$ on the $\mathfrak{gl}_2(\mathbb{C})$ side. In particular, this last condition cannot be achieved for general values of the parameters in the Painlev\'{e} IV equation in the standard $\mathfrak{sl}_2(\mathbb{C})$ setting where $s_{X_1^{(2)},0}=-s_{X_1^{(1)},0}$. This suggests that considering the full space of deformation parameters (i.e., before reduction of the tangent space and the canonical choice of times) and the full space of monodromies is important for understanding $x-y$ duality. In particular, it allows one to reach a larger set of Lax representations. With this in mind, it would be interesting to see if other new Lax representations for the other Painlev\'{e} equations can be obtained from general $\mathfrak{gl}_2(\mathbb{C})$ Lax representation with arbitrary monodromy parameters by $x-y$ duality.
    \item The most natural issue is to generalize this work to all unramified meromorphic connections. The general formula for the gauge matrix of \autoref{LaxMatrixgl3} already exists  \cite{Quantization_2021} and the simple form of the matrix $L(\lambda)$ in the oper gauge in  \autoref{LaxMatrixgl3} indicates that this part should generalize easily using the Darboux coordinates defined by the apparent singularities and their dual on the spectral curve. The next step is to solve the compatibility equations and obtain the Hamiltonian evolutions. This has been achieved for arbitrary meromorphic connections in $\mathfrak{gl}_2(\mathbb{C})$ \cite{MarchalAlameddineP1Hierarchy2023,marchal2023hamiltonian} and we believe that this technical step could be mastered. The third step is then to split the tangent space into trivial and non-trivial directions and we expect the non-trivial directions to be a subspace of dimension equal to the genus of the spectral curve. This part is still mysterious because one needs to understand the geometric origin of the trivial directions which is unclear in the present $\mathfrak{gl}_3(\mathbb{C})$ example (the directions $\left(\mathcal{L}_{\mathbf{a}_i}\right)_{1\leq i\leq 3}$ do not have clear explanations apart from pulling back the known trivial directions in $\mathfrak{gl}_2(\mathbb{C})$ by spectral duality). Eventually, the final step would be to reduce the Jimbo-Miwa-Ueno differentials and the fundamental symplectic two-forms and to study the spectral duality and the potential relation with Hermitian matrix models.
    \item An interesting observation made in this article, and that was also made earlier for all $\mathfrak{gl}_2(\mathbb{C})$ Lax pairs for the six Painlev\'{e} equations is the fact that the Jimbo-Miwa-Ueno differential is expressed through the Hamiltonian form evaluated at $\hbar=0$, up to the stated sign and an exact time-dependent term ( \autoref{PropJMUDifferential} and \autoref{ThJMUP4}). This leads us to propose the following conjecture
    \begin{conjecture}\label{ConjectureJMU}For any isomonodromic deformations of $\hbar$-deformed meromorphic connections, let $\mathbf{t}$ be the deformation times and $\overline{\omega}$ be the associated Hamiltonian differential, then we conjecture that
    \beq \omega_{JMU}(\mathbf{q},\mathbf{p},\mathbf{t})=\epsilon\,\overline{\omega}(\mathbf{q},\mathbf{p},\mathbf{t};\hbar=0) +dU(\mathbf{t})\eeq
    where $\epsilon$ is a constant and $(\mathbf{q},\mathbf{p})$ are the Darboux coordinates associated with the apparent singularities and their dual partner on the spectral curve (or any other Darboux coordinates related to this set by $\hbar$-independent transformations) and $dU$ is a purely time-dependent exact term that could be inserted in the definition of the Hamiltonians.
    \end{conjecture}
    This conjecture is important because it would provide some geometric understanding of the formal parameter $\hbar$ that could be seen as an interpolation parameter between the classical or isospectral world ($\hbar=0$) and the standard isomonodromic Hamiltonian setting ($\hbar=1$).
    \item The relation with Hermitian matrix models is also interesting and would deserve investigations. The first natural question is the following: ``Can every classical spectral curve associated with isomonodromic deformations of meromorphic connections be obtained by the classical spectral curve of a Hermitian matrix model?"
    It is known that solutions of loop equations are random matrices \cite{EynardSolutionsLoopEquations} and using logarithmic terms in the potentials or hard edges it might be possible to obtain sufficient classical spectral curves to obtain the one coming from meromorphic connections. Even if the answer is no, the present examples show that for many cases, the connection with Hermitian matrix models exists by identifying $\hbar=N^{-1}$. The main advantage of Hermitian matrix models is that one can obtain recursive relations from $N$ to $N+1$ and this has been used to derive the fact that the partition function of the matrix model is an isomonodromic tau-function. Thus, these recursive relations (standardly obtained by orthogonal polynomials and their three-term recurrence relations) interpolate between the classical world ($\hbar=0 \,\Leftrightarrow\, N\to \infty$) and the standard Hamiltonian world ($\hbar=1 \,\Leftrightarrow\, N=1$). Consequently, it indicates that some relations for the wave matrices for different values of $\hbar$ should exist and it would be interesting to obtain them without reference to Hermitian matrix integrals. Of particular interest is then the possibility to use these recursive relations to quantize from the classical world to the $\hbar=1$ case (whereas the recursive relations in matrix models go the other way from $N$ to $N+1$, hence from Hamiltonian to classical world). 
    \item Finally, the $x-y$ symmetry, i.e. Harnad's duality generalized to some abstract Laplace-Fourier transforms by P. Boalch and D. Yamakawa \cite{BoalchKlein2004,Boalch2012,yamakawa2014fourierlaplace} using the action of the Weyl algebra on the system could be made more explicit by some explicit knowledge of the Hamiltonians and their reduction to fewer non-trivial directions (i.e. a Lagrangian reduction). It would also be interesting to study the consequences of this symmetry at the level of enumerative geometry in particular for map enumeration in relation with matrix models or, if one could generalize the present results to ``exponential variables" $X:=\text{exp}(\lambda)$ and $Y:=\text{exp}(y)$ (or equivalently when the base curve is no longer $\overline{\mathbb{C}}$ but $\mathbb{P}^1\setminus\{0,\infty\}$) to enumeration of Gromov-Witten invariants. This topic has been very recently tackled in \cite{Weller2024} for some genus $0$ cases with the help of the log-TR of \cite{ABDKSLogTR2024}.
\end{itemize}

\section*{Acknowledgments}The authors would like to thank Nicolas Orantin, Gabriele Rembado and Jean Dou\c{c}ot for fruitful discussions and explanations.

\appendix
\renewcommand{\theequation}{\thesection-\arabic{equation}}

\section{Explicit Hamiltonians associated with \autoref{Defs}} \label{DirectinHams}
The evolutions of the Darboux coordinates of \autoref{Defs} are Hamiltonian. For completeness, we provide here the explicit expressions of the associated Hamiltonians in each direction. 
\small{\begin{align}\label{Hamcertain}\text{Ham}_{(\mathbf{v}_{\infty,1})}(q,p;\hbar)=&-\hbar q  +v_{\infty,1}(\mathbf{t};\hbar) 
\cr
  \text{Ham}_{(\mathbf{v}_{\infty,2})}(q,p;\hbar)=&-\frac{1}{2}\hbar q^2 +v_{\infty,2}(\mathbf{t};\hbar) 
  \cr
\text{Ham}_{(\mathbf{u}_{\infty,1})}(q,p;\hbar)=&-\hbar p+u_{\infty,1}(\mathbf{t};\hbar)
\cr
\text{Ham}_{(\mathbf{u}_{\infty,2})}(q,p;\hbar)=&-\hbar q p +u_{\infty,2}(\mathbf{t};\hbar)
\cr
\text{Ham}_{(\mathbf{a}_{1})}(q,p;\hbar)=&\hbar \bigg(-p^2+(t_{\infty^{(2)},2}+t_{\infty^{(3)},2})q p+ (t_{\infty^{(2)},1}+t_{\infty^{(3)},1})p\cr
&-t_{\infty^{(2)},2}t_{\infty^{(3)},2}q^2-(t_{\infty^{(2)},2}t_{\infty^{(3)},1}+t_{\infty^{(3)},2}t_{\infty^{(2)},1})q\bigg)+a_{1}(\mathbf{t};\hbar)
\cr
\text{Ham}_{(\mathbf{a}_{2})}(q,p;\hbar)=&\hbar \bigg(-p^2+(t_{\infty^{(1)},2}+t_{\infty^{(3)},2})q p+ (t_{\infty^{(1)},1}+t_{\infty^{(3)},1})p\cr
&-t_{\infty^{(1)},2}t_{\infty^{(3)},2}q^2-(t_{\infty^{(1)},2}t_{\infty^{(3)},1}+t_{\infty^{(3)},2}t_{\infty^{(1)},1})q\bigg)+a_{2}(\mathbf{t};\hbar)
\cr
\text{Ham}_{(\mathbf{a}_{3})}(q,p;\hbar)=&\hbar \bigg(-p^2+(t_{\infty^{(1)},2}+t_{\infty^{(2)},2})q p+ (t_{\infty^{(1)},1}+t_{\infty^{(2)},1})p\cr
&-t_{\infty^{(1)},2}t_{\infty^{(2)},2}q^2-(t_{\infty^{(1)},2}t_{\infty^{(2)},1}+t_{\infty^{(2)},2}t_{\infty^{(1)},1})q\bigg)+a_{3}(\mathbf{t};\hbar)
\cr
\text{Ham}_{(\mathbf{e}_{\infty^{(1)},1})}(q,p;\hbar)=&\frac{-p^3+P_1(q)p^2-P_2(q)p+P_3(q)}{(t_{\infty^{(3)},2}- t_{\infty^{(1)},2})(t_{\infty^{(2)},2}- t_{\infty^{(1)},2})}+e_{1,1}(\mathbf{t};\hbar)
\cr
\text{Ham}_{(\mathbf{e}_{\infty^{(2)},1})}(q,p;\hbar)=&\frac{-p^3+P_1(q)p^2-P_2(q)p+P_3(q)}{(t_{\infty^{(3)},2}- t_{\infty^{(2)},2})(t_{\infty^{(1)},2}- t_{\infty^{(2)},2})} +\hbar\frac{p-t_{\infty^{(3)},2}q}{t_{\infty^{(3)},2}- t_{\infty^{(2)},2}}+e_{2,1}(\mathbf{t};\hbar)
\cr
\text{Ham}_{(\mathbf{e}_{\infty^{(3)},1})}(q,p;\hbar)=&\frac{-p^3+P_1(q)p^2-P_2(q)p+P_3(q)}{(t_{\infty^{(1)},2}- t_{\infty^{(3)},2})(t_{\infty^{(2)},2}- t_{\infty^{(3)},2})} -\hbar\frac{p-t_{\infty^{(2)},2}q}{t_{\infty^{(3)},2}- t_{\infty^{(2)},2}}+e_{3,1}(\mathbf{t};\hbar)
  \end{align}}

\normalsize{As} explained below \autoref{Defs}, the purely time-dependent terms (that do not modify the Hamiltonian evolutions) of the Hamiltonians are chosen so that the fundamental symplectic two-form  and the JMU differential have simpler forms. These terms  are explicitly given by:  

\scriptsize{\begin{align}\label{ConstantTermsHam}
    &v_{\infty,1}(\mathbf{t};\hbar):=\frac{\hbar}{2}\left( \frac{t_{\infty^{(3)},1}-t_{\infty^{(2)},1}}{t_{\infty^{(2)},2}-t_{\infty^{(3)},2}} +\left(\ln(t_{\infty^{(1)},2}-t_{\infty^{(2)},2})-\ln(t_{\infty^{(2)},2}-t_{\infty^{(3)},2})\right)\frac{t_{\infty^{(2)},1}+t_{\infty^{(3)},1}}{t_{\infty^{(1)},2}-t_{\infty^{(3)},2}}\right)\cr
    &v_{\infty,2}(\mathbf{t};\hbar):=0\cr
    &u_{\infty,1}(\mathbf{t};\hbar):=\frac{\hbar}{2}\Bigg[\frac{t_{\infty^{(2)},2}t_{\infty^{(3)},1}-t_{\infty^{(2)},1}t_{\infty^{(3)},2}}{t_{\infty^{(2)},2}-t_{\infty^{(3)},2}}\cr&
    +\left(\ln(t_{\infty^{(1)},2}-t_{\infty^{(2)},2})-\ln(t_{\infty^{(2)},2}-t_{\infty^{(3)},2})\right)\frac{t_{\infty^{(2)},1}t_{\infty^{(3)},2}+t_{\infty^{(2)},2}t_{\infty^{(3)},1}}{t_{\infty^{(1)},2}-t_{\infty^{(3)},2}}\Bigg]\cr
    &u_{\infty,2}(\mathbf{t};\hbar):=0\cr
    &a_{1}(\mathbf{t};\hbar):=\hbar \frac{t_{\infty^{(2)},1}t_{\infty^{(3)},1}(t_{\infty^{(1)},2}-t_{\infty^{(2)},2})}{(t_{\infty^{(1)},2}-t_{\infty^{(3)},2})}\left(\ln(t_{\infty^{(2)},2}-t_{\infty^{(3)},2})-\ln(t_{\infty^{(1)},2}-t_{\infty^{(2)},2})\right)\cr
    &a_{2}(\mathbf{t};\hbar):= -\frac{\hbar t_{\infty^{(3)},1}}{2(t_{\infty^{(2)},2}-t_{\infty^{(3)},2})(t_{\infty^{(1)},2}-t_{\infty^{(3)},2})}\Bigg[\cr&
    ((t_{\infty^{(2)},2}-t_{\infty^{(3)},2})(\ln(t_{\infty^{(1)},2}-t_{\infty^{(2)},2})-\ln(t_{\infty^{(2)},2}-t_{\infty^{(3)},2}))+t_{\infty^{(1)},2}-t_{\infty^{(3)},2})\cr
    &((t_{\infty^{(1)},1}-2t_{\infty^{(2)},1}+t_{\infty^{(3)},1})t_{\infty^{(2)},2}+(t_{\infty^{(2)},1}-t_{\infty^{(1)},1})t_{\infty^{(3)},2}+t_{\infty^{(1)},2}(t_{\infty^{(2)},1}-t_{\infty^{(3)},1}))\Bigg]\cr
    &a_{3}(\mathbf{t};\hbar):=-\frac{\hbar t_{\infty^{(2)},1}}{2(t_{\infty^{(2)},2}-t_{\infty^{(3)},2})(t_{\infty^{(1)},2}-t_{\infty^{(3)},2})}\Bigg[\cr&
    (t_{\infty^{(3)},2}-t_{\infty^{(2)},2})[t_{\infty^{(1)},2}(t_{\infty^{(2)},1}-t_{\infty^{(3)},1})+(t_{\infty^{(1)},1}+t_{\infty^{(3)},1})t_{\infty^{(2)},2}\cr& -(t_{\infty^{(1)},1}+t_{\infty^{(2)},1})t_{\infty^{(3)},2}]\left(\ln(t_{\infty^{(1)},2}-t_{\infty^{(2)},2})-\ln(t_{\infty^{(2)},2}-t_{\infty^{(3)},2})\right) \cr
    &+(t_{\infty^{(1)},2}-t_{\infty^{(3)},2})(t_{\infty^{(1)},2}(t_{\infty^{(2)},1}-t_{\infty^{(3)},1})+t_{\infty^{(2)},2}(t_{\infty^{(1)},1}-t_{\infty^{(3)},1})-(t_{\infty^{(1)},1}+t_{\infty^{(2)},1}-2t_{\infty^{(3)},1})t_{\infty^{(3)},2})\Bigg]\cr&
    e_{1,1}(\mathbf{t};\hbar)=\frac{((-t_{\infty^{(1)},0}t_{\infty^{(1)},2}+t_{\infty^{(3)},1}t_{\infty^{(1)},1}-t_{\infty^{(3)},2}t_{\infty^{(2)},0}-t_{\infty^{(3)},0}t_{\infty^{(3)},2})t_{\infty^{(2)},1}-t_{\infty^{(1)},1}t_{\infty^{(2)},0}(t_{\infty^{(2)},2}-t_{\infty^{(3)},2}))}{(t_{\infty^{(1)},2}-t_{\infty^{(3)},2})(t_{\infty^{(1)},2}-t_{\infty^{(2)},2})}\cr&
 e_{2,1}(\mathbf{t};\hbar)=\frac{((t_{\infty^{(3)},2}-t_{\infty^{(1)},2}) t_{\infty^{(2)},0}-t_{\infty^{(3)},1} t_{\infty^{(1)},1}-t_{\infty^{(1)},2} t_{\infty^{(3)},0}+t_{\infty^{(3)},0} t_{\infty^{(3)},2}) t_{\infty^{(2)},1}+t_{\infty^{(1)},1} t_{\infty^{(2)},0} (t_{\infty^{(2)},2}-t_{\infty^{(3)},2})}{(t_{\infty^{(2)},2}-t_{\infty^{(3)},2}) (t_{\infty^{(1)},2}-t_{\infty^{(2)},2})}\cr&
 +\frac{\hbar t_{\infty^{(3)},1}}{2}\left(\frac{1}{t_{\infty^{(2)},2}-t_{\infty^{(3)},2}}+\frac{\ln(t_{\infty^{(1)},2}-t_{\infty^{(2)},2})-\ln(t_{\infty^{(2)},2}-t_{\infty^{(3)},2})}{t_{\infty^{(1)},2}-t_{\infty^{(3)},2}}\right)\cr&
 e_{3,1}(\mathbf{t};\hbar)=\frac{((t_{\infty^{(1)},0} t_{\infty^{(3)},2}+t_{\infty^{(3)},1} t_{\infty^{(1)},1}+t_{\infty^{(1)},2} t_{\infty^{(2)},0}+t_{\infty^{(1)},2} t_{\infty^{(3)},0}) t_{\infty^{(2)},1}-t_{\infty^{(1)},1} t_{\infty^{(2)},0} (t_{\infty^{(2)},2}-t_{\infty^{(3)},2}))}{((t_{\infty^{(2)},2}-t_{\infty^{(3)},2}) (t_{\infty^{(1)},2}-t_{\infty^{(3)},2}))}\cr&
 +\frac{\hbar t_{\infty^{(2)},1}}{2} \left(-\frac{1}{t_{\infty^{(2)},2}-t_{\infty^{(3)},2}}+ \frac{\ln(t_{\infty^{(1)},2}-t_{\infty^{(2)},2})-\ln(t_{\infty^{(2)},2}-t_{\infty^{(3)},2})}{t_{\infty^{(1)},2}-t_{\infty^{(3)},2}}\right)\cr&
e_{i,2}(\mathbf{t};\hbar)=\frac{a_i(\mathbf{t};\hbar)-B_i e_{i,1}(\mathbf{t};\hbar)}{2D_i},\qquad i=1,2,3,
\end{align}}
\normalsize{where}, for $\{i,j,k\}=\{1,2,3\}$, we set
\begin{align}
D_i:=&(t_{\infty^{(i)},2}-t_{\infty^{(j)},2})(t_{\infty^{(i)},2}-t_{\infty^{(k)},2}),\cr
B_i:=&(2t_{\infty^{(i)},2}-t_{\infty^{(j)},2}-t_{\infty^{(k)},2})t_{\infty^{(i)},1}\cr
&+(t_{\infty^{(k)},2}-t_{\infty^{(i)},2})t_{\infty^{(j)},1}
+(t_{\infty^{(j)},2}-t_{\infty^{(i)},2})t_{\infty^{(k)},1}.
\end{align}

The remaining Hamiltonians follow from the defining relation
\beq \mathcal{L}_{\mathbf{a}_i}=2D_i\mathcal{L}_{\mathbf{e}_{\infty^{(i)},2}}+B_i\mathcal{L}_{\mathbf{e}_{\infty^{(i)},1}}.\eeq
With the preceding time-dependent normalization they are
\begin{align}\label{OtherHamiltonians}
\text{Ham}_{(\mathbf{e}_{\infty^{(i)},2})}(q,p;\hbar)
=\frac{\text{Ham}_{(\mathbf{a}_i)}(q,p;\hbar)-B_i\text{Ham}_{(\mathbf{e}_{\infty^{(i)},1})}(q,p;\hbar)}{2D_i},\quad i=1,2,3.
\end{align}
All quantities on the right-hand side are explicitly given in \eqref{Hamcertain} and \eqref{ConstantTermsHam}; this formula also includes the time-only term $e_{i,2}$.

\section{Expression of the auxiliary matrices}\label{AppendixAuxiliaryGeneral}
In this appendix, we give the expressions of the auxiliary matrices $\td{A}_{\boldsymbol{\alpha}}(\lambda)$ for specific directions spanning the whole tangent space.
\begin{align}
    \td{A}_{\mathbf{v}_{\infty,1}}(\lambda)=&\text{diag}(\lambda,\lambda,\lambda)\cr
    \td{A}_{\mathbf{v}_{\infty,2}}(\lambda)=&\text{diag}\left(\frac{\lambda^2}{2},\frac{\lambda^2}{2},\frac{\lambda^2}{2}\right)
\end{align}
\begin{align}   \td{A}_{\mathbf{u}_{\infty,1}}(\lambda)=&\begin{pmatrix}t_{\infty^{(1)},2}\lambda& 1&1\\ \left[\td{A}_{\mathbf{u}_{\infty,1}}\right]_{2,1}
 &t_{\infty^{(2)},2}\lambda+t_{\infty^{(2)},1}-t_{\infty^{(1)},1}&t_{\infty^{(2)},2}q-p+t_{\infty^{(2)},1}\\
\left[\td{A}_{\mathbf{u}_{\infty,1}}\right]_{3,1}& t_{\infty^{(3)},2}q-p+t_{\infty^{(3)},1}& t_{\infty^{(3)},2}\lambda+t_{\infty^{(3)},1}-t_{\infty^{(1)},1} \end{pmatrix}\cr
\left[\td{A}_{\mathbf{u}_{\infty,1}}\right]_{2,1}=&\frac{1}{t_{\infty^{(2)},2}-t_{\infty^{(3)},2}}\Big[((t_{\infty^{(3)},2}q^2-pq+qt_{\infty^{(3)},1}-t_{\infty^{(2)},0})t_{\infty^{(2)},2}+\cr&(t_{\infty^{(2)},1}q-pq+t_{\infty^{(2)},0})t_{\infty^{(3)},2}+(p-t_{\infty^{(3)},1})(p-t_{\infty^{(2)},1}))(t_{\infty^{(1)},2}-t_{\infty^{(2)},2})\Big]\cr
\left[\td{A}_{\mathbf{u}_{\infty,1}}\right]_{3,1}=&\frac{1}{t_{\infty^{(2)},2}-t_{\infty^{(3)},2}}\Big[((t_{\infty^{(2)},2}q^2-pq+t_{\infty^{(2)},1}q-t_{\infty^{(3)},0})t_{\infty^{(3)},2}+\cr&
(t_{\infty^{(3)},1}q-pq+t_{\infty^{(3)},0})t_{\infty^{(2)},2}+(p-t_{\infty^{(3)},1})(p-t_{\infty^{(2)},1}))(t_{\infty^{(3)},2}-t_{\infty^{(1)},2})\Big]\cr
\end{align}
\begin{align}   \td{A}_{\mathbf{u}_{\infty,2}}(\lambda)=&\begin{pmatrix}\lambda(t_{\infty^{(1)},2}\lambda +t_{\infty^{(1)},1}) & \lambda&\lambda\\ \lambda\left[\td{A}_{\mathbf{u}_{\infty,1}}\right]_{2,1}
 &t_{\infty^{(2)},2}\lambda^2+t_{\infty^{(2)},1}\lambda+\hbar&(t_{\infty^{(2)},2}q-p+t_{\infty^{(2)},1})\lambda\\
\lambda\left[\td{A}_{\mathbf{u}_{\infty,1}}\right]_{3,1}& (t_{\infty^{(3)},2}q-p+t_{\infty^{(3)},1})\lambda& t_{\infty^{(3)},2}\lambda^2+t_{\infty^{(3)},1}\lambda+\hbar \end{pmatrix}\cr
\end{align}
\scriptsize{\begin{align}
    \td{A}_{\mathbf{e}_{\infty^{(1)},1}}(\lambda)=&\begin{pmatrix}
        \lambda &\frac{1}{t_{\infty^{(1)},2}-t_{\infty^{(2)},2}}&\frac{1}{t_{\infty^{(1)},2}-t_{\infty^{(3)},2}}\\
        \left[\td{A}_{\mathbf{e}_{\infty^{(1)},1}}\right]_{2,1}&\left[\td{A}_{\mathbf{e}_{\infty^{(1)},1}}\right]_{2,2}&0\\
        \left[\td{A}_{\mathbf{e}_{\infty^{(1)},1}}\right]_{3,1}&0&\left[\td{A}_{\mathbf{e}_{\infty^{(1)},1}}\right]_{3,3}
    \end{pmatrix}\cr
 \left[\td{A}_{\mathbf{e}_{\infty^{(1)},1}}\right]_{2,1}=&
\frac{(t_{\infty^{(3)},2}q^2-pq+t_{\infty^{(3)},1}q-t_{\infty^{(2)},0})t_{\infty^{(2)},2}+(t_{\infty^{(2)},1}q-pq+t_{\infty^{(2)},0})t_{\infty^{(3)},2}+(p-t_{\infty^{(3)},1})(p-t_{\infty^{(2)},1})}{t_{\infty^{(2)},2}-t_{\infty^{(3)},2}}\cr
\left[\td{A}_{\mathbf{e}_{\infty^{(1)},1}}\right]_{3,1}=&\frac{(pq-t_{\infty^{(3)},2}q^2-t_{\infty^{(3)},1}q-t_{\infty^{(3)},0})t_{\infty^{(2)},2}+(pq-t_{\infty^{(2)},1}q+t_{\infty^{(3)},0})t_{\infty^{(3)},2}-(p-t_{\infty^{(3)},1})(p-t_{\infty^{(2)},1})}{t_{\infty^{(2)},2}-t_{\infty^{(3)},2}}\cr
\left[\td{A}_{\mathbf{e}_{\infty^{(1)},1}}\right]_{2,2}=&
\frac{(t_{\infty^{(3)},2}q-p-t_{\infty^{(1)},1}+t_{\infty^{(2)},1}+t_{\infty^{(3)},1})t_{\infty^{(1)},2}-(t_{\infty^{(2)},2}q+t_{\infty^{(2)},1}-t_{\infty^{(1)},1})t_{\infty^{(3)},2}+t_{\infty^{(2)},2}(p-t_{\infty^{(3)},1})}{(t_{\infty^{(1)},2}-t_{\infty^{(3)},2})(t_{\infty^{(1)},2}-t_{\infty^{(2)},2})}\cr
\left[\td{A}_{\mathbf{e}_{\infty^{(1)},1}}\right]_{3,3}=&
\frac{(t_{\infty^{(2)},2}q-p-t_{\infty^{(1)},1}+t_{\infty^{(2)},1}+t_{\infty^{(3)},1})t_{\infty^{(1)},2}-(t_{\infty^{(3)},2}q+t_{\infty^{(3)},1}-t_{\infty^{(1)},1})t_{\infty^{(2)},2}+t_{\infty^{(3)},2}(p-t_{\infty^{(2)},1})}{(t_{\infty^{(1)},2}-t_{\infty^{(3)},2})(t_{\infty^{(1)},2}-t_{\infty^{(2)},2})}\cr
\end{align}}
\scriptsize{\begin{align}\td{A}_{\mathbf{e}_{\infty^{(2)},1}}(\lambda)=&\begin{pmatrix}
        0 &\frac{1}{t_{\infty^{(2)},2}-t_{\infty^{(1)},2}}&0\\
        \left[\td{A}_{\mathbf{e}_{\infty^{(2)},1}}\right]_{2,1}&\lambda+\left[\td{A}_{\mathbf{e}_{\infty^{(2)},1}}\right]_{2,2}&\frac{t_{\infty^{(2)},2}q-p+t_{\infty^{(2)},1}}{t_{\infty^{(2)},2}-t_{\infty^{(3)},2}}\\
        0&\frac{(t_{\infty^{(3)},2}q-p+t_{\infty^{(3)},1})}{t_{\infty^{(2)},2}-t_{\infty^{(3)},2}} & \frac{(t_{\infty^{(1)},2}-t_{\infty^{(3)},2})(p-t_{\infty^{(2)},2}q-t_{\infty^{(2)},1})}{(t_{\infty^{(1)},2}-t_{\infty^{(2)},2})(t_{\infty^{(2)},2}-t_{\infty^{(3)},2})}
    \end{pmatrix}\cr
 \left[\td{A}_{\mathbf{e}_{\infty^{(2)},1}}\right]_{2,1}=&\frac{(pq-t_{\infty^{(3)},2}q^2-t_{\infty^{(3)},1}q+t_{\infty^{(2)},0})t_{\infty^{(2)},2}+(pq-t_{\infty^{(2)},1}q-t_{\infty^{(2)},0})t_{\infty^{(3)},2}-(p-t_{\infty^{(3)},1})(p-t_{\infty^{(2)},1})}{t_{\infty^{(2)},2}-t_{\infty^{(3)},2}}\cr
 \left[\td{A}_{\mathbf{e}_{\infty^{(2)},1}}\right]_{2,2}=&
\frac{(t_{\infty^{(3)},2}q-p+t_{\infty^{(1)},1}-t_{\infty^{(2)},1}+t_{\infty^{(3)},1})t_{\infty^{(2)},2}-(t_{\infty^{(1)},2}q+t_{\infty^{(1)},1}-t_{\infty^{(2)},1})t_{\infty^{(3)},2}+(p-t_{\infty^{(3)},1})t_{\infty^{(1)},2}}{(t_{\infty^{(1)},2}-t_{\infty^{(2)},2})(t_{\infty^{(2)},2}-t_{\infty^{(3)},2})}\cr
\end{align}}
\scriptsize{\begin{align}\td{A}_{\mathbf{e}_{\infty^{(3)},1}}(\lambda)=&\begin{pmatrix}
        0 &0&\frac{1}{t_{\infty^{(3)},2}-t_{\infty^{(1)},2}}\\0&
       \frac{(t_{\infty^{(1)},2}-t_{\infty^{(2)},2})(t_{\infty^{(3)},2}q-p+t_{\infty^{(3)},1})}{(t_{\infty^{(2)},2}-t_{\infty^{(3)},2})(t_{\infty^{(1)},2}-t_{\infty^{(3)},2})}&\frac{(p-t_{\infty^{(2)},2}q-t_{\infty^{(2)},1})}{t_{\infty^{(2)},2}-t_{\infty^{(3)},2}}\\
\left[\td{A}_{\mathbf{e}_{\infty^{(3)},1}}\right]_{3,1}&\frac{p-t_{\infty^{(3)},2}q-t_{\infty^{(3)},1}}{t_{\infty^{(2)},2}-t_{\infty^{(3)},2}} &\lambda+\left[\td{A}_{\mathbf{e}_{\infty^{(3)},1}}\right]_{3,3}
    \end{pmatrix}\cr
 \left[\td{A}_{\mathbf{e}_{\infty^{(3)},1}}\right]_{3,1}=&\frac{(t_{\infty^{(3)},2}q^2-pq+t_{\infty^{(3)},1}q+t_{\infty^{(3)},0})t_{\infty^{(2)},2}+(t_{\infty^{(2)},1}q-pq-t_{\infty^{(3)},0})t_{\infty^{(3)},2}+(p-t_{\infty^{(3)},1})(p-t_{\infty^{(2)},1})}{t_{\infty^{(2)},2}-t_{\infty^{(3)},2}}\cr
 \left[\td{A}_{\mathbf{e}_{\infty^{(3)},1}}\right]_{3,3}=&
\frac{(p-t_{\infty^{(2)},2}q-t_{\infty^{(1)},1}-t_{\infty^{(2)},1}+t_{\infty^{(3)},1})t_{\infty^{(3)},2}+(t_{\infty^{(1)},2}q+t_{\infty^{(1)},1}-t_{\infty^{(3)},1})t_{\infty^{(2)},2}-(p-t_{\infty^{(2)},1})t_{\infty^{(1)},2}}{(t_{\infty^{(2)},2}-t_{\infty^{(3)},2})(t_{\infty^{(1)},2}-t_{\infty^{(3)},2})}\cr
\end{align}}
\normalsize{For} the direction $\mathbf{a}_2$, an explicit compact formula avoids the lengthy expanded entries. Let
\begin{align}
v(\lambda):={}&\bigl(0,\ (t_{\infty^{(2)},2}-t_{\infty^{(3)},2})\lambda
+t_{\infty^{(3)},2}q-p+t_{\infty^{(2)},1},\ t_{\infty^{(2)},2}q-p+t_{\infty^{(2)},1}\bigr),\cr
r_1:={}&vG^{-1},\qquad r_2:=\hbar\partial_\lambda r_1+r_1L,\qquad
r_3:=\hbar\partial_\lambda r_2+r_2L.
\end{align}
Here $G$ and the oper matrix $L$ are the explicit matrices of \autoref{LaxMatrixgl3}. The required auxiliary matrix is
\beq\label{Aa2Compact}
\td{A}_{\mathbf{a}_2}(\lambda)=G^{-1}\left[
\begin{pmatrix}r_1\\r_2\\r_3\end{pmatrix}G-\mathcal{L}_{\mathbf{a}_2}[G]\right].
\eeq
The derivative in this formula includes the evolution of $q,p$. Equivalently, on a function $F(q,p,\mathbf{t})$ it is computed explicitly as
\begin{align}
\mathcal{L}_{\mathbf{a}_2}[F]={}&2\hbar D_2\partial_{t_{\infty^{(2)},2}}F+
\hbar B_2\partial_{t_{\infty^{(2)},1}}F\cr
&+\frac{\partial\text{Ham}_{(\mathbf{a}_2)}}{\partial p}\partial_qF
-\frac{\partial\text{Ham}_{(\mathbf{a}_2)}}{\partial q}\partial_pF,
\end{align}
with $D_2,B_2$ and $\text{Ham}_{(\mathbf{a}_2)}$ given in Appendix A. Formula \eqref{Aa2Compact} follows from \eqref{SecondThirdLineA} and the gauge relation in \autoref{EquivalenceAtdA}; its entries are polynomial in $\lambda$. Direct substitution verifies the zero-curvature equation.
\normalsize{One} may then obtain the auxiliary matrices in the other directions by taking linear combinations of the previous matrices. Indeed we have:
\begin{align}
\begin{pmatrix}\mathcal{L}_{\mathbf{v}_{\infty,2}}\\\mathcal{L}_{\mathbf{u}_{\infty,2}}\\\mathcal{L}_{\mathbf{a}_2}\end{pmatrix}
={}&\begin{pmatrix}1&1&1\\2t_{\infty^{(1)},2}&2t_{\infty^{(2)},2}&2t_{\infty^{(3)},2}\\0&2D_2&0\end{pmatrix}
\begin{pmatrix}\mathcal{L}_{\mathbf{e}_{\infty^{(1)},2}}\\\mathcal{L}_{\mathbf{e}_{\infty^{(2)},2}}\\\mathcal{L}_{\mathbf{e}_{\infty^{(3)},2}}\end{pmatrix}\cr
&+\begin{pmatrix}0&0&0\\t_{\infty^{(1)},1}&t_{\infty^{(2)},1}&t_{\infty^{(3)},1}\\0&B_2&0\end{pmatrix}
\begin{pmatrix}\mathcal{L}_{\mathbf{e}_{\infty^{(1)},1}}\\\mathcal{L}_{\mathbf{e}_{\infty^{(2)},1}}\\\mathcal{L}_{\mathbf{e}_{\infty^{(3)},1}}\end{pmatrix}.
\end{align}
\normalsize{so that}
\begin{align}
  &\mathcal{L}_{\mathbf{e}_{\infty^{(1)},2}}= \frac{t_{\infty^{(3)},2}}{t_{\infty^{(3)},2}-t_{\infty^{(1)},2}}\mathcal{L}_{\mathbf{v}_{\infty,2}}
  + \frac{1}{2(t_{\infty^{(1)},2}-t_{\infty^{(3)},2})}\mathcal{L}_{\mathbf{u}_{\infty,2}}\cr&+ \frac{1}{2((t_{\infty^{(1)},2})^2-t_{\infty^{(1)},2}t_{\infty^{(2)},2}-t_{\infty^{(1)},2}t_{\infty^{(3)},2}+t_{\infty^{(2)},2}t_{\infty^{(3)},2})}\mathcal{L}_{\mathbf{a}_{2}}\cr&
  -\frac{1}{2(t_{\infty^{(1)},2}-t_{\infty^{(3)},2})}\Big[t_{\infty^{(1)},1}\mathcal{L}_{\mathbf{e}_{\infty^{(1)},1}}+t_{\infty^{(2)},1}\mathcal{L}_{\mathbf{e}_{\infty^{(2)},1}}+t_{\infty^{(3)},1}\mathcal{L}_{\mathbf{e}_{\infty^{(3)},1}}\Big]+\cr&
\frac{t_{\infty^{(1)},2}(t_{\infty^{(2)},1}-t_{\infty^{(3)},1})+(t_{\infty^{(1)},1}-2t_{\infty^{(2)},1}+t_{\infty^{(3)},1})t_{\infty^{(2)},2}-t_{\infty^{(3)},2}(t_{\infty^{(1)},1}-t_{\infty^{(2)},1})}{(2(t_{\infty^{(1)},2}-t_{\infty^{(3)},2}))(t_{\infty^{(1)},2}-t_{\infty^{(2)},2})}\mathcal{L}_{\mathbf{e}_{\infty^{(2)},1}}\cr&
\mathcal{L}_{\mathbf{e}_{\infty^{(2)},2}}=\frac{1}{2(t_{\infty^{(1)},2}t_{\infty^{(3)},2}-t_{\infty^{(1)},2}t_{\infty^{(2)},2}+(t_{\infty^{(2)},2})^2-t_{\infty^{(2)},2}t_{\infty^{(3)},2})}\mathcal{L}_{\mathbf{a}_{2}}+\cr&
\frac{(2t_{\infty^{(2)},1}-t_{\infty^{(1)},1}-t_{\infty^{(3)},1})t_{\infty^{(2)},2}+(t_{\infty^{(3)},1}-t_{\infty^{(2)},1})t_{\infty^{(1)},2}+t_{\infty^{(3)},2}(t_{\infty^{(1)},1}-t_{\infty^{(2)},1})}{2(t_{\infty^{(2)},2}-t_{\infty^{(3)},2})(t_{\infty^{(1)},2}-t_{\infty^{(2)},2})}\mathcal{L}_{\mathbf{e}_{\infty^{(2)},1}}\cr&
\mathcal{L}_{\mathbf{e}_{\infty^{(3)},2}}=\frac{t_{\infty^{(1)},2}}{t_{\infty^{(1)},2}-t_{\infty^{(3)},2}}\mathcal{L}_{\mathbf{v}_{\infty,2}}-\frac{1}{2(t_{\infty^{(1)},2}-t_{\infty^{(3)},2})}\mathcal{L}_{\mathbf{u}_{\infty,2}}\cr&
+\frac{1}{2(t_{\infty^{(1)},2}t_{\infty^{(2)},2}-t_{\infty^{(1)},2}t_{\infty^{(3)},2}-t_{\infty^{(2)},2}t_{\infty^{(3)},2}+(t_{\infty^{(3)},2})^2)}\mathcal{L}_{\mathbf{a}_{2}}\cr&
+\frac{1}{2(t_{\infty^{(1)},2}-t_{\infty^{(3)},2})}\Big[t_{\infty^{(1)},1}\mathcal{L}_{\mathbf{e}_{\infty^{(1)},1}}+t_{\infty^{(2)},1}\mathcal{L}_{\mathbf{e}_{\infty^{(2)},1}}+t_{\infty^{(3)},1}\mathcal{L}_{\mathbf{e}_{\infty^{(3)},1}}\Big]+\cr&
\frac{(t_{\infty^{(1)},1}-2t_{\infty^{(2)},1}+t_{\infty^{(3)},1})t_{\infty^{(2)},2}+(t_{\infty^{(2)},1}-t_{\infty^{(1)},1})t_{\infty^{(3)},2}+t_{\infty^{(1)},2}(t_{\infty^{(2)},1}-t_{\infty^{(3)},1})}{(2(t_{\infty^{(2)},2}-t_{\infty^{(3)},2}))(t_{\infty^{(1)},2}-t_{\infty^{(3)},2})}\mathcal{L}_{\mathbf{e}_{\infty^{(2)},1}}\cr
\end{align}
giving the other auxiliary matrices as linear combinations of the ones given in this appendix.

\section{Proof of \autoref{SymplecticReduction}}
\label{ProofSymplecticReduction}
Let us first observe that adding some exact terms $df$ to the definition of the Hamiltonian form does not contribute in the definition of $\Omega$. Let us also observe that the map $\mathbf{t}\to (T_1, T_2, T_3, T_4, T_5, \tau)$ implies that:
\beq \partial_\tau=\frac{\sqrt{(t_{\infty^{(2)},2}-t_{\infty^{(1)},2})(t_{\infty^{(1)},2}-t_{\infty^{(3)},2})(t_{\infty^{(3)},2}-t_{\infty^{(2)},2})}}{2t_{\infty^{(1)},2}-t_{\infty^{(2)},2}-t_{\infty^{(3)},2}}(\partial_{t_{\infty^{(2)},1}}-\partial_{t_{\infty^{(3)},1}})
\eeq
It implies that up to purely time-dependent terms we have
\footnotesize{\begin{align}\label{Hamtau} \text{Ham}_{(\tau)}(\check{q},\check{p})=&\frac{\sqrt{(t_{\infty^{(2)},2}-t_{\infty^{(1)},2})(t_{\infty^{(1)},2}-t_{\infty^{(3)},2})(t_{\infty^{(3)},2}-t_{\infty^{(2)},2})}}{2t_{\infty^{(1)},2}-t_{\infty^{(2)},2}-t_{\infty^{(3)},2}}\left(\text{Ham}_{\mathbf{e}_{\infty^{(2)},1}}(\check{q},\check{p})-\text{Ham}_{\mathbf{e}_{\infty^{(3)},1}}(\check{q},\check{p})\right)
\end{align}}
\normalsize{Let} us recall that we have from the Hamiltonian system
\begin{align}\label{AppendixHam1}
\text{Ham}_{(\mathbf{e}_{\infty^{(1)},1})}(q,p)=&\frac{-p^3+P_1(q)p^2-P_2(q)p+P_3(q)}{(t_{\infty^{(3)},2}- t_{\infty^{(1)},2})(t_{\infty^{(2)},2}- t_{\infty^{(1)},2})} +e_{1,1}(\mathbf{t})\cr
\text{Ham}_{(\mathbf{e}_{\infty^{(2)},1})}(q,p)=&\frac{-p^3+P_1(q)p^2-P_2(q)p+P_3(q)}{(t_{\infty^{(3)},2}- t_{\infty^{(2)},2})(t_{\infty^{(1)},2}- t_{\infty^{(2)},2})} +\hbar\frac{p-t_{\infty^{(3)},2}q}{t_{\infty^{(3)},2}- t_{\infty^{(2)},2}}+e_{2,1}(\mathbf{t})\cr
\text{Ham}_{(\mathbf{e}_{\infty^{(3)},1})}(q,p)=&\frac{-p^3+P_1(q)p^2-P_2(q)p+P_3(q)}{(t_{\infty^{(1)},2}- t_{\infty^{(3)},2})(t_{\infty^{(2)},2}- t_{\infty^{(3)},2})} -\hbar\frac{p-t_{\infty^{(2)},2}q}{t_{\infty^{(3)},2}- t_{\infty^{(2)},2}}+e_{3,1}(\mathbf{t})
\end{align}
In particular, this gives the evolutions of $ \mathcal{L}_{\mathbf{e}_{\infty^{(2)},1}}[q]$,  $\mathcal{L}_{\mathbf{e}_{\infty^{(2)},1}}[p]$,  $\mathcal{L}_{\mathbf{e}_{\infty^{(3)},1}}[q]$,  $\mathcal{L}_{\mathbf{e}_{\infty^{(3)},1}}[p]$
by the standard Hamiltonian derivatives. Using the shift $(q,p)\leftrightarrow (\check{q},\check{p})$, one may then easily compute the evolution of $ \mathcal{L}_{\mathbf{e}_{\infty^{(2)},1}}[\check{q}]$,  $\mathcal{L}_{\mathbf{e}_{\infty^{(2)},1}}[\check{p}]$,  $\mathcal{L}_{\mathbf{e}_{\infty^{(3)},1}}[\check{q}]$,  $\mathcal{L}_{\mathbf{e}_{\infty^{(3)},1}}[\check{p}]$
and obtain the Hamiltonians (up to purely time-dependent terms that are not determined by the dynamics)
\begin{align}
\text{Ham}_{\mathbf{e}_{\infty^{(1)},1}}(\check{q},\check{p})=&\frac{t_{\infty^{(3)},2}-t_{\infty^{(2)},2}}{\Delta}\text{Ham}_{(\tau)}(\check{q},\check{p}),\cr
\text{Ham}_{\mathbf{e}_{\infty^{(2)},1}}(\check{q},\check{p})=&\frac{t_{\infty^{(1)},2}-t_{\infty^{(3)},2}}{\Delta}\text{Ham}_{(\tau)}(\check{q},\check{p}),\cr
\text{Ham}_{\mathbf{e}_{\infty^{(3)},1}}(\check{q},\check{p})=&\frac{t_{\infty^{(2)},2}-t_{\infty^{(1)},2}}{\Delta}\text{Ham}_{(\tau)}(\check{q},\check{p}).
\end{align}
\normalsize{Note} that since the shift is time-dependent, it does not necessarily correspond to plugging the expression of $(q,p)$ in terms of $(\check{q},\check{p})$ into \eqref{AppendixHam1}. We then immediately get from \eqref{Hamtau}:
\begin{align}
     \text{Ham}_{(\tau)}(\check{q},\check{p})=&-\left(\check{q}\check{p}^2+\check{q}^2\check{p}-\tau \check{q}\check{p}-t_{\infty^{(2)},0}\check{p} +(t_{\infty^{(1)},0}+\hbar)\check{q}\right) 
\end{align}
i.e.
\begin{align}
    \hbar \partial_\tau \check{q}=&-(2\check{q}\check{p}+\check{q}^2 -\tau \check{q}-t_{\infty^{(2)},0})\cr
    \hbar \partial_\tau \check{p}=&\check{p}^2+2\check{q}\check{p}-\tau \check{p}+t_{\infty^{(1)},0}+\hbar 
\end{align}
For a compact expression for $d\tau$, write locally $A_i=t_{\infty^{(i)},2}$, $b_i=t_{\infty^{(i)},1}$ and $N=(b_2-b_3)A_1+(b_3-b_1)A_2+(b_1-b_2)A_3$. Then
\beq d\tau=\frac{dN}{\Delta}-\frac{\tau}{2}\left(\frac{dA_2-dA_1}{A_2-A_1}+\frac{dA_1-dA_3}{A_1-A_3}+\frac{dA_3-dA_2}{A_3-A_2}\right).\eeq
\normalsize{It} is then a straightforward but long computation to obtain $d \text{Ham}_{(\tau)}(\check{q},\check{p}) \wedge d\tau$ in terms of $(q,p,\mathbf{t})$. Similarly the shift in coordinates allows one to get $d\check{q}\wedge d\check{p}$ in terms of $(q,p,\mathbf{t})$ so that one can compute $\Omega_2:=d\, \text{Ham}_{(\tau)}(\check{q},\check{p}) \wedge d\tau +\hbar d\check{q}\wedge d\check{p}$ in terms of $(q,p,\mathbf{t})$.

\sloppy{On the other side, expression of the Hamiltonians $\left(\text{Ham}_{\mathbf{e}_{\infty^{(i)},k}}(q,p)\right)_{1\leq i\leq 3,1\leq k\leq 2}$ are known from \eqref{AppendixHam1} and \eqref{OtherHamiltonians} so that one can immediately compute $\Omega=\hbar dq\wedge dp + \underset{i=1}{\overset{3}{\sum}}\underset{k=1}{\overset{2}{\sum}} d\,\text{Ham}_{\mathbf{e}_{\infty^{(i)},k}}(q,p) \wedge d t_{\infty^{(i)},k}$ in terms of $(q,p,\mathbf{t})$. Then, one finally observes that $\Omega+\Omega_2=0$, including all purely time-dependent coefficients.}
\normalsize{This} identity can be checked by expanding both two-forms in the $28$ wedge products of the eight coordinates $(q,p,\mathbf{t})$. Each coefficient of $\Omega+\Omega_2$ vanishes, using $t_{\infty^{(1)},0}+t_{\infty^{(2)},0}+t_{\infty^{(3)},0}=0$ and a coherent branch $\Delta$.

\section{Proof of \autoref{TheoDualHamiltonians}}\label{AppendixIdentificationHamiltonian}
In this appendix, we prove that the identification of times/monodromies/Darboux coordinates following from the duality of spectral curves (\autoref{TheoDualitySpecCurves}) is compatible with the general Hamiltonian evolutions of \autoref{Defs} and \autoref{TheoHamP4}. We shall proceed first in the trivial directions $(\mathcal{L}_{\boldsymbol{\beta}_{\infty,1}},\mathcal{L}_{\boldsymbol{\beta}_{\infty,2}},\mathcal{L}_{\boldsymbol{\beta}_{\text{dil}}},\mathcal{L}_{\boldsymbol{\beta}_{\text{transl}}})$ of the tangent space and then deal with the direction $\partial_{X_1}$. Note that this final direction could also be seen as a consequence of the identification in reduced coordinates given by  \autoref{TheoDualReducedHamiltonian} but for completeness we provide a direct proof.
\subsection{Trivial directions}
Let us first observe that the identification of both sets of times implies at the level of tangent spaces that
\begin{align}
    \partial_{X_1}=&\partial_{t_{\infty^{(2)},1}}\cr
    \partial_{s_{\infty^{(1)},2}}=&\frac{1}{(s_{\infty^{(1)},2})^2}\partial_{t_{\infty^{(3)},2}}+\frac{s_{\infty^{(1)},1}}{(s_{\infty^{(1)},2})^2}\partial_{t_{\infty^{(3)},1}} \cr
    =&(t_{\infty^{(3)},2}-t_{\infty^{(2)},2})^2\partial_{t_{\infty^{(3)},2}}+(t_{\infty^{(3)},2}-t_{\infty^{(2)},2})t_{\infty^{(3)},1}\partial_{t_{\infty^{(3)},1}}\cr
    \partial_{s_{\infty^{(2)},2}}=&\frac{1}{(s_{\infty^{(2)},2})^2}\partial_{t_{\infty^{(1)},2}}+\frac{s_{\infty^{(2)},1}}{(s_{\infty^{(2)},2})^2}\partial_{t_{\infty^{(1)},1}} \cr
    =&(t_{\infty^{(1)},2}-t_{\infty^{(2)},2})^2\partial_{t_{\infty^{(1)},2}}+(t_{\infty^{(1)},2}-t_{\infty^{(2)},2})t_{\infty^{(1)},1}\partial_{t_{\infty^{(1)},1}}\cr
\partial_{s_{\infty^{(1)},1}}=&-\frac{1}{s_{\infty^{(1)},2}}\partial_{t_{\infty^{(3)},1}}=(t_{\infty^{(3)},2}-t_{\infty^{(2)},2})\partial_{t_{\infty^{(3)},1}}\cr
\partial_{s_{\infty^{(2)},1}}=&-\frac{1}{s_{\infty^{(2)},2}}\partial_{t_{\infty^{(1)},1}}=(t_{\infty^{(1)},2}-t_{\infty^{(2)},2})\partial_{t_{\infty^{(1)},1}}
\end{align}
Moreover, we have $P=q$ and $Q=p-t_{\infty^{(2)},2}q$, i.e. $\mathcal{L}_{\boldsymbol{\beta}}[Q]=\mathcal{L}_{\boldsymbol{\beta}}[p] -t_{\infty^{(2)},2}\mathcal{L}_{\boldsymbol{\beta}}[q]-q\mathcal{L}_{\boldsymbol{\beta}}[t_{\infty^{(2)},2}]$ and $\mathcal{L}_{\boldsymbol{\beta}}[P]=\mathcal{L}_{\boldsymbol{\beta}}[q]$ for any deformation vector. Finally, we have from \eqref{TrivialVectorsP4}:
\begin{align}\label{CorrespondanceTrivialDirections}
    \mathcal{L}_{\boldsymbol{\beta}_{\infty,1}}=&\hbar \partial_{s_{\infty^{(1)},1}}+\hbar \partial_{s_{\infty^{(2)},1}}=\hbar(t_{\infty^{(1)},2}-t_{\infty^{(2)},2})\partial_{t_{\infty^{(1)},1}}+ \hbar(t_{\infty^{(3)},2}-t_{\infty^{(2)},2})\partial_{t_{\infty^{(3)},1}}\cr&=\mathcal{L}_{\mathbf{u}_{\infty,1}}-t_{\infty^{(2)},2}\mathcal{L}_{\mathbf{v}_{\infty,1}}\cr
\mathcal{L}_{\boldsymbol{\beta}_{\infty,2}}=&\hbar \partial_{s_{\infty^{(1)},2}}+\hbar \partial_{s_{\infty^{(2)},2}}\cr
=& \hbar(t_{\infty^{(1)},2}-t_{\infty^{(2)},2})\left((t_{\infty^{(1)},2}-t_{\infty^{(2)},2})\partial_{t_{\infty^{(1)},2}}+t_{\infty^{(1)},1}\partial_{t_{\infty^{(1)},1}}\right)\cr&
+\hbar(t_{\infty^{(3)},2}-t_{\infty^{(2)},2})\left( (t_{\infty^{(3)},2}-t_{\infty^{(2)},2})\partial_{t_{\infty^{(3)},2}}+t_{\infty^{(3)},1}\partial_{t_{\infty^{(3)},1}}\right)\cr
=&-\frac{1}{2}(t_{\infty^{(3)},2}t_{\infty^{(1)},1}+t_{\infty^{(1)},2}t_{\infty^{(3)},1})\mathcal{L}_{\mathbf{v}_{\infty,1}}-(t_{\infty^{(1)},2}t_{\infty^{(3)},2}-(t_{\infty^{(2)},2})^2)\mathcal{L}_{\mathbf{v}_{\infty,2}}\cr&
+\frac{1}{2}(t_{\infty^{(1)},1}+t_{\infty^{(3)},1})\mathcal{L}_{\mathbf{u}_{\infty,1}}-\frac{1}{2}(2t_{\infty^{(2)},2}-t_{\infty^{(1)},2}-t_{\infty^{(3)},2})\mathcal{L}_{\mathbf{u}_{\infty,2}}+\frac{1}{2}\mathcal{L}_{\mathbf{a}_{2}}\cr
\mathcal{L}_{\boldsymbol{\beta}_{\text{dil}}}=&\hbar(
2 s_{\infty^{(1)},2} \partial_{s_{\infty^{(1)},2}}+2 s_{\infty^{(2)},2} \partial_{s_{\infty^{(2)},2}}+s_{\infty^{(1)},1} \partial_{s_{\infty^{(1)},1}}+ s_{\infty^{(2)},1} \partial_{s_{\infty^{(2)},1}}- X_1\partial_{X_1}  )\cr
=&  -2\hbar(t_{\infty^{(1)},2}-t_{\infty^{(2)},2})\partial_{t_{\infty^{(1)},2}} -2\hbar(t_{\infty^{(3)},2}-t_{\infty^{(2)},2})\partial_{t_{\infty^{(3)},2}}\cr&
-\hbar( t_{\infty^{(1)},1} \partial_{t_{\infty^{(1)},1}}+ t_{\infty^{(2)},1} \partial_{t_{\infty^{(2)},1}}  + t_{\infty^{(3)},1} \partial_{t_{\infty^{(3)},1}}) \cr 
=&2t_{\infty^{(2)},2}\mathcal{L}_{\mathbf{v}_{\infty,2}}-\mathcal{L}_{\mathbf{u}_{\infty,2}}\cr
\mathcal{L}_{\boldsymbol{\beta}_{\text{transl}}}=&\hbar(s_{\infty^{(1)},2}\partial_{s_{\infty^{(1)},1}}+s_{\infty^{(2)},2}\partial_{s_{\infty^{(2)},1}}-\partial_{X_1})
 =  -\hbar \left( \partial_{t_{\infty^{(1)},1}} + \partial_{t_{\infty^{(2)},1}} + \partial_{t_{\infty^{(3)},1}}  \right) \cr
 = &-\mathcal{L}_{\mathbf{v}_{\infty,1}}
 \end{align}
so that one can verify that the duality gives:
\begin{align}\label{Lbetainfty1}  \mathcal{L}_{\boldsymbol{\beta}_{\infty,1}}[Q]=&\mathcal{L}_{\boldsymbol{\beta}_{\infty,1}}[p] -t_{\infty^{(2)},2}\mathcal{L}_{\boldsymbol{\beta}_{\infty,1}}[q]-q\mathcal{L}_{\boldsymbol{\beta}_{\infty,1}}[t_{\infty^{(2)},2}]\cr
=& \mathcal{L}_{\mathbf{u}_{\infty,1}}[p]-t_{\infty^{(2)},2}\mathcal{L}_{\mathbf{v}_{\infty,1}}[p] -t_{\infty^{(2)},2}( \mathcal{L}_{\mathbf{u}_{\infty,1}}[q]-t_{\infty^{(2)},2}\mathcal{L}_{\mathbf{v}_{\infty,1}}[q])\cr&
-q(\mathcal{L}_{\mathbf{u}_{\infty,1}}[t_{\infty^{(2)},2}]-t_{\infty^{(2)},2}\mathcal{L}_{\mathbf{v}_{\infty,1}}[t_{\infty^{(2)},2}])\cr
\overset{\text{Th.} \ref{Defs}}{=}&0\cr
\mathcal{L}_{\boldsymbol{\beta}_{\infty,1}}[P]=&\mathcal{L}_{\boldsymbol{\beta}_{\infty,1}}[q]= \mathcal{L}_{\mathbf{u}_{\infty,1}}[q]-t_{\infty^{(2)},2}\mathcal{L}_{\mathbf{v}_{\infty,1}}[q]\overset{\text{Th.} \ref{Defs}}{=}-\hbar
\end{align}
and
{\small
\begin{align}\label{Lbetainfty2} 
\mathcal{L}_{\boldsymbol{\beta}_{\infty,2}}[Q]=&\mathcal{L}_{\boldsymbol{\beta}_{\infty,2}}[p] -t_{\infty^{(2)},2}\mathcal{L}_{\boldsymbol{\beta}_{\infty,2}}[q]-q\mathcal{L}_{\boldsymbol{\beta}_{\infty,2}}[t_{\infty^{(2)},2}]\cr
=&\mathcal{L}_{\boldsymbol{\beta}_{\infty,2}}[p] -t_{\infty^{(2)},2}\mathcal{L}_{\boldsymbol{\beta}_{\infty,2}}[q]\cr
=&-\frac{1}{2}(t_{\infty^{(3)},2}t_{\infty^{(1)},1}+t_{\infty^{(1)},2}t_{\infty^{(3)},1})\mathcal{L}_{\mathbf{v}_{\infty,1}}[p]-(t_{\infty^{(1)},2}t_{\infty^{(3)},2}-(t_{\infty^{(2)},2})^2)\mathcal{L}_{\mathbf{v}_{\infty,2}}[p]\cr&
+\frac{1}{2}(t_{\infty^{(1)},1}+t_{\infty^{(3)},1})\mathcal{L}_{\mathbf{u}_{\infty,1}}[p]-\frac{1}{2}(2t_{\infty^{(2)},2}-t_{\infty^{(1)},2}-t_{\infty^{(3)},2})\mathcal{L}_{\mathbf{u}_{\infty,2}}[p]+\frac{1}{2}\mathcal{L}_{\mathbf{a}_{2}}[p]\cr
&+t_{\infty^{(2)},2}\Big[\frac{1}{2}(t_{\infty^{(3)},2}t_{\infty^{(1)},1}+t_{\infty^{(1)},2}t_{\infty^{(3)},1})\mathcal{L}_{\mathbf{v}_{\infty,1}}[q]+(t_{\infty^{(1)},2}t_{\infty^{(3)},2}-(t_{\infty^{(2)},2})^2)\mathcal{L}_{\mathbf{v}_{\infty,2}}[q]\cr&
-\frac{1}{2}(t_{\infty^{(1)},1}+t_{\infty^{(3)},1})\mathcal{L}_{\mathbf{u}_{\infty,1}}[q]+\frac{1}{2}(2t_{\infty^{(2)},2}-t_{\infty^{(1)},2}-t_{\infty^{(3)},2})\mathcal{L}_{\mathbf{u}_{\infty,2}}[q]-\frac{1}{2}\mathcal{L}_{\mathbf{a}_{2}}[q]\Big]\cr
\overset{\text{Th.} \ref{Defs}}{=}& \textcolor{blue}{-\frac{\hbar}{2}(t_{\infty^{(3)},2}t_{\infty^{(1)},1}+t_{\infty^{(1)},2}t_{\infty^{(3)},1})}-\hbar (\textcolor{red}{t_{\infty^{(1)},2}t_{\infty^{(3)},2}}\textcolor{yellow}{-(t_{\infty^{(2)},2})^2})q
\textcolor{green}{-\frac{\hbar}{2}(2t_{\infty^{(2)},2}-t_{\infty^{(1)},2}-t_{\infty^{(3)},2})p}\cr&
+\frac{\hbar}{2}\Big[\textcolor{green}{-(t_{\infty^{(1)},2}+t_{\infty^{(3)},2})p}+\textcolor{red}{2t_{\infty^{(1)},2}t_{\infty^{(3)},2}q}\textcolor{blue}{+t_{\infty^{(3)},1}t_{\infty^{(1)},2}+t_{\infty^{(3)},2}t_{\infty^{(1)},1}}\Big]\cr&
+t_{\infty^{(2)},2}\Big[\frac{\hbar}{2}(t_{\infty^{(1)},1}+t_{\infty^{(3)},1})-\frac{\hbar}{2}(\textcolor{yellow}{2t_{\infty^{(2)},2}}\textcolor{orange}{-t_{\infty^{(1)},2}-t_{\infty^{(3)},2}})q\cr&
-\frac{\hbar}{2}\left(\textcolor{orange}{(t_{\infty^{(1)},2}+t_{\infty^{(3)},2})q}\textcolor{green}{-2p}+t_{\infty^{(1)},1}+t_{\infty^{(3)},1}\right) \Big]
\cr
=&0\cr
\mathcal{L}_{\boldsymbol{\beta}_{\infty,2}}[P]=&\mathcal{L}_{\boldsymbol{\beta}_{\infty,2}}[q]=-\frac{1}{2}(t_{\infty^{(3)},2}t_{\infty^{(1)},1}+t_{\infty^{(1)},2}t_{\infty^{(3)},1})\mathcal{L}_{\mathbf{v}_{\infty,1}}[q]-(t_{\infty^{(1)},2}t_{\infty^{(3)},2}-(t_{\infty^{(2)},2})^2)\mathcal{L}_{\mathbf{v}_{\infty,2}}[q]\cr&
+\frac{1}{2}(t_{\infty^{(1)},1}+t_{\infty^{(3)},1})\mathcal{L}_{\mathbf{u}_{\infty,1}}[q]-\frac{1}{2}(2t_{\infty^{(2)},2}-t_{\infty^{(1)},2}-t_{\infty^{(3)},2})\mathcal{L}_{\mathbf{u}_{\infty,2}}[q]+\frac{1}{2}\mathcal{L}_{\mathbf{a}_{2}}[q]\cr
\overset{\text{Th.} \ref{Defs}}{=}&-\frac{\hbar}{2}(t_{\infty^{(1)},1}+t_{\infty^{(3)},1})+\frac{\hbar }{2}(2t_{\infty^{(2)},2}-t_{\infty^{(1)},2}-t_{\infty^{(3)},2})q\cr&
+\frac{\hbar}{2}\left((t_{\infty^{(1)},2}+t_{\infty^{(3)},2})q-2p+(t_{\infty^{(1)},1}+t_{\infty^{(3)},1})\right)=\hbar t_{\infty^{(2)},2} q-\hbar p\cr
=&-\hbar Q 
\end{align}}

and
\begin{align}\label{Lbetadil}
   \mathcal{L}_{\boldsymbol{\beta}_{\text{dil}}} [Q] = & \mathcal{L}_{\boldsymbol{\beta}_{\text{dil}}}[p] -t_{\infty^{(2)},2}\mathcal{L}_{\boldsymbol{\beta}_{\text{dil}}}[q]-q\mathcal{L}_{\boldsymbol{\beta}_{\text{dil}}}[t_{\infty^{(2)},2}]\cr
   =&\mathcal{L}_{\boldsymbol{\beta}_{\text{dil}}}[p] -t_{\infty^{(2)},2}\mathcal{L}_{\boldsymbol{\beta}_{\text{dil}}}[q]\cr
   =&2t_{\infty^{(2)},2}\mathcal{L}_{\mathbf{v}_{\infty,2}}[p]-\mathcal{L}_{\mathbf{u}_{\infty,2}}[p] -2(t_{\infty^{(2)},2})^2\mathcal{L}_{\mathbf{v}_{\infty,2}}[q]+t_{\infty^{(2)},2}\mathcal{L}_{\mathbf{u}_{\infty,2}}[q]\cr
\overset{\text{Th.} \ref{Defs}}{=}&2\hbar t_{\infty^{(2)},2}q-\hbar p-\hbar t_{\infty^{(2)},2}q=-\hbar(p-t_{\infty^{(2)},2}q)\cr
=&-\hbar Q\cr
 \mathcal{L}_{\boldsymbol{\beta}_{\text{dil}}} [P] = & \mathcal{L}_{\boldsymbol{\beta}_{\text{dil}}} [q] =2t_{\infty^{(2)},2}\mathcal{L}_{\mathbf{v}_{\infty,2}}[q]-\mathcal{L}_{\mathbf{u}_{\infty,2}}[q]
 \overset{\text{Th.} \ref{Defs}}{=}\hbar q\cr=&\hbar P
\end{align}
and finally in the last trivial direction
\begin{align}\label{Lbetatransl}
    \mathcal{L}_{\boldsymbol{\beta}_{\text{transl}}} [Q] = &  -\mathcal{L}_{\mathbf{v}_{\infty,1}} [p] +  t_{\infty^{(2)},2} \mathcal{L}_{\mathbf{v}_{\infty,1}} [q] +  q\mathcal{L}_{\mathbf{v}_{\infty,1}} [t_{\infty^{(2)},2}] \overset{\text{Th.} \ref{Defs}}{=} -\hbar \cr
     \mathcal{L}_{\boldsymbol{\beta}_{\text{transl}}} [P] = & -\mathcal{L}_{\mathbf{v}_{\infty,1}} [q] \overset{\text{Th.} \ref{Defs}}{=} 0 
\end{align}
so that \eqref{Lbetainfty1}, \eqref{Lbetainfty2}, \eqref{Lbetadil}, \eqref{Lbetatransl} are consistent with \eqref{TrivialEvolutionsP4}.

\subsection{Non-trivial direction}
Let us finally look at the direction $\mathcal{L}_{X_1}$ in order to complement the tangent space. We have from \eqref{HamilP4General}:
\beq
    \text{Ham}^{(\text{P4})}_{X_1}=(Q-X_1)\left( P^2-R_1(Q)P+\frac{\hbar P}{Q-X_1} +R_2(Q)+\hbar s_{\infty^{(1)},2}\right)
\eeq
so that the evolutions of $(Q,P)$ are given by
\begin{align}
    \mathcal{L}_{X_1}[Q]=&(Q-X_1)\left(2P-R_1(Q)\right) +\hbar\cr
    \mathcal{L}_{X_1}[P]=&-P^2+R_1(Q)P-R_2(Q)-\hbar s_{\infty^{(1)},2}+(Q-X_1)\left(R_1'(Q)P-R_2'(Q)\right)
\end{align}
Thus, the identification from spectral duality gives from \eqref{DefR1R2}:
\begin{align}\label{R1R2Dual}
    R_1(\xi)=&\frac{t_{\infty^{(2)},0}}{\xi-t_{\infty^{(2)},1}}+\left(\frac{1}{t_{\infty^{(1)},2}-t_{\infty^{(2)},2}}  +\frac{1}{t_{\infty^{(3)},2}-t_{\infty^{(2)},2}}\right)\xi \cr&-\frac{t_{\infty^{(1)},1}}{t_{\infty^{(1)},2}-t_{\infty^{(2)},2}}-\frac{t_{\infty^{(3)},1}}{t_{\infty^{(3)},2}-t_{\infty^{(2)},2}}\cr
    R_2(\xi)
=& \frac{\xi^2-(t_{\infty^{(1)},1}+t_{\infty^{(3)},1})\xi +t_{\infty^{(1)},1}t_{\infty^{(3)},1} }{(t_{\infty^{(1)},2}-t_{\infty^{(2)},2})(t_{\infty^{(3)},2}-t_{\infty^{(2)},2})}-\frac{t_{\infty^{(1)},0}}{t_{\infty^{(3)},2}-t_{\infty^{(2)},2}}-\frac{t_{\infty^{(3)},0}}{t_{\infty^{(1)},2}-t_{\infty^{(2)},2}}
\end{align}
and
\begin{align}\label{CheckDualityNonTrivialDirection}\hbar \partial_{t_{\infty^{(2)},1}}[q]=&\mathcal{L}_{X_1}[q]= \mathcal{L}_{X_1}[P] \cr
=&-P^2+R_1(Q)P-R_2(Q)-\hbar s_{\infty^{(1)},2}+(Q-X_1)\left(R_1'(Q)P-R_2'(Q)\right)\cr
=&-q^2+qR_1(p-t_{\infty^{(2)},2}q)-R_2(p-t_{\infty^{(2)},2}q)+\frac{\hbar}{t_{\infty^{(3)},2}-t_{\infty^{(2)},2}}\cr
&+(p-t_{\infty^{(2)},2}q -t_{\infty^{(2)},1})\left(qR_1'(p-t_{\infty^{(2)},2}q)-R_2'(p-t_{\infty^{(2)},2}q)\right)\cr
\hbar \partial_{t_{\infty^{(2)},1}}[p]=&\mathcal{L}_{X_1}[p]=\mathcal{L}_{X_1}[Q+t_{\infty^{(2)},2}P]=\mathcal{L}_{X_1}[Q]+t_{\infty^{(2)},2} \mathcal{L}_{X_1}[P]\cr
=& (Q-X_1)\left(2P-R_1(Q)\right) +\hbar\cr
&+ t_{\infty^{(2)},2}\Big[-P^2+R_1(Q)P-R_2(Q)-\hbar s_{\infty^{(1)},2}+(Q-X_1)\left(R_1'(Q)P-R_2'(Q)\right)\Big]\cr
=&  \left(p-t_{\infty^{(2)},2}q-t_{\infty^{(2)},1}\right)\left(2q-R_1(p-t_{\infty^{(2)},2}q)\right)\cr
&+t_{\infty^{(2)},2}\Big[-q^2+qR_1(p-t_{\infty^{(2)},2}q)-R_2(p-t_{\infty^{(2)},2}q)\cr
&+\left(p-t_{\infty^{(2)},2}q-t_{\infty^{(2)},1}\right)(qR_1'(p-t_{\infty^{(2)},2}q)-R_2'(p-t_{\infty^{(2)},2}q)  \Big]+\frac{\hbar t_{\infty^{(3)},2}}{t_{\infty^{(3)},2}-t_{\infty^{(2)},2}}\cr
\end{align}
Finally we recall that from \autoref{Defs} we have the expression of $\text{Ham}_{(\mathbf{e}_{\infty^{(2)},1})}(q,p)$ that provides
\begin{align}\label{Evolutionr21} \hbar \partial_{t_{\infty^{(2)},1}}q=&\frac{-3p^2+2P_1(q)p-P_2(q)}{(t_{\infty^{(3)},2}- t_{\infty^{(2)},2})(t_{\infty^{(1)},2}- t_{\infty^{(2)},2})} +\frac{\hbar}{t_{\infty^{(3)},2}- t_{\infty^{(2)},2}}\cr
\hbar \partial_{t_{\infty^{(2)},1}}p=&\frac{-P_1'(q)p^2+P_2'(q)p-P_3'(q)}{(t_{\infty^{(3)},2}- t_{\infty^{(2)},2})(t_{\infty^{(1)},2}- t_{\infty^{(2)},2})} +\hbar\frac{t_{\infty^{(3)},2}}{t_{\infty^{(3)},2}- t_{\infty^{(2)},2}}
\end{align}
It is then a direct but cumbersome computation to check that using \eqref{R1R2Dual} in \eqref{CheckDualityNonTrivialDirection}, the last expressions match with \eqref{Evolutionr21} and therefore the evolutions are consistent on both sides.

\newpage
\addcontentsline{toc}{section}{References}
\bibliographystyle{plain}
\bibliography{Biblio}

@article{marchal2023hamiltonian,
doi = {10.1088/1361-6544/ae19dd},
year = {2025},
volume = {38},
number = {11},
pages = {115018},
author = {O. Marchal and N. Orantin and M. Alameddine},
title = {Hamiltonian representation of isomonodromic deformations of general rational connections on $\mathfrak{gl}_2(\mathbb{C})$},
journal = {Nonlinearity}
}

@article{marchal2023isomonodromic,
    author = {O. Marchal and M. Alameddine},
    title = "{Isomonodromic and isospectral deformations of meromorphic connections: the $\mathfrak{sl}_2(\mathbb{C})$ case}",
    doi = "10.1088/1361-6544/ad7b96",
    journal = "Nonlinearity",
    volume = "37",
    number = "11",
    pages = "115006",
    year = "2024"
}

@article{MarchalAlameddineP1Hierarchy2023,
    author = {O. Marchal and M. Alameddine},
    title = "{Hamiltonian Representation of Isomonodromic Deformations of Twisted Rational Connections: The Painlev{\'e} 1 Hierarchy}",
    doi = "10.1007/s00220-024-05187-0",
    journal = "Commun. Math. Phys.",
    volume = "406",
    number = "1",
    pages = "12",
    year = "2025"
}

@article{FockRosly99,
	Author = {V.~V.~Fock and A.~A.~Rosly},
	Journal = {Am. Math. Soc. Transl.},
	Title = {Poisson structure on moduli of flat connections on {R}iemann surfaces and $r$-matrix},
	Volume = {191},
	Year = {1999}}

@article{Goldman84,
	Author = {W.~Goldman},
	Journal = {Adv. Math. },
	Pages = {200-225},
	Title = {The symplectic nature of the fundamental group of surfaces},
	Volume = {54},
	Year = {1984}}

@article{AtiyahBott,
	Author = {M.~Atiyah and R.~Bott},
	Journal = {Philos. Trans. Royal Soc. A},
	Number = {1505},
	Pages = {523--615},
	Title = {The {Y}ang-{M}ills equations over {R}iemann surfaces},
	Volume = {308},
	Doi = {10.1098/rsta.1983.0017},
	Year = {1983}}

@phdthesis{BoalchThesis,
	Author = {P. Boalch},
	School = {Oxford D.Phil.},
	Title = {Symplectic geometry and isomonodromic deformations },
	Year = {1999}}

@misc{bergre2009determinantal,
	Archiveprefix = {arXiv},
	Author = {M.~Berg\`ere and B.~Eynard},
	Eprint = {0901.3273},
	Note = {arXiv:0901.3273},
	Primaryclass = {math-ph},
	Title = {Determinantal formulae and loop equations},
	Year = {2009}}

@article{IwakiMarchalSaenz,
	Author = {K.~Iwaki and O.~Marchal and A.~Saenz},
	Doi = {10.1016/j.geomphys.2017.10.009},
	Fjournal = {Journal of Geometry and Physics},
	Issn = {0393-0440},
	Journal = {J. Geom. Phys.},
	Mrclass = {34M55 (34E20 34M56 60B20)},
	Pages = {16--54},
	Title = {Painlev\'{e} equations, topological type property and reconstruction by the topological recursion},
	Volume = {124},
	Year = {2018}}

@article{MOsl2,
	Author = {O.~Marchal and N.~Orantin},
	Doi = {10.1063/5.0002260},
	Fjournal = {Journal of Mathematical Physics},
	Issn = {0022-2488},
	Journal = {J. Math. Phys.},
	Mrclass = {34M56 (32G34 81Q20)},
	Number = {6},
	Pages = {061506, 33},
	Title = {Isomonodromic deformations of a rational differential system and reconstruction with the topological recursion: the {$\mathfrak{sl}_2$} case},
	Volume = {61},
	Year = {2020}}

@article{Quantization_2021,
	Author = {B.~Eynard and E.~Garcia-Failde and O.~Marchal and N.~Orantin},
	Title = {Quantization of classical spectral curves via topological recursion},
    journal={Commun. Math. Phys.},
    fjournal={Communications in Mathematical Physics},
    volume={405},
    number={116},
	Year = {2024}
}

@article{JimboMiwaUeno,
	Author = {M.~Jimbo and T.~Miwa and K. Ueno},
	Fjournal = {Physica D. Nonlinear Phenomena},
	Journal = {Phys. D},
	Number = {2},
	Pages = {306--352},
	Title = {Monodromy preserving deformation of linear ordinary differential equations with rational coefficients: I. General theory and $\tau$-function},
	Volume = {2},
	Year = {1981}}

@article{JimboMiwa,
	Author = {M.~Jimbo and T.~Miwa},
	Doi = {10.1016/0167-2789(81)90021-X},
	Fjournal = {Physica D. Nonlinear Phenomena},
	Issn = {0167-2789},
	Journal = {Phys. D},
	Mrclass = {34A20 (14K25 58A15 58F07 81C05)},
	Mrnumber = {625446},
	Mrreviewer = {V. A. Golubeva},
	Number = {3},
	Pages = {407--448},
	Title = {Monodromy preserving deformation of linear ordinary differential equations with rational coefficients. {II}},
	Url = {https://doi.org/10.1016/0167-2789(81)90021-X},
	Volume = {2},
	Year = {1981}}

@article{BergereBorotEynard,
	Author = {M.~Berg\`ere and G.~Borot and B.~Eynard},
	Doi = {10.1007/s00023-014-0391-8},
	Fjournal = {Annales Henri Poincar\'{e}. A Journal of Theoretical and Mathematical Physics},
	Issn = {1424-0637},
	Journal = {Ann. Henri Poincar\'{e}},
	Number = {12},
	Pages = {2713--2782},
	Title = {Rational differential systems, loop equations, and application to the $q^{th}$ reductions of {KP}},
	Volume = {16},
	Year = {2015}}

@incollection{Norbury_survey,
	Author = {P.~Norbury},
	Booktitle = {String-{M}ath 2014},
	Mrclass = {14H81 (05A15 14N10 81S10)},
	Mrnumber = {3524233},
	Mrreviewer = {Hsian-Hua Tseng},
	Pages = {41--65},
	Publisher = {Amer. Math. Soc., Providence, RI},
	Series = {Proc. Sympos. Pure Math.},
	Title = {Quantum curves and topological recursion},
	Volume = {93},
	Year = {2016}}

@article{AKEMANN1997475,
title = {Universal correlators for multi-arc complex matrix models},
fjournal = {Nuclear Physics B},
journal={Nucl. Phys. B.},
volume = {507},
number = {1},
pages = {475--500},
year = {1997},
author = {G. Akemann},
}

@article{AMBJORN1993127,
title = {Matrix model calculations beyond the spherical limit},
fjournal = {Nuclear Physics B},
journal={Nucl. Phys. B.},
volume = {404},
number = {1},
pages = {127--172},
year = {1993},
author = {J. Ambjørn and L. Chekhov and C.F. Kristjansen and Yu. Makeenko},
}

@article{E1MM,
	Author = {B.~Eynard},
	Journal = {JHEP},
	Title = {Topological expansion for the {$1$}-hermitian matrix model correlation functions},
	Volume = {0411:031},
	Year = {2004}}

@misc{EynardSolutionsLoopEquations,
    author={B. Eynard},
    Archiveprefix = {arXiv},
    Eprint = {1909.09372},
	Note = {arXiv:1909.09372},
	Primaryclass = {math-ph},
    Title = {Solutions of loop equations are random matrices},
	Year = {2019}
}

@misc{TRReview,
    author={B. Eynard},
    Archiveprefix = {arXiv},
    Eprint = {1412.3286},
	Note = {arXiv:1412.3286},
	Primaryclass = {math-ph},
    Title = {A short overview of the ``Topological recursion"},
	Year = {2014}
}

@article{EO2MM,
	Author = {B.~Eynard and N.~Orantin},
	Journal = {J. Phys. A: Math. Theor.},
	Title = {Topological expansion of mixed correlations in the hermitian 2 matrix model and $x-y$ symmetry of the {$F_g$} invariants},
	Volume = {41},
	Year = {2008}}

@misc{EOxy,
	Author = {B.~Eynard and N.~Orantin},
	Note = {arXiv:1311.4993},
	Title = {About the {$x$}-{$y$} symmetry of the {$F_g$} algebraic invariants},
	Year = {2013},
    Archiveprefix = {arXiv},
	Eprint = {1311.4993},
	Primaryclass = {math-ph},
}

@article{CEO06,
	Author = {L.~Chekhov and B.~Eynard and N.~Orantin},
	Doi = {10.1088/1126-6708/2006/12/053},
	Fjournal = {Journal of High Energy Physics},
	Issn = {1126-6708},
	Journal = {JHEP},
	Mrclass = {81T45 (15A52)},
	Number = {12},
	Pages = {053, 31},
	Title = {Free energy topological expansion for the 2-matrix model},
	Volume = {2006},
	Year = {2006}}

@inproceedings{MulaseMM,
	Author = {M.~Mulase},
	Editor = {K.~Fukaya and M.~Furuta and T.~Kohno and D.~Kotschick},
	booktitle = {Topology, {G}eometry and {F}ield theory},
	Publisher = {World Scientific},
	Title = {Matrix models and integrable systems},
	Year = {1994}
}

@article{Guionnet,
    author = {G. Borot and A. Guionnet and K.K. Kozlowski},
    title = "{Large-{$N$} {A}symptotic {E}xpansion for Mean Field Models with {C}oulomb Gas Interaction}",
    fjournal = {International Mathematics Research Notices},
    journal={Int. Math. Res. Not.},
    volume = {2015},
    number = {20},
    year = {2015},
}

@article{BorotGuionnet,
    author = {G. Borot and A. Guionnet},
    title = "{Asymptotic expansion of matrix models in the multi-cut regime}",
    journal = "Forum Math. Sigma",
    volume = "12",
    year = "2024"
}

@article{Borot2011AsymptoticEO,
  title={Asymptotic {E}xpansion of $\beta$ {M}atrix {M}odels in the {O}ne-cut {R}egime},
  author={G. Borot and A. Guionnet},
  fjournal={Communications in Mathematical Physics},
  journal={Commun. Math. Phys.},
  year={2013},
  volume={317},
  number={2},
  pages={447--483},
}

@article{EO07,
	Author = {B.~Eynard and N.~Orantin},
	Journal = {Commun. Number Theory Phys.},
	Number = {2},
	Title = {Invariants of algebraic curves and topological expansion},
	Volume = {1},
	Year = {2007}}

@article{Boalch2012,
	Author = {P. Boalch},
	Fjournal = {Publications Math\'{e}matiques de l'IH\'{E}S},
	Journal = {Publ. Math. IH\'{E}S},
	Pages = {1--68},
	Title = {Simply-laced isomonodromy systems},
	Volume = {116},
	Year = {2012}}

@article{Boalch2001,
	Author = {P. Boalch},
	Fjournal = {Advances in Mathematics},
	Journal = {Adv. Math.},
	Number = {2},
	Pages = {137--205},
	Title = {Symplectic Manifolds and Isomonodromic Deformations},
	Volume = {163},
	Year = {2001}}

@misc{Boalch2022,
	Archiveprefix = {arXiv},
	Author = {P. Boalch and J. Dou{\c c}ot and G. Rembado},
	Eprint = {2209.12695},
	Note = {arXiv:2209.12695},
	Primaryclass = {math.AG},
	Title = {Twisted local wild mapping class groups: configuration spaces, fission trees and complex braids},
	Year = {2022}}

@article{okamoto1979deformation,
	Author = {K. Okamoto},
	Journal = {J. Fac. Sci. Univ. Tokyo Sect. IA Math},
	Pages = {501--518},
	Title = {D{\'e}formation d'une {\'e}quation diff{\'e}rentielle lin{\'e}aire avec une singularit{\'e} irr{\'e}guliere sur un tore},
	Volume = {26},
	Year = {1979}
}

@article{Okamoto1986,
title={Studies on the {P}ainlev\'{e} equations: {I}.-{S}ixth {P}ainlev\'{e} equation {PVI}},
 volume={146}, 
 journal={Ann. di {M}at. {P}ura ed {A}ppl.},
 fjournal={Annali di Matematica Pura ed Applicata},
 author={K. Okamoto},
 year={1986},
 pages={337--381}
}

@article{Okamoto1986Iso,
title={Isomonodromic deformation and Painlev\'{e} equations, and the {G}arnier system},
 volume={33}, 
 journal={J. Fac. Sci. Univ. Tokyo},
 author={K. Okamoto},
 year={1986},
 pages={575--618}
}

@article{Okamoto1980,
title={Polynomial {H}amiltonians associated with {P}ainlev\'{e} equations.},
 volume={56}, 
 journal={Proc. Japan Acad.},
 author={K. Okamoto},
 year={1980},
 pages={264--268}
}

@article{Garnier,
	Author = {R. Garnier},
	Journal = {Annales scientifiques de l'E.N.S.},
	Pages = {177--307},
	Title = {Solution du probl\`{e}me de Riemann pour les syst\`{e}mes diff{\'e}rentiels lin\'{e}aires du second ordre},
	Volume = {43},
	Year = {1927}}

@article{HarnadHurtubise,
	Author = {M.R. Adams and J. Harnad and J. Hurtubise},
	Journal = {Lett. Math. Phys.},
	Pages = {299--308},
	Title = {Dual moment maps into loop algebras},
	Volume = {20},
	Year = {1990}}

@article{Fuchs,
	Author = {R. Fuchs},
	Journal = {Comptes Rendus},
	Pages = {555--558},
	Title = {Sur quelques \'{e}quations diff{\'e}rentielles lin\'{e}aires du second ordre},
	Volume = {141},
	Year = {1905}}

@article{Gambier,
	Author = {B. Gambier},
	Journal = {Acta Math.},
	Pages = {1--55},
	Title = {Sur les \'{e}quations diff\'{e}rentielles du second ordre et du premier degr\'{e} dont l'int\'{e}grale g\'{e}n\'{e}rale est \`{a} points critiques fixes},
	Volume = {33},
	Year = {1910}}

@article{Painleve,
	Author = {P. Painlev\'{e}},
	Journal = {Acta Math.},
	Pages = {1--85},
	Title = {Sur les \'{e}quations diff\'{e}rentielles du second ordre et d'ordre sup\'{e}rieur dont l'int\'{e}grale g\'{e}n\'{e}rale est uniforme},
	Volume = {25},
	Year = {1902}}

@article{Picard,
	Author = {E. Picard},
	Journal = {J. Math. Pures Appl.},
	Pages = {135--319},
	Title = {M\'{e}moire sur la th\'{e}orie des fonctions alg\'{e}briques de deux variables},
	Volume = {5},
	Year = {1889}}

@article{woodhouse2007duality,
	Author = {N.M.J. Woodhouse},
	Fjournal = {Journal of Geometry and Physics},
	Journal = {J. Geom. Phys.},
	Number = {4},
	Pages = {1147--1170},
	Title = {Duality for the general isomonodromy problem},
	Volume = {57},
	Year = {2007}}

@article{Bertola2001DualityBP,
  title={Duality, {B}iorthogonal {P}olynomials and {M}ulti-{M}atrix {M}odels},
  author={M. Bertola and B. Eynard and J. Harnad},
  journal={Commun. Math. Phys.},
  fjournal={Communications in Mathematical Physics},
  year={2002},
  volume={229},
  pages={73--120},
}

@article{Bertola2003,
  title={Duality of spectral curves arising in two-matrix models},
  author={M. Bertola and B. Eynard and J. Harnad},
  journal={Theor. Math. Phys.},
  fjournal={Theoretical and Mathematical Physics},
  year={2003},
  volume={134},
  pages={27--38},
}

@article{Balser,
author = {Balser, W. and Jurkat, W.B. and Lutz, D.A.},
title = {On the {R}eduction of {C}onnection {P}roblems for {D}ifferential {E}quations with an {I}rregular {S}ingular {P}oint to {O}nes with {O}nly {R}egular {S}ingularities, {I}},
journal={SIAM J. Math. Anal.},
fjournal = {SIAM Journal on Mathematical Analysis},
volume = {12},
number = {5},
pages = {691--721},
year = {1981},
}

@article{HarnadIts,
author = {J. Harnad and A. R. Its},
title = {Integrable {F}redholm Operators and Dual Isomonodromic Deformations},
journal={Commun. Math. Phys.},
fjournal = {Communications in Mathematical Physics},
volume = {226},
pages = {497--530},
year = {2002},
}

@article{HaraokaMiddleConv,
author = {Y. Haraoka and G. Filipuk},
title = {Middle convolution and deformation for {F}uchsian systems},
journal={J. Lond. Math. Soc.},
fjournal = {Journal of the London Mathematical Society},
volume = {76},
number = {2},
pages = {438--450},
year = {2007}
}

@article{BibiloMiddle,
author = {Y. Bibilo and G. Filipuk},
title = {{Middle convolution and non-Schlesinger deformations}},
volume = {91},
journal={Proc. Jpn. Acad. A: Math. Sci.},
fjournal = {Proceedings of the Japan Academy, Series A, Mathematical Sciences},
number = {5},
pages = {66--69},
year = {2015},
}

@article{BertolaKorotkin2021,
	Author = {M. Bertola and D. Korotkin},
	Fjournal = {Communications in Mathematical Physics},
	Journal = {Commun. Math. Phys.},
	Pages = {245--290},
	Title = {Tau-functions and monodromy symplectomorphisms},
	Volume = {388},
	Year = 2021}

@misc{Weller2024,
	Archiveprefix = {arXiv},
	Author = {Q. Weller},
	Eprint = {2406.17081},
	Note = {arXiv:2406.17081},
	Primaryclass = {math-ph},
	Title = {The {L}aplace {T}ransform and {Q}uantum {C}urves},
	Year = {2024}
}

@article{ABDKSLogTR2024,
    author = {A. Alexandrov and B. Bychkov and P. Dunin-Barkowski and M. Kazarian and S. Shadrin},
    title = "{Log Topological Recursion Through the Prism of x-y Swap}",
    doi = "10.1093/imrn/rnae213",
    journal = "Int. Math. Res. Not.",
    volume = "2024",
    number = "21",
    pages = "13461--13487",
    year = "2024"
}

@article{Bertola:2004ws,
    author = "Bertola, M. and Eynard, B. and Harnad, J.",
    title = "{Semiclassical orthogonal polynomials, matrix models and isomonodromic tau functions}",
    journal = "Commun. Math. Phys.",
    volume = "263",
    pages = "401--437",
    year = "2006"
}

@article{MBertola_2003,
year = {2003},
volume = {36},
number = {12},
author = {M. Bertola and  B. Eynard and  J. Harnad},
title = {Partition functions for matrix models and isomonodromic tau functions},
journal={J. Phys. A Math. Gen.},
fjournal = {Journal of Physics A: Mathematical and General},
}

@Inbook{Morozov1999,
author="A. Morozov",
title="Matrix Models as Integrable Systems",
bookTitle="Particles and Fields",
year="1999",
publisher="Springer New York",
pages="127--210",
}

@article{Yamakawa2017TauFA,
	Author = {D. Yamakawa},
	Fjournal = {Josai Mathematical Monographs},
	Journal = {Josai Math. Monogr.},
	Page = {139-160},
	Title = {Tau functions and {H}amiltonians of isomonodromic deformations},
	Volume = {10},
	Year = {2017}}

@article{Yamakawa2019FundamentalTwoForms,
	Author = {D. Yamakawa},
	Fjournal = {Journal of Integrable Systems},
	Journal = {J. Integrable Syst.},
	Number = {1},
	Title = {Fundamental two-forms for isomonodromic deformations},
	Volume = {4},
	Year = {2019}}

@article{YamakawaMiddleConvolutions,
	Author = {D. Yamakawa},
	Journal = {Math. Ann.},
	Title = {Middle convolution and {H}arnad duality},
	Volume = {349},
    pages={215--262},
	Year = {2011}}

@article{Luu2019,
    author = {Luu, M.T.},
    title = {Spectral curve duality beyond the two-matrix model},
    fjournal = {Journal of Mathematical Physics},
    journal={J. Math. Phys.},
    volume = {60},
    number = {8},
    year = {2019},
}

@article{Birkhoff,
	Author = {G.D. Birkhoff},
	Fjournal = {Transactions of the {A}merican {M}athematical {S}ociety},
	Journal = {Trans. {A}m. {M}ath. {S}oc.},
	Number = {4},
	Pages = {436--470},
	Title = {Singular Points of Ordinary Linear Differential Equations},
	Volume = {10},
	Year = {1909}}

@book{Wasowbook,
	Author = {W.~Wasow},
	Isbn = {978-0486824581},
	Publisher = {Dover Publications},
	Title = {Asymptotic {E}xpansions for {O}rdinary {D}ifferential {E}quations},
	Year = {2018}}

@article{Malmquist1922,
	Author = {J. Malmquist},
	Fjournal = {Arkiv f\"{o}r matematik, astronomi och fysik},
	Journal = {Ark. {M}at. {A}str. {F}ys.},
	Pages = {1--89},
	Title = {Sur les \'{e}quations diff\'{e}rentielles du second ordre dont l'int\'{e}grale g\'{e}n\'{e}rale a ses points critiques fixes},
	Volume = {17},
	Year = {1922}
}

@article{schlesinger1912klasse,
	Author = {Schlesinger, L.},
	Title = {{\"U}ber eine {K}lasse von {D}ifferentialsystemen beliebiger {O}rdnung mit festen kritischen {P}unkten},
	journal={J. f{\"u}r Math.},
	volume={141},
	pages={96--145},
	Year = {1912}
	}

@article{BoalchKlein2004,
author = {P. Boalch},
title = {From {K}lein to {P}ainlev\'{e} via {F}ourier, {L}aplace and {J}imbo},
journal={Proc. Lond. Math. Soc.},
fjournal = {Proceedings of the London Mathematical Society},
volume = {90},
number = {1},
pages = {167--208},
year = {2005}
}

@article{yamakawa2014fourierlaplace,
      title={Fourier-{L}aplace transform and isomonodromic deformations}, 
      author={D. Yamakawa},
      year={2016},
      volume = {59},
      pages = {315--349},
      number={3},
      journal={Funkc. Ekvacioj.},
      fjournal = {Funkcialaj Ekvacioj},
}

@article{Sanguinetti_2004,
   title={The geometry of dual isomonodromic deformations},
   volume={52},
   number={1},
   journal={J. Geom. Phys.},
   fjournal={Journal of Geometry and Physics},
   author={Sanguinetti, G. and Woodhouse, N.M.J.},
   year={2004},
   pages={44--56} 
}

@article{fuchs1907lineare,
  title={{\"U}ber lineare homogene Differentialgleichungen zweiter Ordnung mit drei im Endlichen gelegenen wesentlich singul{\"a}ren Stellen},
  author={R. Fuchs},
  fjournal={Mathematische Annalen},
  journal={Math. Ann.},
  volume={63},
  number={3},
  pages={301--321},
  year={1907},
  publisher={Springer}
}

@article{JMMS,
  title={Density matrix of an impenetrable {B}ose gas and the fifth {P}ainlev{\'e} transcendent},
  author={M. Jimbo and T. Miwa and Y. Mori and M. Sato},
  journal={Phys. D: Nonlinear Phenom.},
  fjournal={Physica D: Nonlinear Phenomena},
  year={1980},
  volume={1},
  pages={80--158}
}

@article{Harnad_1994,
   title={Dual isomonodromic deformations and moment maps to loop algebras},
   volume={166},
   number={2},
   journal={Commun. Math. Phys.},
   fjournal={Communications in Mathematical Physics},
   author={Harnad, J.},
   year={1994},
   pages={337--365} 
}

@book{Katz,
 author = {N.M. Katz},
 publisher = {Princeton University Press},
 title = {Rigid Local Systems. (AM-139), Volume 139},
 urldate = {2024-06-19},
 series={Annals of Mathematics Studies},
 year = {1996}
}

@article{Hiroe_2017,
   title={Linear differential equations on the {R}iemann sphere and representations of quivers},
   volume={166},
   number={5},
   journal={Duke Math. J.},
   fjournal={Duke Mathematical Journal},
   author={Hiroe, K.},
   year={2017}
}

@misc{borot2021topological,
      title={Topological recursion for fully simple maps from ciliated maps}, 
      author={G. Borot and S. Charbonnier and E. Garcia-Failde},
      year={2021},
      eprint={2106.09002},
      archivePrefix={arXiv},
note={arXiv:2106.09002}
}

@article{Hock_2023a,
   title={A simple formula for the $x-y$ symplectic transformation in topological recursion},
   volume={194},
   journal={J. Geom. Phys.},
   fjournal={Journal of Geometry and Physics},
   author={Hock, A.},
   year={2023}
 }

@article{Hock_2023b,
   title={Laplace transform of the $x-y$ symplectic transformation formula in {T}opological {R}ecursion},
   volume={17},
   number={4},
   journal={Commun. Number Theory Phys.},
   fjournal={Communications in Number Theory and Physics},
   author={Hock, A.},
   year={2023},
   pages={821--845} 
}

@article{alexandrov2022universal,
  author    = {A. Alexandrov and B. Bychkov and P. Dunin-Barkowski and M. Kazarian and S. Shadrin},
  title     = {A universal formula for the $x-y$ swap in topological recursion},
  journal   = {J. Eur. Math. Soc.},
  year      = {2025},
  doi       = {10.4171/JEMS/1615},
}

@article{alexandrov2024symplectic,
    author = {A. Alexandrov and B. Bychkov and P. Dunin-Barkowski and M. Kazarian and S. Shadrin},
    title = "{Symplectic duality via log topological recursion}",
    doi = "10.4310/cntp.241203001416",
    journal = "Commun. Num. Theor. Phys.",
    volume = "18",
    number = "4",
    pages = "795--841",
    year = "2024"
}

@article{alexandrov2023kp,
    author = {A. Alexandrov and B. Bychkov and P. Dunin-Barkowski and M. Kazarian and S. Shadrin},
    title = "{KP integrability through the $x-y$ swap relation}",
    doi = "10.1007/s00029-025-01035-8",
    journal = "Selecta Math.",
    volume = "31",
    number = "2",
    pages = "42",
    year = "2025"
}

@article{Eynard:2002kg,
    author = {B. Eynard},
    title = {Large {$N$} expansion of the $2$ matrix model},
    journal = {JHEP},
    fjournal={Journal of High Energy Physics},
    volume = {01},
    issue = {051},
    year = {2003}
}

@article{BouchardGeneralizedTR,
  author = {Vincent Bouchard and Joel Hutchinson and Prachi Loliencar and Michael Meiers and Matthew Rupert},
  title = {A generalized topological recursion for arbitrary ramification},
  year = {2014},
  journal = {Ann. Henri Poincar\'e},
  volume = {15},
  pages = {143--169},
  doi = {10.1007/s00023-013-0233-0},
  eprint = {1208.6035},
  archivePrefix = {arXiv},
  primaryClass = {math-ph},
  note = {arXiv:1208.6035},
  url = {https://arxiv.org/abs/1208.6035}
}

\end{document}